%% file: journal_PhaseCode_Final.tex
\documentclass[journal]{IEEEtran}
\input{header}

\newcommand{\vect}[1]{\boldsymbol{#1}}
\newcommand{\mat}[1]{\boldsymbol{#1}}
\newcommand{\SET}{\mathbb{S}}
\newcommand{\CMP}{\mathbb{C}}
\newcommand{\REAL}{\mathbb{R}}
\newcommand{\BIGO}{\mathcal{O}}
\newcommand{\IMG}{{\bf i}}
\newcommand{\TSP}{^{\rm T}}
\newcommand{\HET}{^{\rm H}}
\newcommand{\DIAG}[1]{{\rm diag}(#1)}
\newcommand{\PRO}[1]{\mathbb{P}\left\{#1\right\}}
\newcommand{\PROL}[1]{\mathbb{P}\{#1\}}
\newcommand{\EXP}[1]{\mathbb{E}\left[#1\right]}
\newcommand{\EXPL}[1]{\mathbb{E}[#1]}
\newcommand{\SEP}[1]{\left\|#1\right\|_{{\it \psi}_1}}
\newcommand{\SEPL}[1]{\|#1\|_{{\it \psi}_1}}
\newcommand{\LNR}{\mathcal{A}}
\newcommand{\SGS}[1]{\left\|#1\right\|_{{\it \psi}_2}}
\newcommand{\SGSL}[1]{\|#1\|_{{\it \psi}_2}}
\newcommand{\ABS}[1]{\left|#1\right|}
\newcommand{\ABSL}[1]{|#1|}
\newcommand{\FBN}[1]{\left\|#1\right\|_F}
\newcommand{\FBNL}[1]{\|#1\|_F}
\newcommand{\OPN}[1]{\left\|#1\right\|}
\newcommand{\OPNL}[1]{\|#1\|}
\newcommand{\TWON}[1]{\left\|#1\right\|_2}

\newcommand{\ONEN}[1]{\left\|#1\right\|_1}

\newcommand{\TR}[1]{{\rm tr}\left(#1\right)}

\newcommand{\SUPP}[1]{{\rm supp}(#1)}

\begin{document}

\title{{\bf{\em{PhaseCode}}}: Fast and Efficient Compressive \\Phase Retrieval based on Sparse-Graph Codes}


\author{Ramtin Pedarsani, Dong Yin, Kangwook Lee, and Kannan Ramchandran\thanks{Ramtin Pedarsani is with the ECE Department at UC Santa Barbara. email: ramtin@ece.ucsb.edu.}\thanks{
Dong Yin and Kannan Ramchandran are with the EECS Department at UC Berkeley. email:\{dongyin,kannanr\}@eecs.berkeley.edu.}\thanks{
Kangwook Lee is with the EE Department at KASIT. email: kw1jjang@kaist.ac.kr.}\thanks{This paper was presented in part in Allerton 2015 and IEEE ISIT 2016.}}
\maketitle

\begin{abstract}
We consider the problem of recovering a complex signal $\vect{x} \in \mathbb{C}^n$ from $m$ intensity measurements of the form $\ABSL{\vect{a}_i\HET\vect{x}}, ~1 \leq i \leq m$, where $\vect{a}_i\HET$ is the $i$-th row of measurement matrix $\mat{A} \in \mathbb{C}^{m \times n}$. Our main focus is on the case where the measurement vectors are unconstrained, and where $\vect{x}$ is exactly $K$-sparse, or the so-called general compressive phase retrieval problem.  
We introduce {\bf \em{PhaseCode}}, a novel family of fast and efficient algorithms 
that are based on a sparse-graph coding framework.  
We show that in the noiseless case, the PhaseCode algorithm can recover an arbitrarily-close-to-one fraction of the $K$ non-zero signal components using only slightly more than $4K$ measurements when the support of the signal is uniformly random, with order-optimal time and memory complexity of $\Theta(K)$\footnote{
Here, we define the notation $\BIGO(\cdot)$, $\Theta(\cdot)$, and $\Omega(\cdot)$. We have $f=\BIGO(g)$ if and only if there exists a constant $C_1>0$ such that $\ABS{f/g}<C_1$; $f=\Theta(g)$ if and only if there exist two constants $C_1,C_2>0$ such that $C_1<\ABS{f/g}<C_2$; and $f=\Omega(g)$ if and only if there exists a constant $C_1>0$ such that $\ABS{f/g}>C_1$.}. It is known that the fundamental limit for the number of measurements in compressive phase retrieval problem is $4K - o(K)$ for the more difficult problem of recovering the signal exactly and with no assumptions on its support distribution \cite{Tarokh,Heinosaari}. This shows that under mild relaxation of the conditions, our algorithm is the first constructive \emph{capacity-approaching} compressive phase retrieval algorithm: in fact, our algorithm is also order-optimal in complexity and memory. Further, we show that for any signal $\vect{x}$, PhaseCode can recover a random $(1 - p)$-fraction of the non-zero components of $\vect{x}$ with high probability, where $p$ can be made arbitrarily close to zero, with sample complexity $m = c(p)K$, where $c(p)$ is a small constant depending on $p$ that can be precisely calculated, with optimal time and memory complexity.  As a result, assuming that the non-zero components of $\vect{x}$ are lower bounded by $\Theta(1)$ and upper bounded by $\Theta(K^{\gamma})$ for some positive constant $\gamma < 1$, we are able to provide a strong $\ell_1$ guarantee for the estimated signal $\hat{\vect{x}}$ as follows: $\| \hat{\vect{x}} - \vect{x}\|_1 \leq p \| \vect{x} \|_1(1 + o(1))$, where $p$ can be made arbitrarily close to zero. As one instance, the PhaseCode algorithm can provably recover, with high probability, a random $1 - 10^{-7}$ fraction of the significant signal components, using at most $m = 14K$ measurements.

Next, motivated by some important practical classes of optical systems, we consider a ``Fourier-friendly" constrained measurement setting, and show that its performance matches that of the unconstrained setting, when the signal is sparse in the Fourier domain with uniform support. {In the Fourier-friendly setting that we consider, the measurement matrix is constrained to be a cascade of Fourier matrices (corresponding to optical lenses) and diagonal matrices (corresponding to diffraction mask patterns).}  

Finally, we tackle the compressive phase retrieval problem in the presence of noise, where measurements are in the form of $y_i=\ABSL{\vect{a}_i\HET\vect{x}}^2+w_i,$
and $w_i$ is the additive noise to the $i$th measurement. We assume that the signal is quantized, and each non-zero component can take $L_m$ possible magnitudes and $L_p$ possible phases. We consider the regime where $K=\beta n^\delta$, $\delta\in(0,1)$.
We use the same architecture of {\em PhaseCode} for the noiseless case, and robustify it using two schemes: the almost-linear scheme and the sublinear scheme. We prove that with high probability, the almost-linear scheme recovers $\vect{x}$ with sample complexity
$\Theta(K \log(n))$ and computational complexity $\Theta(L_m L_p n \log(n))$, and the sublinear scheme recovers $\vect{x}$ with sample complexity $\Theta(K\log^3(n))$ and computational complexity $\Theta(L_m L_p K\log^3(n))$.

Throughout, we provide extensive simulation results that validate the practical power of our proposed algorithms for the sparse unconstrained and Fourier-friendly measurement settings, for noiseless and noisy scenarios.  
\end{abstract}

\section{Introduction}

\subsection{Phase Retrieval Problem}
Compressive sensing (CS) has recently emerged as a powerful framework for understanding the fundamental limits for signal acquisition and recovery \cite{Candes4,Donoho}.  The basic premise of CS is that a high-dimensional signal that is sparse in some basis can be recovered from linear projections of the signal with respect to an appropriate lower-dimensional measurement system.    A key attribute of CS is that the measurement system is linear and phase-preserving.  That is, the acquired samples, complex-valued in general,  contain both the magnitude and phase of the measurements.

In many applications of interest, e.g. related to optics \cite{Walther}, X-ray crystallography \cite{Milane,Harrison}, astronomy \cite{Dainty}, ptychography \cite{Rodenburg}, quantum optics \cite{Mohammad}, etc., the phase information in the measured samples is not available.   For example, in optical systems, 
one can measure only the intensity of the measurements as they relate to the photon count on a detector.  Thus, the phase of the measurements is lost.  Indeed, the problem of recovering a signal from only the magnitude of its Fourier transform has been a well-studied problem in the signal processing literature for several decades under the umbrella of phase retrieval \cite{Hayes}.  It has recently received renewed interest in the ``post-compressed-sensing'' era \cite{Baraniuk,Rangan,Sastry}, allowing for the  insights from compressive sensing to be incorporated into the phase retrieval problem when the signal of interest is sparse, and the measurement matrix is unconstrained. 

Concretely, consider a signal $\vect{x} \in \mathbb{C}^n$ and a measurement matrix $\vect{A} \in \mathbb C^{m \times n}$. The phase retrieval problem is to recover $\vect{x}$ from the observations $\vect{y} = \ABSL{\mat{A}\vect{x}}, \vect{x} \in {\mathbb C}^n$, where the magnitude is taken on each element of the vector $\mat{A}\vect{x}$. The compressive phase retrieval problem targets the case where $\vect{x}$ is $K$-sparse.

In this paper, we study the phase retrieval problem under the following settings:

\begin{itemize}
\item[(i)]    General compressive phase retrieval of sparse signals\footnote{This is easily extended, as is well known, to the case where the signal $\vect{x}$ is sparse w.r.t. some other basis, such as a wavelet, but in the interests of conceptual clarity, we will not consider such extensions in this work.}; and
\item[(ii)] ``Fourier-friendly'' compressive phase retrieval of signals having a sparse spectrum. 
\end{itemize}

We now summarize these settings:

\begin{itemize}
\item[(i)]  {\bf General compressive phase retrieval of sparse signals}:   In this setting, we are free to design the measurement matrix $A$ without any constraints, and this represents the primary contribution of this paper.  We consider it for three reasons.

   {\bf (1)} It is of broadest theoretical interest, being the most general  compressive phase retrieval problem, for which we propose a sparse-graph coding framework that is a significant departure from currently popular approaches based on convex optimization, Semi-Definite Programming (SDP), alternating minimization, gradient descent, etc. \cite{Candes1,Sanghavi,Sastry,Li,Hassibi1,Hassibi2,Mahdi}.  

 {\bf (2)} It provides the intellectual insights and the foundational framework needed to address more constrained problems, such as those studied under the Fourier-friendly setting of category {\bf (ii)}.  

{\bf (3)} It is of independent interest in applications related to certain quantum optical systems.  For example, compressive sensing has been used in recent work involving quantum optics\cite{Mohammad} to measure the transverse wavefunction of a photon, where the design of the measurement matrix has no constraints.

\item[(ii)] {\bf Fourier-friendly compressive phase retrieval of signals having a sparse spectrum}: In this category, motivated by applications related to Fourier optical systems, the measurement matrix $A$ is constrained to be Fourier-friendly (see Section \ref{sec:practical} for a detailed treatment).  Concretely, 
$\mat{A}$ is constrained to be the cascade of (up to a couple of) stages of a diagonal matrix (corresponding to a so-called optical mask or coded diffraction pattern) and a Fourier transform (corresponding to an optical lens).  This constraint is motivated by practical optical systems \cite{Popov}, array imaging \cite{Bunk}, etc., as also addressed recently by \cite{Candes3}.


\end{itemize}

\subsection{Main Contributions}
A key contribution of this work is in the introduction of modern coding theory techniques such as density evolution and sparse-graph codes \cite{RUbook} for the compressive phase retrieval problem.   Exploiting these techniques and a similar measurement system to \cite{Jaggi,Hassibi1} allows us to come up with the provably efficient and fast  PhaseCode algorithm that is order-optimal in terms of number of measurements needed, time-complexity, and memory-complexity, which are all $\mathcal{O}(K)$.   Furthermore, we provide precise constants for the number of measurements needed to achieve a targeted reliability.  
To the best of our knowledge, \emph{this is the first work that provides precise constants for the number of measurements}.  More specifically, the main contribution of this paper are the following:
\begin{itemize}
\item [(i)] For an arbitrary signal $\vect{x}$, the PhaseCode algorithm can provably recover a random fraction of at least $1 - 10^{-7}$ of the active signal components with $14K$ measurements, with optimal time and memory complexity $\Theta(K)$.  This is one instance of an entire family of trade-offs between the number of measurements needed and the fraction of non-zero signal components that can be recovered using PhaseCode. More precisely, we show that for any signal $\vect{x}$, PhaseCode can recover a random $(1 - p)$-fraction of the non-zero components of $\vect{x}$ with high probability, for arbitrarily-close-to-zero constant $p$ with sample complexity $m = c(p)K$, where $c(p)$ is a small constant depending on $p$ that can be precisely calculated. As a result, assuming that the non-zero components of $\vect{x}$ are lower bounded by $\Theta(1)$ and upper bounded by $\Theta(K^{\gamma})$ for some positive constant $\gamma < 1$, we are able to provide a strong $\ell_1$ guarantee for the estimated signal $\hat{\vect{x}}$ as follows: $\| \hat{\vect{x}} - \vect{x}\|_1 \leq p \| \vect{x} \|_1(1 + o(1))$, where $p$ can be made arbitrarily close to zero.

\item [(ii)] The PhaseCode algorithm can recover an arbitrarily-close-to-one fraction of the non-zero components of $\vect{x}$ using $4K(1+\epsilon)$ measurements for an arbitrarily small constant $\epsilon > 0$, when the support of the non-zero components of $\vect{x}$ is uniformly random, with optimal time and memory complexity of $\Theta(K)$.  It is well-known that $4K - o(K)$ measurements is the fundamental limit for unique recovery of $K$-sparse signals \cite{Tarokh,Heinosaari} for the more difficult problem of recovering the signal exactly with no assumptions on the support of the signal. This shows that under mild relaxation of the conditions, the PhaseCode algorithm is \emph{capacity-approaching}.  
 
\item [(iii)] Another key contribution of this work is to adapt the PhaseCode algorithm to a more constrained Fourier-friendly setting that is useful in certain optical systems, when $\vect{x}$ has a sparse spectrum. Specifically, we show how it is possible to elegantly integrate the Chinese-Remainder-Theorem-centric framework of Pawar and Ramchandran \cite{Sameer} (that was used to find a fast sparse Discrete-Fourier-Transform) into our PhaseCode framework without any loss of system performance in terms of measurement cost or  computational complexity.  See Section \ref{sec:practical} for details.

\item [(iv)] We demonstrate that PhaseCode can be robustified in the presence of noise.
We use the same architecture of {\em PhaseCode} for the noiseless case, and robustify it using two schemes: the almost-linear scheme and the sublinear scheme. We assume that the signal is quantized, and each non-zero component can take $L_m$ possible magnitudes and $L_p$ possible phases. We prove that with high probability, the almost-linear scheme recovers $\vect{x}$ with sample complexity
$\Theta(K \log(n))$ and computational complexity $\Theta(L_m L_p n \log(n))$, and the sublinear scheme recovers $\vect{x}$ with sample complexity $\Theta(K\log^3(n))$ and computational complexity $\Theta(L_m L_p K\log^3(n))$.

\end{itemize}

We provide pseudocode of our algorithms (in Appendix \ref{app:pseudocode}) and an extensive set of simulation results for all of the above settings that validate our theoretical findings, and verify the close match between theory and practice.  

\subsection{Related Work}

The phase retrieval problem has been studied extensively over several decades.  We do not attempt to provide a comprehensive  literature review here; instead, we highlight here only some of the pertinent and diverse approaches to  this problem that we are aware of.  A large body of literature is dedicated to the phase retrieval problem for the case where the signal to be recovered has no structure and is not sparse. ``Phaselift" proposed by Candes \emph{et al.}  \cite{Candes1} and ``PhaseCut" proposed by Waldspurger \emph{et al.}  \cite{Waldspurger} are examples of convex optimization methods to solve the problem using semi-definite programming with $\Theta(n \log(n))$ measurements. While algorithms based on SDP provide theoretical performance guarantees and are robust to noise, they suffer from a high computational complexity of $\mathcal{O}(n^3)$ rendering them unsuited for many practical applications that require $n$ to scale.\footnote{This limits the use of SDP-based methods to small to moderate values of $n$ in practice.  In contrast, we show simulations in the paper where $n$ can be very large, even as large as $10^{10}$. See Figures \ref{fig:fig1_performance} and \ref{fig:fig3_time_complexity}.}
In \cite{Sanghavi}, the authors propose an algorithm based on alternating minimization that reconstructs the signal with $\Theta(n\log(n)^3)$ measurements. In \cite{Mahdi}, the authors propose a non-convex algorithm based on Wirtinger flow that reconstructs the signal with measurement and computational complexity of $\Theta(n \log n)$. 

In \cite{Balan,Bandeira,Bodmann,Heinosaari}, several sets of authors investigate the fundamental limits of phase retrieval problem, with the goal of finding necessary or sufficient conditions on the minimum number of measurements
needed to guarantee that the solution is unique.  In summary, $4n-4$ measurements are shown to be sufficient \cite{Bodmann}, and $4n - o(n)$ measurements are necessary \cite{Heinosaari} to reconstruct any signal perfectly.  

We now review some relevant literature on compressive phase retrieval. To the best of our knowledge, the first algorithm for compressive phase retrieval was proposed by Moravec \emph{et al.} in \cite{Baraniuk}. This approach requires knowledge of the $\ell_1$ norm of the signal, making it impractical in most scenarios. The authors in \cite{Tarokh} showed that $4K-1$ measurements are theoretically sufficient to reconstruct the signal, but did not propose any low-complexity algorithm. This number was later improved to $4K -2$ in \cite{wang2014phase,akccakaya2015sparse}. The PhaseLift method is also proposed for the sparse case in \cite{Sastry} and \cite{Li}, requiring $\Theta(K^2\log(n))$ intensity measurements, and having a computational complexity of $\mathcal{O}(n^3)$, making the method less practical for large-scale applications. In \cite{bandeira2013near}, the authors propose an efficient algorithm based on polarization method that is able to stably
reconstruct any $K$-sparse vector from $\Theta(K \log(n))$ noisy intensity measurements with complexity polynomial in $n$.  The alternating minimization method in \cite{Sanghavi} can also be adapted to the sparse case with $\Theta(K^2\log(n))$ measurements and a complexity of $\mathcal{O}(K^3n \log(n))$.  Compressive phase retrieval  via generalized approximate message passing (PR-GAMP) is proposed in \cite{Rangan}, with good performance in both runtime and noise robustness shown via simulations without theoretical justification.  

A common attribute of all of the above-mentioned compressive phase retrieval references is that they assume that the measurement matrix can be designed freely.  This renders them inapplicable to many application-constrained settings
such as Fourier-optical systems.   In \cite{Candes3}, Candes \emph{et al.} consider measurement matrices that are Fourier-friendly as described in the previous subsection, but only for the non-sparse case.  They show that  PhaseLift is able to recover the signal with $\Theta(n \log(n)^4)$ measurements by using $\Theta(\log(n)^4)$ masks or coded diffraction patterns. For the sparse case,  Jaganathan \emph{et al.} consider the phase retrieval problem from Fourier measurements only \cite{Hassibi1,Hassibi2}. They propose an SDP-based algorithm, and show that the signal can be provably recovered with $\Theta(K^2 \log(n))$ Fourier measurements \cite{Hassibi1}. They also propose a combinatorial algorithm for the case where the measurement matrix can be designed
without constraints, and show that the signal can be recovered with $\Theta(K \log(n))$ measurements and time complexity of $\mathcal{O}(Kn \log(n))$ \cite{Hassibi1}. 

In the prior literature that we are aware of, the works which overlap the most in spirit with ours are (i) the recently proposed SUPER algorithm for compressive phase retrieval by Cai \emph{et al.} in \cite{Jaggi}; and (ii) the FFAST algorithm of Pawar and Ramchandran \cite{Sameer} which also features the use of coding-theoretic tools for efficiently computing  a sparse Discrete Fourier Transform.  With regard to the FFAST algorithm \cite{Sameer}, despite the common use of coding-theoretic tools, our problem formulation, analysis, and resulting algorithm are significantly different, mainly because our problem involves the loss of measurement phase, unlike that of FFAST. 

With regard to the SUPER algorithm of \cite{Jaggi}, again, while there are some similarities between the two approaches -- mainly to do with the use of certain system subcomponents such as a similar (but not identical) trigonometric-modulation method to resolve phase ambiguities, and the common use of a giant-component-cluster in the initial phase of our proposed PhaseCode algorithm (see Section \ref{sec:proof1} for details), our works are significantly distinct at many levels.  First, the SUPER algorithm targets only the general unconstrained compressive phase retrieval setting, whereas, as described earlier, we also target Fourier-friendly constrained settings that are applicable in optical systems.  Secondly, even in the unconstrained phase retrieval setting, there are significant distinctions between the two works with respect to theory, algorithm, and performance guarantees. As a quick  overview, the SUPER algorithm uses  $\Theta(K)$ measurements and features $\Theta(K\log(K))$ complexity with a zero-error-floor asymptotically.  In contrast, by trading off the zero-error-floor for an arbitrarily-small controllable error-floor, our solution features key advantages.  Specifically, this allows us to design a capacity-approaching measurement system that is based on a new and novel sparse-graph coding framework. The use of a sparse-graph coding framework in PhaseCode allows for  iterative message-passing operations between the left nodes (signal components) of the sparse-graph code and the right nodes or measurements (see Section \ref{sec:ed}).  This contrasts the more inefficient strictly ``one-way'' procedure in SUPER \cite{Jaggi} wherein measurements of different stages are processed sequentially rather than iteratively. Moreover, PhaseCode has an optimal $\Theta(K)$ decoding complexity with optimal $\Theta(K)$ memory requirements.  
We also demonstrate how PhaseCode can be robustified in the presence of noise, unlike the work of \cite{Jaggi}. We note that SUPER can also achieve $O(K)$ results with error floor. However, their approach is unable to characterize and optimize this error floor when the number of measurements is $cK$ for a specific constant $c$. Finally, we note that peeling-based algorithms and expander graphs have been used for compressive sensing \cite{XuHassibi}.

\subsection{Paper Organization}

The rest of the paper is organized as follows. 
In Section \ref{sec:formulation}, we define the general compressive phase retrieval problem. In Section \ref{sec:mainidea}, we explain the main idea of PhaseCode algorithm. We present PhaseCode algorithm in detail in Section \ref{sec:ed}. The main theoretical results of the paper are provided in Section \ref{sec:analysis}. Via extensive simulations, we evaluate PhaseCode algorithms, validating the theorem. In Section \ref{sec:practical}, we demonstrate how our proposed measurements can be adapted to a Fourier-friendly setting. In Section \ref{sec:noisy}, we show that PhaseCode can be robustified to noise. Finally, we conclude the paper in Section \ref{sec:con}.


\section{Problem Formulation and Overview of the Main Result}\label{sec:formulation}
Consider a complex signal $\vect{x} \in \mathbb{C}^n$ of length $n$ which is exactly $K$-sparse; that is, only $K$ out of $n$ components of vector $\vect{x}$ are non-zero. Let $\mat{A} \in \mathbb{C}^{m \times n}$ be the measurement matrix that needs to be designed. The phase retrieval problem is to recover the signal $\vect{x}$ from magnitude measurements $y_i = \ABSL{\vect{a}_i\HET\vect{x}}$, where $\vect{a}_i\HET$ is the $i$-th row of measurement matrix $\mat{A} \in \mathbb{C}^{m \times n}$. Figure \ref{fig:formulation} illustrates the block diagram of our problem. 

\begin{figure}[h]
\centering
    \includegraphics[width= 0.45\textwidth]{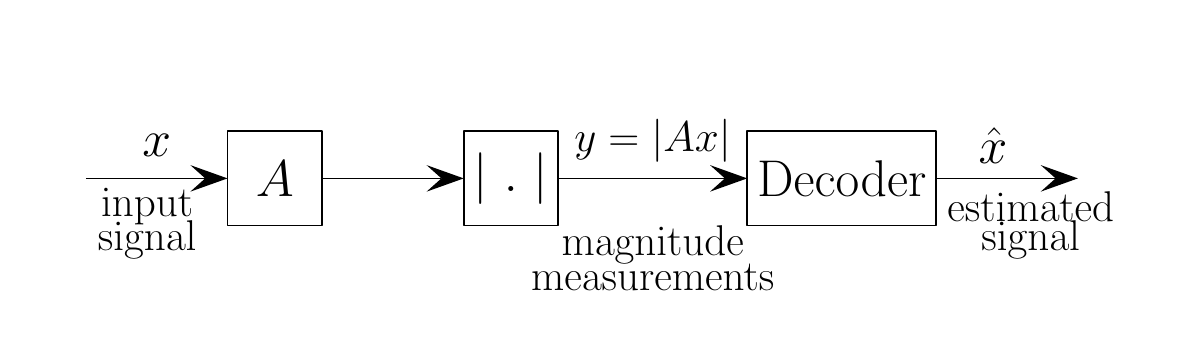}
  \caption{Block diagram of general compressive phase retrieval problem. The measurements are  $y_i = \ABSL{\vect{a}_i\HET\vect{x}}$, where $\vect{a}_i\HET$ is the $i$-th row of measurement matrix $\mat{A}$. The objectives are to design measurement matrix $\mat{A}$ and the decoding algorithm to guarantee high reliability, while having small sample complexity as well as small time and memory complexity. \label{fig:formulation}} 
\end{figure}

\begin{figure*}
\vspace{-0.2in}
\centering
\begin{subfigure}{0.4\textwidth}
\centering
\includegraphics[width= \textwidth]{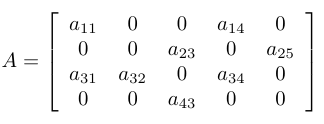}
\caption{\footnotesize{ Measurement matrix $A$.}}
\label{fig:bipartite_a}
\end{subfigure}
\qquad
\begin{subfigure}{0.2\textwidth}
\centering
\includegraphics[width= \textwidth]{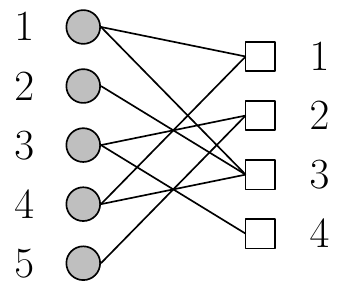}
\caption{\footnotesize {Bipartite graph $G$.}}
\label{fig:bipartite_g}
\end{subfigure}
\vspace{-0.1in}
\caption{\footnotesize{\textbf{Sparse graph codes.} The rows of $A$ (the measurements) correspond to right nodes in the bipartite graph $G$, while the columns of $A$ (the signal components) correspond to the left nodes of $G$.}}
\label{fig:ex_bipartite}
\vspace{-0.1in}
\end{figure*}

The main objectives of the general compressive phase retrieval problem is to design matrix $\mat{A}$, and the decoding algorithm to recover $\vect{x}$ such that 
\begin{itemize}
\item The number of measurements $m$ is as small as possible. Ideally, one wants $m$ to be close to the fundamental limit of $4K - o(K)$ \cite{Tarokh,Heinosaari}. 
\item The decoding algorithm is fast with low computational complexity and memory requirements. Ideally, one wants the time complexity and the memory complexity of the algorithm to be $\mathcal{O}(K)$, which is optimal. 
\item The reliability of the recovery algorithm should be maximized. Ideally, one wants the probability of failure to be vanishing as the problem parameters $K$ and $m$ get large.
\end{itemize}

\begin{remark}
In this work, we are interested in the asymptotic $K$ regime. However, even when $K$ is small, with proper modification of our algorithm, high reliability can be guaranteed when $m$ gets large. It is worth mentioning that in this case, the number of measurements will be larger than the  fundamental limit that is $4K(1 + o(1))$. We do not discuss this any further in the interest of presentation clarity.
\end{remark}

The main result of our paper is stated in the following (informal) theorem. 

\begin{theorem}
Consider a $K$-sparse signal $\vect{x} \in \mathbb{C}^n$, and the measurement matrix $\mat{A} \in \mathbb{C}^{m \times n}$ chosen by the PhaseCode algorithm. 
\begin{itemize}
\item [(i)] PhaseCode can recover a random $(1 - p)$-fraction of the non-zero components of $\vect{x}$ with high probability, for arbitrarily-close-to-zero constant $p$. The measurement complexity of the algorithm is $m = c(p)K$, where $c(p)$ is a small constant depending on $p$ that can be precisely calculated. The time and memory complexity of PhaseCode are also $\Theta(K)$. Further, for the estimated signal $\hat{\vect{x}}$, assuming that the non-zero components of $\vect{x}$ are lower bounded by $\Theta(1)$ and upper bounded by $\Theta(K^{\gamma})$ for some positive constant $\gamma < 1$, we have 
$$\| \hat{\vect{x}} - \vect{x}\|_1 \leq p \| \vect{x} \|_1 \left (1 + \Theta(\frac{1}{\log(K)}) \right ).$$ 
\item [(ii)] Assuming that the support of $\vect{x}$ is distributed uniformly at random, with high probability, PhaseCode can recover an arbitrarily-close-to-one fraction of the non-zero components with $m = 4K(1+\epsilon)$ measurements for arbitrarily small constant $\epsilon > 0$. 
\end{itemize}

These results are more precisely stated in Theorems \ref{thm:main} and \ref{thm:2} in Section \ref{sec:analysis}. See Table \ref{tab:1} for some selected values of $p$ and $m$.

\end{theorem}

\section{Main Idea of the PhaseCode Algorithm}\label{sec:mainidea}
We now describe the main idea behind PhaseCode.  As mentioned, the main novelty of our work is that we use sparse-graph codes, and the powerful tools of modern coding theory for design and analysis.

The design of an appropriate measurement matrix $\mat{A}$ for the compressive phase retrieval problem is equivalent to the design of an appropriate bipartite graph $G$, as for each measurement matrix, there exists a corresponding bipartite graph. 
Specifically, the rows of $\mat{A}$ (the measurements) are the right nodes in the bipartite graph $G$, while the columns of $\mat{A}$ (the signal components) are left nodes of $G$. We call the left nodes of $G$ that correspond to an active (non-zero) signal component as active left nodes. Left node $i$ is connected to right node $j$ if $a_{ji}$ is non-zero. The example shown in Figure \ref{fig:ex_bipartite} illustrates this connection.


%
%
%
%
As
is well-known and also intuitive, in the phase-retrieval problem, the signal of interest can be recovered
only to within an unknown global phase. The idea of our iterative reconstruction algorithm is to detect a non-zero signal component,
give it global zero-phase, and align all other signal components with respect to it. This suggests
the intuition of building up one or more clusters of non-zero components, where in our terminology,
these clusters are identified by their colors; i.e. all the non-zero components belonging to a particular cluster have the same
color. Two (or more) non-zero components (active left nodes) can be colored with the same color if their components are known in location, magnitude and phase relative to each other.

Our goal in designing the measurement matrix of the sparse graph is to create \emph{iteratively decodable} right nodes (set of appropriately designed measurements). The key property of a right node that is conducive to our desired coloring operation is as follows.  
\emph{If a right node  is connected to one or more known components (colored active left nodes with the same color) and exactly one uncolored active left node (unresolved active signal component), then that component can be resolved, i.e. the uncolored active left node will be colored with the same color. } See Figure \ref{fig:multiton}.
%
%

\begin{figure}
\centering
    \includegraphics[width= 0.35\textwidth]{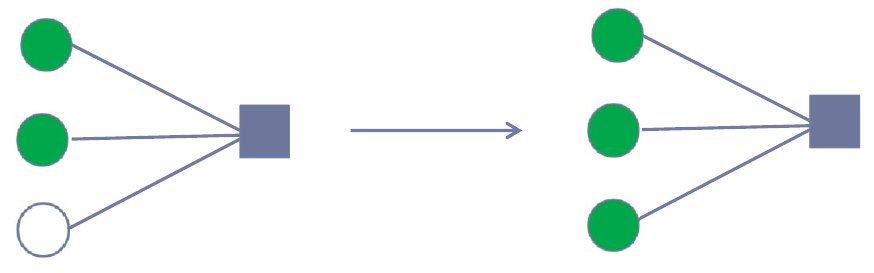}
  \caption{\textbf{Coloring operation.} The figure illustrates when a right node is connected to exactly one uncolored active left node, and the other active left nodes connected to the right node are colored with the same color, then the uncolored active left node is colored with that color. In the graph, we have shown only the active left nodes.\label{fig:multiton}} 
\end{figure}

Our idea is to make this coloring ``primitive operation" {\em iteratively trigger}
more such coloring primitive operations in the system.  
Of course, the key is to design the graph efficiently to ensure that the domino-effect will continue till 
{\em all} the active left nodes are colored, while minimizing the number of right nodes needed to accomplish this (measurement cost). 

This is the high-level connection between the compressive phase retrieval problem and sparse-graph code design. Our recovery process is conceptually similar to the ``peeling" decoding of packets based on Low-Density-Parity-Check (LDPC) codes in packet-erasure communication systems, with the key distinction that {\em we cannot measure phase}.   
This makes our problem more challenging, therefore requiring a different analysis of the density evolution in the graph, as we will describe.  
But at a high level, our coloring primitive operation plays the analogous role of peeling in LDPC decoding. 

Of course, a natural question is how our measurement system detects if a right node is indeed connected to one or more colored active left node and exactly one uncolored active left node.  We can do so with a set of 4 cleverly designed ``trigonometric'' measurements that are part of each right node. 
We will explain the trigonometric measurements in detail in Section \ref{sec:measurement}. 

\section{PhaseCode Algorithm}\label{sec:ed}
First we define $\mat{A} \in \mathbb{C}^{4M \times n}$ to be a ``row tensor product"\footnote{Here, we apologize for not following popular convention for the notation for tensor product of matrices; instead, we define our own notation that is convenient for our purpose, which should hopefully not cause any confusion.} of matrices $\mat{T}$ and $\mat{H}$, where $\mat{H} \in \{0,1\}^{M \times n}$ is a binary ``code" matrix, to be shortly explained, and $\mat{T} \in 4 \times n$ is the ``trigonometric modulation'' matrix that provides $4$ measurements per each row of $\mat{H}$. We define a row tensor product of matrices $\mat{T}$ and $\mat{H}$, $\mat{T} \otimes \mat{H}$, as follows. Let $\mat{A} = \mat{T} \otimes \mat{H} = [\mat{A}_1\HET, \mat{A}_2\HET, \ldots, \mat{A}_M\HET]\HET$ and $\mat{A}_i \in \mathbb{C}^{4 \times n}$. Then, $A_i(jk) = T_{jk} H_{ik}, ~ 1\leq j \leq 4, ~ 1 \leq k \leq n$. 

\begin{example}
Consider matrices 
$$
\mat{T} = \left [ \begin{array}{ccc}
0.1 & 0.2 & 0.3 \\
0.4 & 0.5 & 0.6 
\end{array} \right] 
 ~ \text{and} ~ 
\mat{H} = \left [ \begin{array}{ccc}
0 & 1 & 0 \\
1 & 1 & 0 \\
0 & 0 & 1
\end{array} \right]. 
$$
Then, our measurement matrix $\mat{A}$ is designed from:
$$
\mat{A} = \mat{T} \otimes \mat{H} = \left [ \begin{array}{ccc}
0 & 0.2 & 0 \\
0 & 0.5 & 0 \\
0.1 & 0.2 & 0 \\
0.4 & 0.5 & 0 \\
0 & 0 & 0.3 \\
0 & 0 & 0.6
\end{array} \right].
$$
\end{example}
Matrix $\mat{H}$ is constructed using a carefully chosen random bipartite graph model with $n$ left nodes and $m$ right nodes. Each left node refers to a component of $x$, and each right node refers to a set of 4 measurements. There are $K$ active left nodes corresponding to the $K$ non-zero components of $x$. The bipartite graph is constructed as follows. $H_{ij} = 1$ if and only if left node $j$ is connected to right node $i$, and $H_{ij} = 0$ otherwise. 

While we provide the details of how to design matrix $\mat{T}$ in Section \ref{sec:measurement}, for completeness of the description, we state it precisely here deferring explanation to Section \ref{sec:measurement}. Let $\om '$ be a uniformly random phase between $0$ and $2 \pi$. We design $\mat{T} \in \mathbb C^{4 \times n}$ to be 
\begin{equation}\label{eq:G}
T = \left( \begin{array}{cccc}
e^{\bi \om} & e^{\bi 2 \om} & \ldots & e^{\bi n \om} \\
e^{-\bi \om} & e^{-\bi 2 \om} & \ldots & e^{-\bi n \om} \\
\cos(\om) & \cos(2\om) & \ldots & \cos(n \om) \\
e^{\bi \om '} & e^{\bi 2 \om '} & \ldots & e^{\bi n \om '}
\end{array}
\right).
\end{equation}

\begin{table}
\centering
\begin{tabular}{ |c |  c| }
\hline
Notation & Description \\
\hline
$\vect{x}$ & complex signal of length $n$ \\
\hline
$K$ & sparsity of the signal \\
\hline
$n$ & length of the signal \\
\hline
$m$ & number of measurements \\
\hline
$M$ & number of the rows of the code matrix \\
\hline
$\mat{A}$ & measurement matrix \\
\hline 
$\mat{H}$ & code matrix \\
\hline 
$\mat{T}$ & modulation matrix \\
\hline
\end{tabular}
\caption{Table of Notation.}
\label{tab:notation}
\end{table}


As in \cite{Sameer}, in the bipartite graph model, we use the following terminology extensively throughout the paper:
\begin{itemize}
\item \emph{Singleton:} A right node is a singleton if it is connected to exactly one \emph{active} left node.  
\item \emph{Doubleton:} A right node is a doubleton if it is connected to exactly two active left nodes. 
\item \emph{Multiton:} A right node is a multiton if it is connected to more than one active left node.\footnote{In our terminology, a doubleton is also a multiton.} 
\end{itemize}

We now describe PhaseCode algorithm, and analyze it in Section \ref{sec:analysis}. 
With the aid of the carefully designed matrix $\mat{T}$, our decoder is capable of performing the following functions:

\begin{itemize}
\item When an active left node is connected to a singleton right node, the active left node can be colored with a new color. That is, the non-zero component can be found in magnitude and location. However, the relative phase of the component with respect to other resolved components cannot be recovered. Figure \ref{fig:singleton} illustrates this operation. 

Note that in our terminology, each color refers to a local coordinate with a local phase, for example, the red coordinate, blue coordinate, etc. Then, the relative phase of two non-zero components that are colored as red is known. However, the relative phase of a blue component and a red component is not known. 

\item When a right node is connected to exactly one \emph{uncolored} active left node, and the other non-empty set of active left nodes connted to the right node have all the same color (let's say green), then the uncolored active left node is colored with that color (i.e. it becomes green). Figure \ref{fig:multiton} illustrates this operation. 

\item When \emph{all} the active left nodes connected to a right node are colored, with exactly two colors, then those two colors can be combined into a single composite color. Figure \ref{fig:merge} illustrates this operation.\footnote{We use this operation only in the second iteration of PhaseCode.} 
\end{itemize}

\begin{figure}
\centering
    \includegraphics[width= 0.35\textwidth]{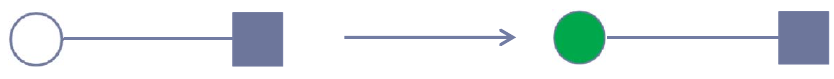}
  \caption{\textbf{Singleton coloring operation.} The figure illustrates when a right node is a singleton, the corresponding active left node gets colored with a new color. In the graph, we have not showed the left nodes corresponding to 0 signal components.\label{fig:singleton}} 
\end{figure}

\begin{figure}
\centering
    \includegraphics[width= 0.35\textwidth]{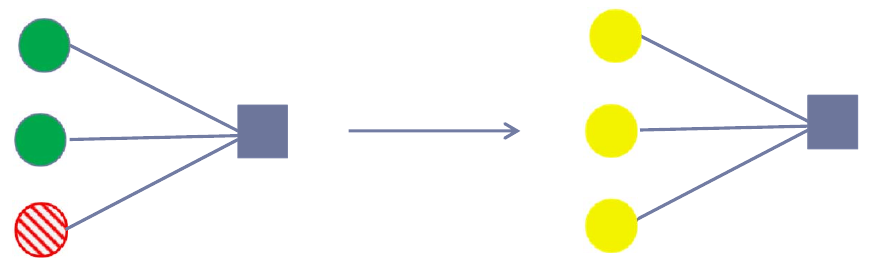}
  \caption{\textbf{Combining colors.} The figure illustrates when a right node is connected to only colored active left nodes with two colors, then the colors can be combined. \label{fig:merge}} 
\end{figure}

\begin{uni}
In the first iteration of the algorithm, all the left nodes connected to singletons are colored. In the second iteration, all the doubletons that are connected to two colored active left nodes from the first iteration (strong doubletons), are detected, and their colors are combined. Then, the \emph{largest} set of active left nodes having the same color\footnote{Whenever two active left nodes having colors $C_1$ and $C_2$ are combined, they get the same composite color $C_{12}$, and all other active left nodes with colors $C_1$ and $C_2$ are also recolored to $C_{12}$.} is selected, and \emph{every other colored active left node gets uncolored}. 
At this point, there is only \emph{one} color and no new colors are added to the system. 
In the following iterations, if a right node is connected to exactly one uncolored active left node and at least one colored active left node, then that uncolored active left node gets colored. (See Figure \ref{fig:multiton}.) The algorithm continues until no more active left nodes can be colored.
\end{uni}
We provide the pseudocode of the algorithm in Appendix \ref{app:pseudocode}.

\begin{remark}
PhaseCode has $\Theta(K)$ time and memory complexity. 
\end{remark}
%

\begin{example}
Let $K = 4$, $M = 5$ and $d = 2$. Without loss of generality, label the active left nodes by $1$ to $4$. Suppose that the bipartite graph is such that the right nodes are connected to $\{1\}$, $\{1,2\}$, $\{3\}$, $\{1,3\}$, and $\{2,3,4\}$. In the first iteration,  $1$ and $3$ are colored, let us say by red and blue, respectively since these active left nodes are connected to singletons. 
In the second iteration, 
PhaseCode finds a strong doubleton, $\{1,3\}$, that is connected to colored left nodes $1$ and $3$. Thus, their colors are combined to a composite color, let us say green, which will be the only color of the system after this iteration. In the third iteration, left node $2$ is colored through the right node $\{1,2\}$, since $2$ is the only uncolored left node connected to this right node. Finally, in the forth iteration, left node $4$ is colored through right node $\{2,3,4\}$. This completes the successful decoding of PhaseCode algorithm.  
\end{example}

\subsection{Measurement Design: ``Trig-Modulation''}\label{sec:measurement}
In this section, we will explain the choice of the measurement matrix $\mat{T}$. Our design of $\mat{T}$ draws heavily from the proposed trigonometric subsystem in \cite{Jaggi} with proper modifications to better match our sparse-graph code subsystem, $\mat{H}$, that is distinct from \cite{Jaggi}. We also show that one can decrease the number of these trig-based measurements from $5$ per right node as proposed in \cite{Jaggi} to $4$ per right node as we describe that is crucial in designing a capacity-approaching scheme. 

Define the length-$4$ vector $\vect{y}_i$ to be the measurement vector corresponding to the $i$-th row of matrix $\mat{H}$ for $1 \leq i \leq M$. Then $\vect{y} = [\vect{y}_1^T, \vect{y}_2^T, \ldots, \vect{y}_M^T]^T$, where $\vect{y}_i = [y_{i,1},y_{i,2},y_{i,3},y_{i,4}]^T$. Let $\om = \frac{\pi}{2n}$. We design the measurement matrix $\mat{T} = [t_{j\ell}]$ as follows. For all $\ell, ~ 1 \leq \ell \leq n$,
\begin{align}
t_{1\ell} &= e^{\bi \om \ell}, \\
t_{2\ell} &= e^{-\bi \om \ell},  \\
t_{3\ell} &= 2 \cos(\om \ell), \\
t_{4\ell} &= e^{\bi \om' \ell},
\end{align}
where as mentioned in Section \ref{sec:ed}, $\om '$ is a random phase uniformly distributed between $0$ and $2 \pi$.

As mentioned in Section \ref{sec:ed}, the measurement matrix should enable us to do the following operations: (1) Detect whether we have a singleton right node, and if yes, what the location index and magnitude of the corresponding active left node are (See Figure \ref{fig:singleton}); (2) detect if a multiton right node is connected to colored active left nodes having exactly two unique colors, and if yes, what the relative phase of the colored components is. We call these as mergeable multitons (See Figure \ref{fig:merge});  (3) detect if a multiton right node is connected to colored active left nodes with the same color and only one uncolored active left node, the measurement system should be able to find the index, magnitude, and relative phase of the uncolored active left node. We call these right nodes resolvable multitons as in \cite{Jaggi} (See Figure \ref{fig:multiton}). In the following, we show how each of these detections can be accomplished using ``guess and check" approach. We provide pseudocode of these detection procedures in Appendix \ref{app:pseudocode}.

\begin{itemize}
\item [(i)] {\bf Singletons}: Suppose that we want to check the hypothesis that the $i$-th right node is a singleton. If the right node is a singleton, only one non-zero component of $\vect{x}$, let's say $x_\ell$, is present in vector $y_i$, that is $y_{i,1} = |x_\ell e^{\bi \om \ell}|$, $y_{i,2} = |x_\ell e^{-\bi \om \ell}|$, and so on. Thus, the $i$-th right node is a singleton only if $y_{i,1} = y_{i,2} =  y_{i,4}$. The event that $i$ is not a singleton, and all these measurements are equal has measure $0$ since $\om'$ is a uniformly random phase.\footnote{In practice, every measurement system has a finite precision level. Moreover, practical systems suffer from the presence of noise. The measurement system introduced here is clearly not robust to noise and finite precision of the measurement matrix, but we will show in Section \ref{sec:noisy} that PhaseCode can be robustified to noise while maintaining its iterative decoding architecture.} In order to find the index $\ell$, one uses $y_{i,3}$ to get
$$
\ell = \frac{1}{\om}\cos^{-1}\left(\cos(\om \ell)\right) = \frac{1}{\om}\cos^{-1}\left (\frac{y_{i,3}}{2y_{i,1}} \right).
$$
Note that $\cos(\om \ell)$ is positive if $0 \leq \om \leq \frac{\pi}{2n}$ for all $\ell, ~1 \leq \ell \leq n$.

\item [(ii)] {\bf Mergeable multitons}: Consider a right node $i$ as in Figure \ref{fig:merge}, which is already known to be connected to some (say, red) active left nodes (non-empty set $\mathcal{R}$) and some (say, blue) active left nodes (non-empty set $\mathcal{B}$). This means that the red (or blue) signal components are known in location, magnitude, and phase relative to each other. However, the relative phase of blue and red components' coordinate systems is not known. If there is no other active left node connected to $i$, we show that the relative phase can be found. Thus, the colors can be combined. (We again deploy a guess and check strategy.) First, we guess that right node $i$ is connected to no other active left nodes. Then, we have access to measurement 
$$
y_{i,1} = |r + b|,
$$
where $r = \sum_{\ell \in \mathcal{R}} x_j e^{\bi \om \ell}$ is the sum of complex numbers corresponding to the red components, and $b = \sum_{\ell \in \mathcal{B}} x_\ell e^{\bi \om \ell}$ is the sum of complex numbers corresponding to the blue components. Since red components are known up to a local phase, $|r|$ is known. Similarly, $|b|$ is also known. Without loss of generality, pick some $\ell_r \in \mathcal{R}$ and set the phase of $x_{\ell_r}$ to $0$ to form the local coordinate for red components. Furthermore, pick some $\ell_b \in \mathcal{B}$ and set the phase of $x_{\ell_b}$ to $0$ to form the local coordinate for blue components. Given the local coordinates, $r = |r| e^{\bi \phi_r}$ and $b = |b| e^{\bi \phi_b}$ are known. By the cosine law, the true relative phase between $r$ and $b$ can be found as
\begin{equation}\label{eq:cosine}
\theta = \cos^{-1}\left (\frac{|r|^2 + |b|^2 - y_{i,1}^2}{2 |r| |b|} \right),
\end{equation}
up to a plus-minus sign. Assuming that the plus sign is true, we can merge these components as follows. Without loss of generality, we set the phase of $x_{\ell_r}$ to $0$. Thus, $r = |r| e^{\bi \phi_r}$ and $b = |b|e^{\bi (\phi_r + \theta)}$. This shows that the local coordinate in $\mathcal{B}$ should be rotated by an angle $\theta + \phi_r - \phi_b$ to match with the new coordinate. Hence, we recover all the blue components with respect to the coordinate of red components, and the colors can be combined. A similar procedure can be done for the solution of $\theta$ with a minus sign. Now we again use the check equation to find whether one of these relative phases passes the check equation. If none of them passes, our guess is wrong, and right node $i$ is not a mergeable multiton. Thus, we need to check whether
$$
|\sum_{\ell \in \mathcal{R} \cup \mathcal{B}} x_\ell e^{\bi \om' \ell}| = y_{i,4} 
$$
is satisfied or not for the $2$ values of $\theta$ derived in \eqref{eq:cosine}. If the guess is correct, the probability that the check fails is $0$ since $\om'$ is random. Moreover, if the guess is not correct, the probability that the check passes is $0$. 

\item [(iii)] {\bf Resolvable multitons}: Consider a right node $i$, for which we know that it is connected to some known active left nodes that have the same color. We want to check if $i$ is connected to exactly one other active left node; i.e. one unknown non-zero component of $x$, say $x_\ell$, as in Figure \ref{fig:multiton}. We now describe our guess and check strategy to check if right node $i$ is indeed a resolvable multiton, and if so, to find $\ell$ and $x_\ell$. If our guess is correct, we have access to measurements of the form:
\begin{align}\label{eq1}
y_{i,1} &= |a + e^{\bi \om \ell} x_\ell| = |u|,\\ \label{eq2}
y_{i,2} &= |b + e^{-\bi \om \ell} x_\ell| = |v|, \\ \label{eq3}
y_{i,3} &= |c + 2\cos(\om \ell) x_\ell| = |w|, \\ \label{eq4}
y_{i,4} &= |d + e^{\bi \om' \ell} x_\ell|,
\end{align}
where complex numbers $a$, $b$, $c$ and $d$ are known values that depend on the values and locations of the known colored active left nodes. For the purpose of readability, we show the calculations of how to solve the system of equations \eqref{eq1}-\eqref{eq4} in Appendix \ref{app:quad}.

\end{itemize}

\section{Main Result}\label{sec:analysis}
In this section, we analyze the performance of PhaseCode and provide the main theoretical results of this paper. 

\begin{table*}[t] 
\centering
\begin{tabular}{ c |  c  c  c  c c c}
$d$ & $5$ & $6$  &$7$   &$8$ & $9$ & $10$ \\
\hline
$m(p)$  &  $12.44K$ &$12.72K$  &$13.28K$   &$\mathbf{13.92K}$ & $14.64K$ & $15.4K$ \\
\hline
$p$ &   $1.1 \times 10^{-3}$ & $8 \times 10^{-5}$&  $3.2 \times 10^{-6}$ &  $\mathbf{1 \times 10^{-7}}$   & $2.9 \times 10^{-9}$ & $7 \times 10^{-11}$
\end{tabular}
\caption{Family of trade-offs between error floor and number of measurements for Phasecode. The table shows that to achieve higher reliability, i.e. smaller error floor, the number of measurements $m$ should be increased.}
\label{tab:1}
\end{table*}

\subsection{Bipartite Graph Construction}\label{sec:graph}
As mentioned earlier, we design our code matrix based on a random bipartite graph model. Given a bipartite graph with $n$ left nodes and $M$ right nodes, define the pruned bipartite graph corresponding to $\vect{x}$ to be a bipartite graph with $K$ left nodes corresponding to the non-zero components of $\vect{x}$ and $M$ right nodes, such that all the left nodes corresponding to the zero components and their connected edges are deleted. From now on, we consider the pruned graph for analysis. Moreover, from now on, by a left node (of the pruned graph), we refer to an active left node.

We first define the left and right edge degree distribution of the random bipartite graph.
Define $\rho_i$ to be the probability that a randomly selected edge is connected to a right node of degree $i$, and $\lambda_i$ to be the probability that a randomly selected edge is connected to a left node of degree $i$. Define the edge degree distributions or edge degree polynomials of right and left nodes as follows.
\begin{align*}
\rho(x) &= \sum_{i \geq 1} \rho_i x^{i-1}; \\
\lambda(x) &= \sum_{i \geq 1} \lambda_i x^{i-1}.
\end{align*}

We construct two random bipartite graph models as follows:
\begin{itemize}
\item[(i)] Regular left degree: In this construction, each left node is connected to $d$ right nodes randomly, where $d$ is a constant to be chosen. Thus, the degree of all left nodes are $d$. More formally, let $\mathcal{C}^K(d,M)$ be the ensemble of regular left degree bipartite graphs with $K$ left nodes, $M$ right nodes, and left degree $d$. We pick a bipartite graph uniformly at random from this ensemble. When $M$ and $K$ get large, the degree of a random right node is Poisson distributed with parameter $\eta = \frac{Kd}{M}$. Note that the degree of a right node in the pruned graph is the number of active left nodes connected to it. Since $\rho_i$ is the fraction of edges that are connected to a right node of degree $i$, we have
\begin{align*}
\rho_i &= \frac{iM}{Kd} \PP(\text{random right node has degree}~i)  \\
 &= \frac{i}{\eta} \frac{\eta^i e^{-\eta}}{i!} \\
 &= \frac{\eta^{i-1}e^{-\eta}}{(i-1)!}.
\end{align*}
Then, the left edge and right edge degree distributions are
\begin{align}
\lambda(x) &= x^{d-1} \\ \label{eq:edge}
\rho(x) &= e^{-\eta(1-x)}.
\end{align}
\item[(ii)] Irregular left degree: In this construction, we design the left degree distribution $\lambda(x)$ based on a truncated harmonic distribution as follows. 
Let $h(x) = \sum_{i=1}^x 1/i$. Then,
\begin{align}\label{eq:irr}
\lambda_i = \frac{1}{i-1} \times \frac{1}{h(D-1)}, ~ 2 \leq i \leq D,
\end{align}
where $D$ is a (large) constant to be determined. The harmonic distribution for irregular LDPC codes is well-known to be capacity-achieving for BEC channels \cite{Luby2}. 
\end{itemize}
The main theoretical results of this paper for the noiseless case are as follows.
 
\begin{theorem}\label{thm:main}
Let $\mat{A} = \mat{T} \otimes \mat{H}$ be the measurement matrix, where $\mat{H}$ is chosen uniformly at random from the ensemble $\mathcal{C}^n(d,M)$ and $\mat{T}$ is the modulation matrix defined in \eqref{eq:G}. Using the $m$ measurements $\vect{y}=|\mat{A}\vect{x}|$, for any $p > 0$, Regular PhaseCode can recover at least a $1-p$ fraction of the non-zero components of $x$ chosen uniformly at random, where $m = c(p) K$ and tabulated in Table \ref{tab:1} for selected values. 
As a particular operating point, Regular PhaseCode is able to recover a random fraction $1 - 10^{-7}$ of non-zero components of $\vect{x}$ with $14K$ measurements with probability  $1 - \mathcal{O}(1/m)$.  Furthermore, the decoding complexity of the algorithm is $\Theta(K)$ which is order-optimal. 
\end{theorem}

\begin{theorem}\label{thm:2}
Let $\mat{A} = \mat{T} \otimes \mat{H}$ be the measurement matrix, where $\mat{H}$ is chosen according to the irregular construction in \eqref{eq:irr}, and $\mat{T}$ is the modulation matrix defined in \eqref{eq:G}. Under the assumption that the support of the sparse signal is uniformly random, using $m = 4K(1+ \epsilon)$ measurements $\vect{y}=|\mat{A}\vect{x}|$ for arbitrarily small $\epsilon > 0$, Irregular PhaseCode is able to recover all but an arbitrarily small random fraction of the non-zero components of $\vect{x}$ with probability $1 - \mathcal{O}(1/m)$. Furthermore, the decoding complexity of the algorithm is $\Theta(K)$ which is order-optimal. 
\end{theorem}
We provide the proofs in Sections \ref{sec:proof1} and \ref{sec:proof2}.
\begin{corollary}\label{cor:l1}
Suppose that for a particular choice of parameters, PhaseCode has error floor $p$. For any signal $\vect{x} \in \mathbb{C}^n$, assuming that the non-zero components of $\vect{x}$ are lower bounded by $\Theta(1)$ and upper bounded by $\Theta(K^{\gamma})$ for some positive constant $\gamma < 1$, we have 
$$
\| \hat{\vect{x}} - \vect{x} \|_1 \leq p\| \vect{x} \|_1(1 + \Theta(\frac{1}{\log(K)})),
$$ 
with probability $1 - \mathcal{O}(K^{\frac{1+\gamma}{2}}e^{-\frac{2K^{(1-\gamma)/2}}{\log^2(K)}})$ over the randomized choice of $\mat{A}$.
\end{corollary}
We provide the proof of Corollary  \ref{cor:l1} in Appendix \ref{app:l1}.

\begin{figure}
    \centering
    \includegraphics[width=0.4\textwidth]{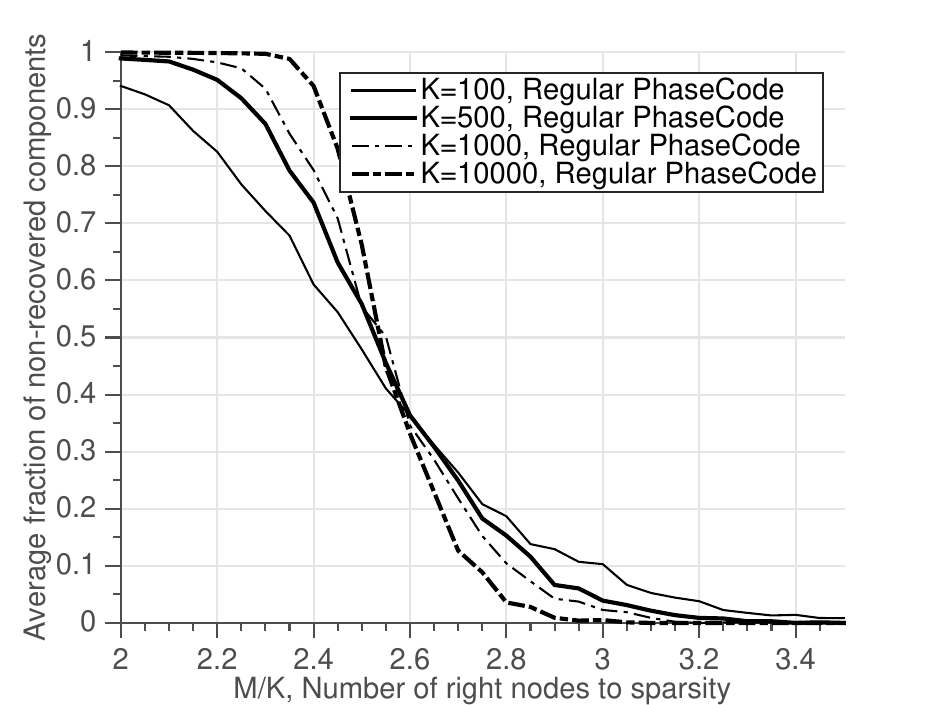}
     \caption{\textbf{Performance of regular PhaseCode Algorithm.} We evaluate Regular PhaseCode algorithm via simulations. We chose the $3$rd column of the table as an operating point, i.e., $(d, m, p^*(m)) = (7, 13.28K, 3.2\times 10^{-6})$. PhaseCode algorithm successfully recovers almost all active signal components with high probability when $m = 4 \times 3.32 K = 13.28K$.} 
    \label{fig:fig1_performance}
\end{figure}
\begin{figure}
    \centering
    \includegraphics[width=0.4\textwidth]{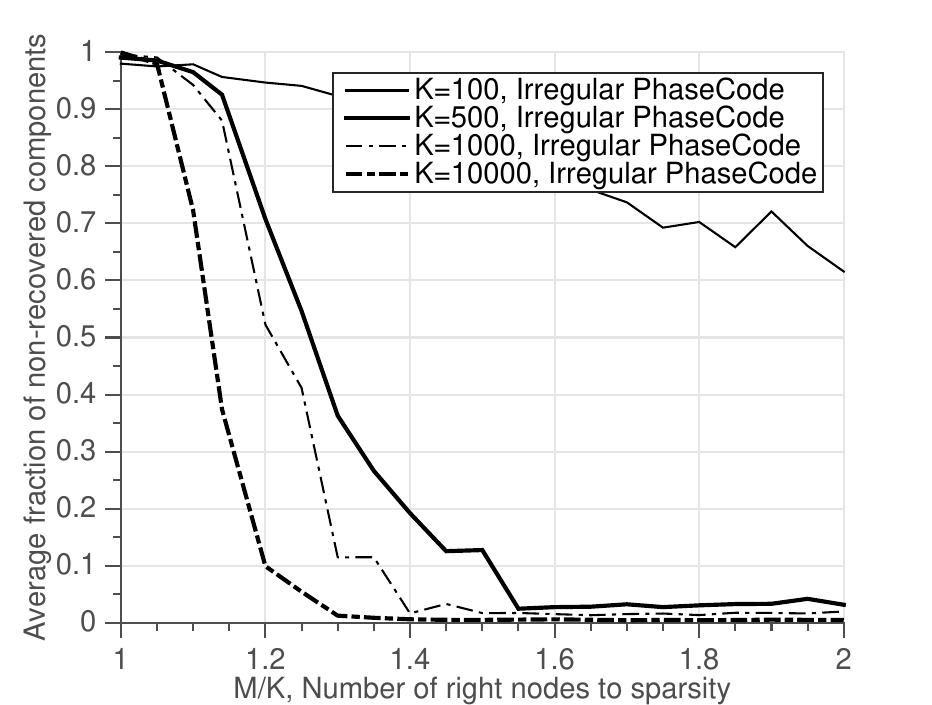}
     \caption{\textbf{Performance of PhaseCode Algorithm.} We also evaluate Irregular PhaseCode, which demonstrates that it is capacity-approaching. We observe that for $K=10000$ irregular PhaseCode can recover almost all the non-zero signal components with $m = 4 \times 1.3K$ measurements.} 
    \label{fig:irregular}
\end{figure}
\begin{figure}
    \centering
    \includegraphics[width=0.4\textwidth]{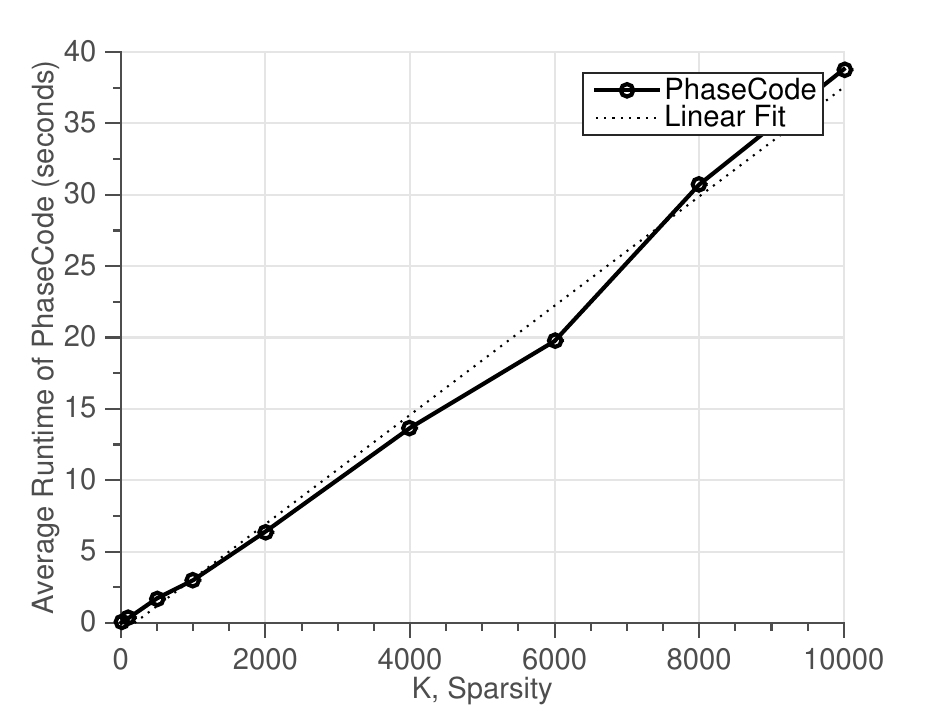}
    \caption{\textbf{Time Complexity of PhaseCode.} We measure run-time of PhaseCode algorithm. We choose $n=10^{10}$ and vary $K$. 
    }
    \label{fig:fig3_time_complexity}
\end{figure}

Before presenting the proof of the main theorems, we illustrate the performance of regular and irregular PhaseCode via simulations in Figures \ref{fig:fig1_performance} and \ref{fig:irregular}. Theorem \ref{thm:main} guarantees that regular PhaseCode recovers a fraction $p^*(m)$ of $\vect{x}$ with $m$ measurements with high probability, where $(d, m, p^*(m))$ can be chosen from Table \ref{tab:1}. We choose the $3$-rd column of the table as an operating point, i.e., $(d, m, p^*(m)) = (7, 13.28K, 3.2\times 10^{-6})$ for regular PhaseCode. 
We define the error probability to be the fraction of non-zero components of $\vect{x}$ that are not recovered.
We measure the error probability while $m$ is varied between $8K$ and $14K$ by averaging over $1000$ simulation runs. 
We repeat the same procedure for several values of $K$.
As expected, the PhaseCode algorithm successfully recovers essentially all the signal components when $m = 13.28K$.
We also show simulation results for irregular PhaseCode in Fig. \ref{fig:irregular} that support Theorem \ref{thm:2}. For example, when $K = 10000$, the coloring algorithm successfully recovers the signal with only $4 \times 1.3K = 5.2K$ measurements. From the simulations, it is clear that to operate close to capacity, one needs large asymptotics for $K$.

Theorems \ref{thm:main} and \ref{thm:2} also state that the decoding complexity of PhaseCode is $\Theta(K)$, which is order-optimal. In addition to that, its memory complexity is $\Theta(K)$, which is also order-optimal. In order to corroborate the claims, we measure the running time of the PhaseCode Algorithm. We choose the same operating point for regular PhaseCode as in the above simulations. 
We randomly generate signals of length $n=10^{10}$, and increase the sparsity $K$ up to $10^4$ to see how the average runtime scales. The results are plotted in Figure \ref{fig:fig3_time_complexity}; as $K$ increases, the measured decoding time linearly increases. Indeed, PhaseCode successfully recovers $K=10^4$ non-zero components in less then $40$ seconds. The exact runtime can be further improved considering that the simulator is written in Python and is not fully optimized, and that the simulation is done on a normal laptop.\footnote{For the measurements, we used a laptop with 2GHz Intel Core i7 and 8GB memory.}



\subsection{Proof of Theorem \ref{thm:main}}\label{sec:proof1}
We first provide a brief outline of the proof elements, highlighting the main technical components needed to show that PhaseCode recovers an arbitrarily-close-to-one fraction of non-zero signal components with high probability. 
\begin{itemize}
\item \emph{Density evolution:} We analyze the performance of PhaseCode on a typical random bipartite graph (regular or irregular), for a fixed number of iterations, $\ell$. First, we assume that a local neighborhood of depth $2 \ell$ of
every edge in the graph is tree-like, i.e., cycle-free. Under this assumption, all
the messages between right and left nodes, in the first $j$ iterations of the algorithm, are independent. Using this independence assumption, we derive a recursive equation that represents the evolution of the expected number of unresolved components at each iteration.
\item \emph{Convergence to the cycle-free case:} : Using a Doob martingale argument as in \cite{RU01}, we show that the $2 \ell$ neighborhood of most of the edges of a randomly chosen graph from the ensemble is cycle-free with high probability. This proves that PhaseCode decodes all but a small fraction of the left nodes with high probability in a constant number of iterations. The main difference of our convergence analysis compared to \cite{RU01} is that the right edge degree distribution in our graphs is Poisson distributed, while the right degree is regular in \cite{RU01}. 
\end{itemize}

\begin{table}
\centering
\begin{tabular}{ |c |  c| }
\hline
Notation & Description \\
\hline
$p_j$ & average fraction of non-recovered significant components at iteration $j$ \\
\hline
$\eta$ & average degree of right nodes \\
\hline
$d$ & degree of left nodes in $d$-regular construction \\
\hline
$D$ & truncation level for the harmonic distribution \\
\hline
$\lambda(x)$ & left edge degree polynomial \\
\hline
$\rho(x)$ & right edge degree polynomial \\
\hline 
$Z$ & number of uncolored edges after $\ell$ iterations \\
\hline
\end{tabular}
\caption{Table of Notation for Sections \ref{sec:proof1} and \ref{sec:proof2}.}
\end{table}

At each iteration of PhaseCode, we call the giant component as the largest set of signal components (left nodes) that have been resolved relative to each other. The algorithm follows $3$ major steps to recover the active left nodes by coloring them.  
\begin{itemize}
\item \emph{Step $1$:} All the singleton right nodes and their corresponding left nodes are detected.
\item \emph{Step $2$:} Strong doubletons are detected, and the color of the corresponding 2 left nodes get merged. We call the largest set of left nodes that {\em chain} hands together through these strong doubletons to be the giant component at this step.
\item \emph{Step $3$:} After the initial giant component is formed, at each iteration of the algorithm, left nodes are colored one at a time through resolvable multitons, and become part of the giant component. 
\end{itemize}

Now we analyze the message passing algorithm. A left node $v$ passes a $0$ message to neighbor right node $c$ if it is not colored (i.e. it is not part of the giant component). Let $p_j$ be the probability that a random message sent from a left node to a right node is $0$, at iteration $j$ of the algorithm. The density evolution equation is an equation relating $p_j$ to $p_{j+1}$. Similarly, a right node $c$ passes a message $0$ to neighbor left node $v$ if it can not get colored (become part of the giant component). Let $q_j$ be the probability that a random message sent from a right node to a left node is $0$, at iteration $j$ of the algorithm. Under the tree-like assumption, and for $j \geq 2$ one has
\begin{equation}\label{eq:density}
p_{j+1} = (1 + e^{-\eta} - e^{-\eta p_j})^{d-1}.
\end{equation}
Here is a proof of Equation \eqref{eq:density}. A left node $v$ passes a $0$ message to right node $c$ at step $j+1$, if all of the other $d-1$ neighbor right nodes of $v$ pass message $0$ to $v$ at step $j$. That is $p_{j+1} = q_j^{d-1}$. Note that for $j \geq 2$, if a right node is a singleton, it passes message $0$ to neighbor left nodes, since in PhaseCode only resolvable multitons can color active left nodes after the second step of the algorithm. 


We calculate $q_j$ as follows. A right node $c$ sends a message to a left node $v$ that it is part of the giant component if $c$ is connected to a non-empty set of left nodes other than $v$, and those left nodes are all in the giant component. Thus, 
\begin{align*}
1 - q_j = \sum_{i = 2}^\infty \rho_i (1-p_j)^{i-1} &= \rho(1-p_j) - \rho_1 \\
&= e^{-\eta p_j} - e^{-\eta}.
\end{align*}
\begin{figure}
\centering
    \includegraphics[width= 0.25\textwidth]{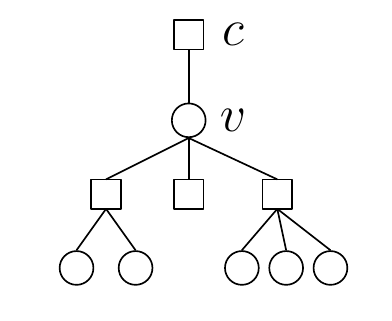}
  \caption{Length-$2$ tree-like neighborhood of $(v,c)$ for $d=4$. The neighborhood is the subgraph of all the edges and nodes of paths having length less than or equal to $2$, that start from $v$ and the first edge of the path is not $(v,c)$.\label{fig:tree}} 
\end{figure}
This proves \eqref{eq:density}. See Figure \ref{fig:tree} for an illustration of the proof for the case $d = 4$.  

\begin{remark}
Note that if a right node is a singleton, it cannot recover the corresponding active left node in both phase and magnitude. This is a fundamental difference of our decoding process compared to that of conventional peeling-based
decoders such as the LDPC decoder for erasure channel \cite{RUbook}. In LDPC decoding, since there is no phase ambiguity, as soon as a singleton is detected, the corresponding non-zero component is recovered and it is peeled
from all other right nodes that are connected to that component. However, active left nodes in singletons cannot be peeled in our setting. Indeed, our problem has the peculiar attribute that singleton right nodes, while critical to initiating the growth of the giant component at the outset, are not useful once a giant component is formed, and too many singletons actually hurt the system performance by featuring useless isolated measurements. This is a significant departure from ``phase-aware'' measurement systems like LDPC codes. This is also the key reason to why our density evolution equation in \eqref{eq:density} differs from that of linear measurement systems \cite{Sameer,RUbook}, which is
$
p_{j+1} = (1 - e^{-\eta p_j})^{d-1}.
$
\end{remark}

An interesting but unfortunate fact is that $p_0 = 1$ is a fixed point of the density evolution equation. Thus, one cannot use \eqref{eq:density} at the outset to follow the evolution of $p_j$, and to argue that it goes close to $0$, since $p_j$ can get stuck at $1$. To use Equation \eqref{eq:density}, we need a more careful characterization of the first two steps of the algorithm that form the giant component. At the first iteration, all the (active) left nodes that are connected to at least one singleton right node are found. Since the relative phase of these signal components is not known, no giant component is formed yet; thus, $p_1 =1$. At the second iteration, the giant component is formed by merging the colors of left nodes in strong doubletons. Recall that a strong doubleton right node is a right node that is connected to two colored left nodes. After the giant component is formed in the second iteration, the probability that a randomly chosen left node is not part of the giant component is $p_2$. If one can show that $p_2$ is small enough such that after a fixed number of iterations $p_j$ gets close to $0$, then concentration bounds can be used to show that the number of left nodes not being in the giant component is indeed highly concentrated around its mean after $\ell$ iterations, that is $K p_\ell$. In Lemma \ref{lem:unstable}, we show that if $p_2 = 1- \delta$ for some arbitrary constant $0 < \delta < 1$ independent of $K$, $p_j$ gets close to $0$ after a constant number of iterations. Clearly $p_2 = 1 - \delta$ if there exists a giant component of size linear in $K$ after the second step. In Lemma \ref{lem:giant}, we find the condition for left-regular bipartite graph under which a linear size giant component will be formed after the second step of the algorithm. 

\begin{lemma}\label{lem:giant} 
There exist operating points $(d,M = cK)$ for which with probability $1 - \mathcal{O}(1/M)$, a giant component of size linear in $K$ is formed after the second step of PhaseCode. In particular, $(d = 5,3.11 \leq c \leq 19.24)$ and $(d=8, 3.48 \leq c \leq 55.36)$ are two of these operation points. 
\end{lemma}
See Appendix \ref{app:gc} for the proof.

\begin{remark}
Lemma \ref{lem:giant} shows that for large enough $m$, a positive fraction of the signal components can get recovered after the second iteration of the algorithm. Thus, PhaseCode gets a proper \emph{jump-start}, which is essential for proving that the algorithm terminates after a constant number of iterations, and successfully recovers an arbitrarily-close-to-one fraction of the signal components.
\end{remark}

\begin{remark}
As one observes in Lemma \ref{lem:giant}, if $M$ is larger than some threshold (which corresponds to more measurements), the giant component will not get formed. At a first glance, this sounds counter-intuitive since having more right nodes seems to only help. However, one should keep in mind that in the statement of the lemma, the left degree $d$ is kept fixed. Intuitively, when $M$ is too large, for a fixed small $d$, the bipartite graph (with active left nodes) becomes so sparse that there are too few doubletons to form a giant component.
\end{remark}

\begin{corollary}\label{cor:giant}
There exists a constant $0 < \delta < 1$ independent of $K$, such that $p_2 = 1 - \delta$. 
\end{corollary}

Due to the formation of a linear-size giant component in step $2$ of the algorithm, we can revisit the density evolution equation \eqref{eq:density}:
$$
p_{j+1} = (1+e^{-\eta} - e^{- \eta p_j})^{d-1},
$$ 
with the aid of Corollary \ref{cor:giant}, which guarantees that $p_2$ is strictly smaller than $1$. Recall that $p_0 = 1$ is a fixed point of \eqref{eq:density}. But with the giant component formation, we can break away from the shackles of ``being stuck'' at $p_0 = 1$. With $p_2 < 1$, we hope to find a better fixed point of \eqref{eq:density} to which our density evolution will converge.
  
Towards this end, ideally one wants Equation \eqref{cor:giant} to have the property 
\begin{equation}\label{eq:decreasing}
p_{j+1} = (1 + e^{-\eta} - e^{-\eta p_j})^{d-1} < p_j,
\end{equation}
for all $p_j \in (0,1)$. Let's take a closer look at the fixed point equation
\begin{equation}\label{eq:fixed}
t = f(t) = (1 + e^{-\eta} - e^{-\eta t})^{d-1}.
\end{equation}

\begin{table*}[t]
\centering

\begin{tabular}{ c | c  c  c  c  c c c}

$d$ & $4$ & $\mathbf{5}$ & $6$&  $7$ &   $\mathbf{8}$ & $9$ & $10$ \\
\hline
$p^*$ &  $2.7 \times 10^{-2}$ &   $1.1 \times 10^{-3}$ & $8 \times 10^{-5}$&  $3.2 \times 10^{-6}$ &  $1 \times 10^{-7}$   & $2.9 \times 10^{-9}$ & $7 \times 10^{-11}$\\ 
\hline
$c$  & $3.31$ &  $\mathbf{3.11}$ &$3.18$  &$3.32$   &$3.48$ & $3.66$ & $3.85$
\end{tabular}
\caption{The table shows how the error floor, $p^*$, and $c = M/K$ (which indirectly determines the number of measurements) vary for different values of left degree, $d$.  The minimum value of $c$ is $3.11$ that is achieved when $d=5$. Moreover, one can see that $p^*$ decreases as $d$ increases.}
\label{tab:2}
\end{table*} 

As mentioned, one solution is $t^*_1 = 1$. As we can break away from $t_1^*$, fortunately there exists another solution approximately at $t^*_2 \simeq e^{-\eta (d-1)}$ which is close to $0$. 
To see this, consider the equation $y = (1 + e^{-\eta} - e^{-\eta x})^{d-1}$. Suppose that $0 < x = e^{- \eta (d-1)} \ll 1$. Then, $e^{-\eta x} \simeq 1$ and $1 + e^{-\eta} - e^{-\eta x} \simeq e^{-\eta}$. Thus, $y = x$ which shows that $e^{-\eta (d-1)}$ is approximately another fixed point of \eqref{eq:density}.\footnote{Of course, one can easily find the exact solution to \eqref{eq:fixed}, using numerical methods for given values of $d$ and $\eta$.} From now on, we will refer to this fixed point as the error floor $p^*$.

\begin{lemma}\label{lem:unstable}
Let $d=5$. If $2.33K \leq M \leq 13.99K$, then the fixed point equation \eqref{eq:fixed} has exactly 2 solutions for $t \in [0,1]$: $t^*_1 = 1$ and $t^*_2 \simeq e^{-\eta (d-1)}$ (See Figure \ref{fig:density}). For $d = 8$, a similar result holds if $2.63K \leq M \leq 47.05K$. 
\end{lemma}
See Appendix \ref{app:fp} for the proof.

\begin{figure*}[t]
        \centering
        \begin{subfigure}{0.35\textwidth}
                \includegraphics[width=7.5cm,height=5cm]{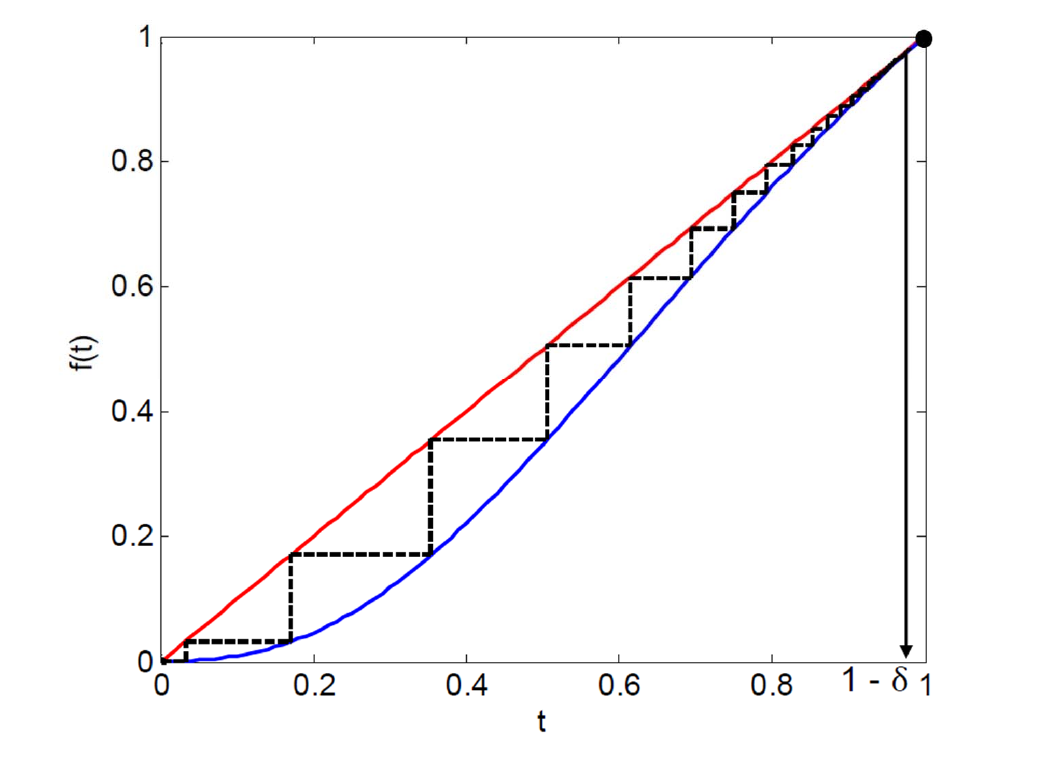}
                \caption{The density evolution curve \\for parameters $d=5$ and $\eta=2$.}
                \label{fig:de1}
        \end{subfigure}
        \qquad
        \begin{subfigure}{0.35\textwidth}
                \includegraphics[width=7.5cm,height=5cm]{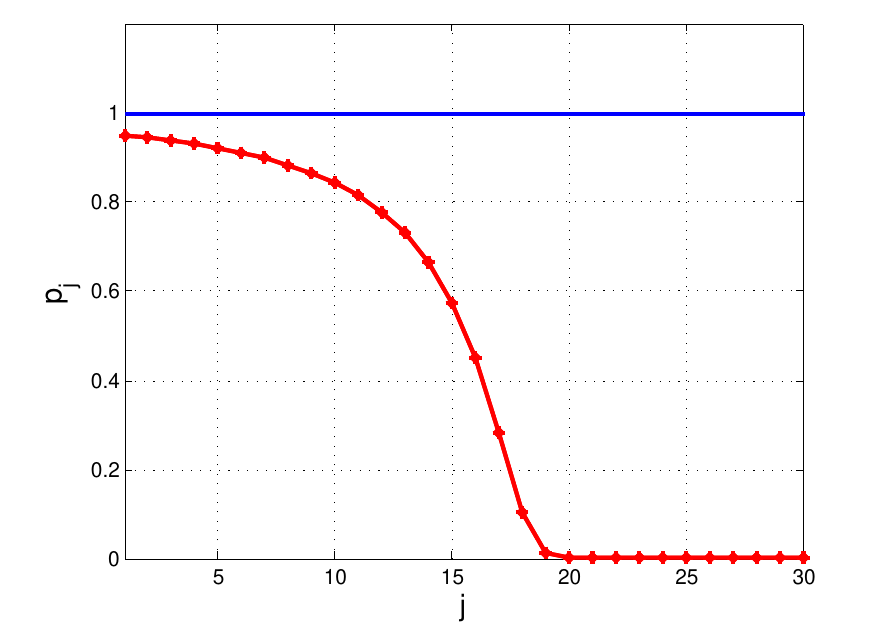}
                \caption{The evolution of $p_j$ after each \\iteration for $d=5$ and $\eta = 2$.}
                \label{fig:de2}
        \end{subfigure}
        \caption{Figure $(a)$ illustrates the density evolution equation, $p_{j+1} = f(p_j)$, for Regular PhaseCode. In order to track the evolution of $p_{j}$, pictorially, one draws a vertical line from $(p_j,p_j)$ to $(p_j,f(p_j))$, and then a horizontal line between $(p_j,f(p_j))$ and $(f(p_j),f(p_j))$. Since the two curves meet at $(1,1)$ if $p_0 =1$, then $p_j$ gets stuck at $1$. However, if $p_0 = 1 - \delta$, $p_j$ decreases after each iteration, and it gets very close to $0$. Figure $(b)$ illustrates the same phenomenon by showing the evolution of $p_j$ versus the iteration, $j$. Note that in this example, $p_j$ gets very close to $0$ after only $20$ iterations. }\label{fig:density}
\end{figure*}



The following corollary is a direct result of Lemma \ref{lem:unstable}.

\begin{corollary}\label{cor:de}
For any $\epsilon > 0$, there exists a constant $\ell(\epsilon)$ such that $p_\ell \leq p^* + \epsilon$. 
\end{corollary}

Table \ref{tab:2} illustrates how the error floor $p^*$ and the minimum ratio of right nodes to active left nodes $c = M/K$ change for different values of $d$. If our reliability target allows the error floor to be set at $1.1\times 10^{-3}$, then $d=5$ minimizes the number of required right nodes. Recall that the total number of measurements is $m = 4M$ which matches the result of Table \ref{tab:1}. (See Section \ref{sec:ed}) If one wants to achieve smaller error floor, then $d$ and $c$ should be both increased.

In the density evolution analysis so far, we have shown that the \emph{average} fraction of active signal components that cannot be recovered will be arbitrarily close to the error floor after a fixed number of iterations, provided that the tree-like assumption is valid. It remains to show that the actual fraction of left nodes that are not in the giant component after $\ell$ iterations is highly concentrated around $p_\ell$. Towards this end, first in Lemma \ref{lem:tree} we show that a neighborhood of depth $\ell$ of a typical edge is a tree with high probability for a constant $\ell$. Second, in Lemma \ref{lem:concentration}, we use the standard Doob's martingale argument \cite{RU01}, to show that the number of active signal components that are not recovered after $\ell$ iterations of the algorithm is highly concentrated around $K p_{\ell}$. 

Consider a directed edge $\vec{e} = (v,c)$ from a left-node $v$ to a right-node $c$. Define the directed neighborhood of depth $\ell$ of $(\vec{e})$ as $\mathcal{N}_{\vec{e}}^{\ell}$, that is the subgraph of all the edges and nodes on paths having length less than or equal to $\ell$, that start from $v$ and the first edge of the path is not $\vec{e}$. As an example, the directed neighborhood of depth $2$ of $(\vec{e})$ is shown in Figure \ref{fig:tree}.
 
\begin{lemma}\label{lem:tree}
For a fixed $\ell^*$, $\mathcal{N}_{\vec{e}}^{2\ell^*}$  is a tree-like neighborhood with probability at least $1 - \mathcal{O}(\log(K)^{\ell^*}/K)$.
\end{lemma}

The proof is provided in Appendix \ref{app:tree}. 

\begin{lemma}\label{lem:concentration}
Over the probability space of the ensemble of $d$-left-regular graphs $\mathcal{C}^K_1(d,M)$, let $Z$ be the number of uncolored edges\footnote{An edge is colored if its corresponding left node is colored.} after $\ell$ iterations of the PhaseCode algorithm. Then, for any $\epsilon > 0$, there exist a large enough $K$ and constants $\beta$ and $\gamma$ such that 
\begin{align}\label{eq:ctcf}
|\mathbb{E}[Z] - Kdp_\ell|& <  Kd\epsilon /2\\ \label{eq:mg}
\mathbb{P}(|Z - Kdp_\ell| &> Kd\epsilon) < 2e^{-\beta \epsilon^2 K^{1/(4\ell + 1)}},
\end{align} 
where $p_\ell$ is derived from the density evolution equation \eqref{eq:density}.
\end{lemma}

The proof is provided in Appendix \ref{app:mg}. 

Now gathering the results of Corollary \ref{cor:de} and Lemmas \ref{lem:giant} and \ref{lem:concentration} completes the proof of Theorem \ref{thm:main}. Note that since the construction of the bipartite graph is random, the fraction $p$ of the non-zero components that can be missed are distributed uniformly at random among the $K$ non-zero components. Indeed, the missed components are only a function of the graph structure that has a distribution which is oblivious to the indices of the left nodes by construction. Further, note that the dominant probability of error is due to the event that the giant component is not formed in the second iteration which happens with probability $\mathcal{O}(1/K)$. It is worth mentioning that Lemma \ref{lem:tree} is used only to prove Lemma \ref{lem:concentration}. Thus, the event that an edge does not have a tree-like neighborhood, which happens with probability $\mathcal{O}(\frac{\log(K)^{\ell^*}}{K})$, is not an error event of the algorithm. Given that a giant component has been formed after the second step of the algorithm, the error event of the algorithm is the event that more than a fraction $p$ of the non-zero signal components are missed, and the probability of such event is upper bounded in \eqref{eq:mg}.

\subsection{Proof of Theorem \ref{thm:2}}\label{sec:proof2}

Recall that we design the left degree distribution $\lambda(x) = \sum_{i \geq 1} \lambda_i x^{i-1}$ of Irregular PhaseCode as follows: 
$\lambda_i = 0$ for $i \geq D+1$ and
\begin{align}
\lambda_i = \frac{1}{i-1} \times \frac{1}{h(D-1)}, ~ 2 \leq i \leq D,
\end{align}
where $D$ is a (large) constant and $h(x) = \sum_{i=1}^x 1/i$.

We design the number of right nodes to be $M = K/(1-\epsilon)\simeq K(1+\epsilon)$. How to choose constants $D$ and $\epsilon$ will be shortly clarified in Lemma \ref{lem:ir}. The average degree of left nodes (of the pruned graph with $K$ active left nodes) is $\bd = \frac{1}{\sum_i \lambda_i/i}$. To see this, let $E$ be the number of edges of the graph. Then, the number of left nodes of degree $i$ is $E\lambda_i/i$ since $\lambda_i$ is the fraction of edges with degree $i$ on the left. Thus, the number of left nodes is $\sum_i E\lambda_i/i$. So the average left degree is 
$$
\bar{d} = \frac{E}{\sum_i E\lambda_i/i} = \frac{1}{\sum_i \lambda_i/i}.
$$ 
Thus, with our design, 
$$
\bar{d} = (\sum_{i=2}^D \frac{\lambda_i}{i})^{-1} = h(D-1) \frac{D}{D-1}.
$$
Consequently, the Poisson density parameter of the right-node degree distribution is:
$$
\eta = \frac{K\bd}{M} = h(D-1) \frac{D}{D-1} (1-\epsilon).
$$


\begin{lemma}\label{lem:de}
Let $f(x) = \lambda(1 + e^{-\eta} - e^{-\eta x})$. The fixed point equation $x = f(x)$ has exactly two solutions, $x^*_1 = 1$ and $0 < x^*_2 < 1$, in the interval $x \in [0,1]$. Furthermore, if $f'(1) > 1$, then $f(x) < x$ for $x \in (x^*_2,1)$. 
\end{lemma}

See Appendix \ref{app:de2} for the proof.

\begin{figure*}[t]
        \centering
        \begin{subfigure}{0.35\textwidth}
                \includegraphics[width=7.5cm,height=5cm]{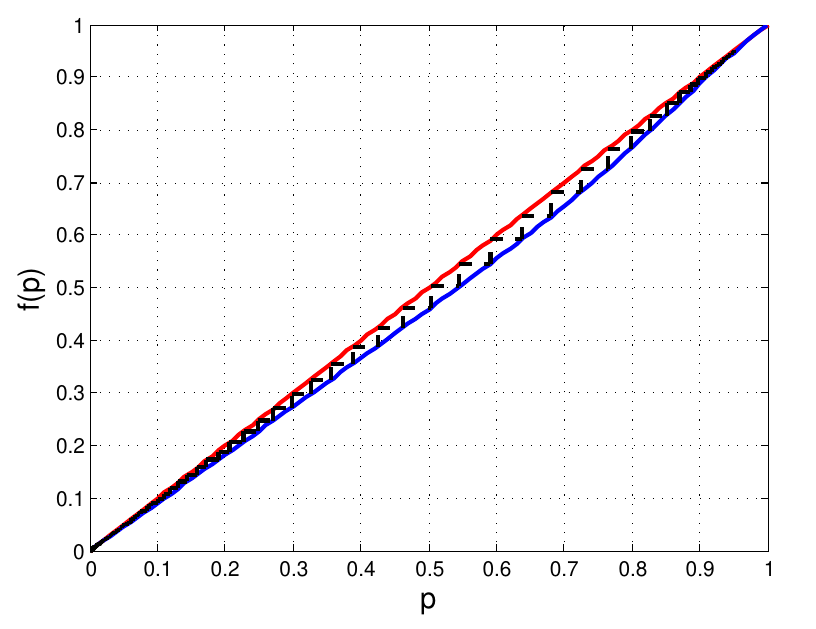}
                \caption{The density evolution curve for parameters $K = 10^5$, $\epsilon=0.1$ and $D = 10^3$.}
                \label{fig:IRde1}
        \end{subfigure}
        \qquad
        \begin{subfigure}{0.35\textwidth}
                \includegraphics[width=7.5cm,height=5cm]{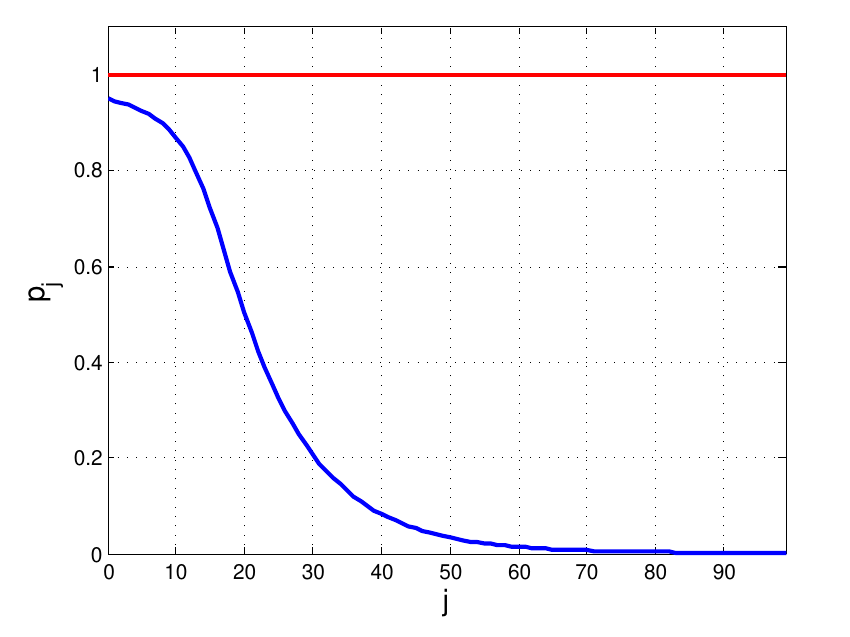}
                \caption{The evolution of $p_j$ after each iteration for parameters $K = 10^5$, $\epsilon=0.1$ and $D = 10^3$.}
                \label{fig:IRde2}
        \end{subfigure}
        \caption{Figure $(a)$ illustrates the density evolution equation for Irregular PhaseCode, which is similar to Figure \ref{fig:de1}. Figure $(b)$ illustrates the same phenomenon showing the evolution of $p_j$ versus the iteration, $j$.  Note that in this example, since $\epsilon = 0.1$ and we are operating very close to the capacity, $p_j$ gets very close to $0$ after around $90$ iterations, which is much larger than around $20$ iterations needed by Regular PhaseCode so that $p_j$ gets very close to $0$. The reason is that the gap between the two curves in $(a)$ gets smaller once the number of measurements is close to the capacity.}\label{fig:density1}
\end{figure*}

As shown in Lemma \ref{lem:de}, given that $f'(1)>1$, the density evolution has a fixed point at $1$, and the other fixed point of the equation is approximately $p^* \simeq \lambda(e^{-\eta})$, which corresponds to the error floor of the algorithm. In the following lemma, we show that for any arbitrarily small numbers $p^*$ and $\epsilon$, there exists a large enough constant $D(p^*,\epsilon)$ such that $f'(1) > 1$. This shows that with only $4M = 4K/(1-\epsilon) \simeq  4K(1+\epsilon)$ measurements, Irregular PhaseCode algorithm can recover an arbitrarily-close-to-one fraction of the non-zero signal components. So given that the coloring procedure starts (the density evolution equation can be started from $1- \delta$), Irregular PhaseCode is capacity-approaching. 

Now we show that a linear size giant component of colored left nodes can be formed similar to Lemma \ref{lem:giant} using a second stage of only $m' = \epsilon' K$ extra measurements. 
By assumption of Theorem \ref{thm:2}, the support of the non-zero components of the signal is uniformly random. Now fix some arbitrarily small constant $\delta' > 0$. Let $\tilde{\vect{x}}$ be the vector of the first $\delta' n$ components of the signal. By the law of large numbers, the number of non-zero elements of $\tilde{\vect{x}}$ is $\delta' K + o(K)$. Consider the sub-problem of forming a giant component of size linear in $K$ in $\tilde{\vect{x}}$. By Lemma \ref{lem:giant}, one can design $m'  = 14 \delta' K$ measurements to form the giant component. Thus, $\epsilon' = 14 \delta'$. Since $\delta'$ can be made arbitrarily small, $\epsilon'$ can also be made arbitrarily small.  

The main lemma for establishing the proof of Theorem \ref{thm:2} is as follows.

\begin{lemma}\label{lem:ir}
For any $p^* > 0$ and any $\epsilon > 0$, there exists a large enough constant $D(\epsilon,p^*)$ such that $M = K(1-\epsilon)^{-1}\simeq K(1+\epsilon)$ is the number of right nodes (bins), and $p_j$ converges to $p^*$ as $j$ goes to infinity.
\end{lemma}

See Appendix \ref{app:ir} for the proof.

\begin{corollary}\label{cor:IRde}
Given that $p_2 = 1 - \delta$, for any $\epsilon_1 > 0$, there exists a constant $\ell(\epsilon_1)$ such that $p_\ell \leq p^* + \epsilon_1$. 
\end{corollary}

The rest of the proof is similar to Theorem \ref{thm:main}. It remains to show that the actual fraction of active signal components that are not recovered after $\ell$ iterations is highly concentrated around $p_\ell$. Since the maximum degree of left nodes is again a constant $D$, the exact procedure in Section \ref{sec:proof1} (Lemmas \ref{lem:tree} and \ref{lem:concentration}) can be used to get a similar concentration bound as in Lemma \ref{lem:concentration}. Now the total number of measurements is $m = 4K(1+ \epsilon) + m' = 4K(1 + \epsilon + \epsilon')$. Since $\epsilon$ and $\epsilon'$ can be made arbitrarily small, the proof of Theorem \ref{thm:2} is complete.

\section{Fourier-Friendly PhaseCode}\label{sec:practical}

In some applications such as optical imaging \cite{Rodenburg,Popov}, the design of the measurement matrix cannot be arbitrary. In optical imaging, the object of interest, signal $\vect{x}$, can be passed through an optical diffraction pattern or a mask and an optical Fourier lens. A typical setup for optical imaging is shown in Figure \ref{fig:masklens}. With a complex-valued mask, we can modulate each component of the signal $x_i$ by some complex number $d_i$, while the lens takes the Fourier transform of the signal. For example, consider passing the signal through a mask and then Fourier lens which is common in optical imaging. The output of this transform is $\mat{F}\mat{D}\vect{x}$, where $\mat{F}$ is the DFT matrix of length $n$ and $\mat{D} \in \mathbb{C}^{n\times n}$ is a diagonal mask matrix (Figure \ref{fig:MF}). In general, it is possible to have multiple stages of masks and lenses. While increasing the number of stages can make the system more complex, in many optical systems, having up to two stages is considered practical \cite{Wang,Pavani}. In our proposed solution, we will have two masks for all measurements.

In this section, we show how one can have a Fourier-friendly implementation of the set of measurements described in previous sections. We first provide an overview of the result of \cite{Sameer} on constructing a sparse-graph code using ``Chinese Remainder Theorem'', in Subsection \ref{sec:CRT}. In Subsection \ref{sec:FFS}, we show how our proposed measurements can be obtained in a Fourier-friendly setup, with the aid of the result of \cite{Sameer}.

\begin{figure}
\centering
    \includegraphics[width= 0.45\textwidth]{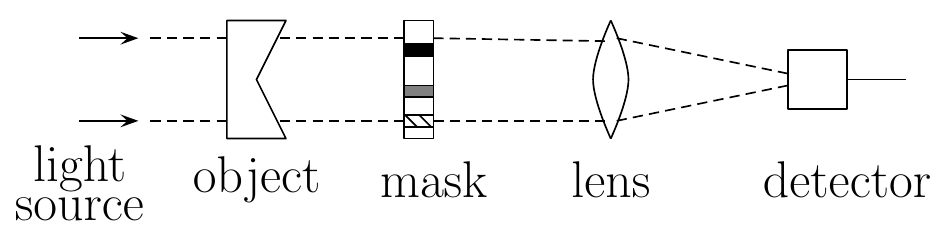}
  \caption{A typical setup for many optical system where the object of interest is passed through a coded diffraction pattern or a mask , and then through a Fourier lens.\label{fig:masklens}} 
\end{figure}

\begin{figure}
\centering
    \includegraphics[width= 0.45\textwidth]{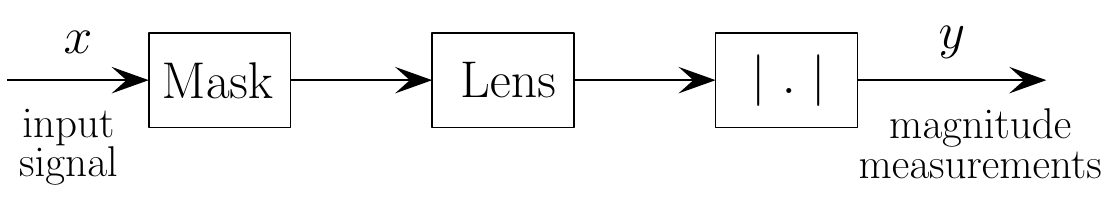}
  \caption{The block diagram of an optical imaging system where signal $x$ is passed through a mask (modulated by a diagonal matrix), and then passed through a lens (DFT matrix). The magnitude block, $|.|$, is showing that the phase information is not available in the measurements. \label{fig:MF}} 
\end{figure}

\subsection{Ensemble of Graphs Constructed by Chinese Remainder Theorem}\label{sec:CRT}
In this subsection, we provide a brief overview of the result in \cite{Sameer} that uses the ``Chinese Remainder Theorem" (CRT) to construct a deterministic and well-structured coding matrix that is also of practical interest. We use this construction to design a Fourier-friendly measurement matrix. For more details about the theory of the ensemble of graphs constructed by the CRT, we refer the readers to \cite{Sameer}.

In Section \ref{sec:proof1}, we analyzed the performance of PhaseCode for the ensemble of graphs $\mathcal{C}^K_1(d,M)$. In this ensemble, each left node is connected to exactly $d$ right nodes randomly. Now we consider another ensemble $\mathcal{C}^K_2(\mathcal{F},m)$. Define the set $\mathcal{F}$ as $\mathcal{F} = \{ f_1,f_2,\ldots,f_d \}$. Partition the right nodes into $d$ sets. Let the number of right nodes in stage $i$ be $f_i$; thus, $\sum_{i=1}^d f_i = m$. In this construction, each left node is connected to exactly one right node per stage randomly. Therefore, we again end up with having a bipartite graph with left regular degree $d$. Assuming that $f_i = F + \Theta(1)$ for all $i$ and consequently $F = \Theta(K)$, the edge degree distribution of the right nodes does not change for large enough $K$ and is given in \eqref{eq:edge}. Therefore, the tree analysis and the density evolution equation stated in \eqref{eq:density} remain the same, and one can essentially get all the previous results using this ensemble. 

Note that sampling a graph from $\mathcal{C}^K_2(\mathcal{F},m)$ has no practical advantage over sampling from the ensemble $\mathcal{C}^K_1(d,M)$. However, we use the CRT to show that if the $K$ non-zero components of the signal is chosen uniformly at random with replacement from the $n$ components, and if $K$ is in the sub-linear regime (more specifically, $K = n^\delta$ for some $\delta \in (0,1)$), one can design a deterministic coding matrix which consists of $d$ stages of sub-matrices with rows that are circularly-shifted versions of a deterministic \emph{subsampling} pattern. The subsampling rate at stage $i$ is $f_i$. In the following example, we demonstrate how the deterministic matrix is constructed. 

\begin{example}
Suppose that the coding matrix has two stages with $f_1 = 2$ and $f_2 = 3$. Assume that $n=6$. Then, the coding matrix is 
$$
\left( \begin{tabular}{ccccccccc}
1 & 0 & 1 & 0 & 1 & 0  \\
0 & 1 & 0 & 1 & 0 & 1  \\
\hline
1 & 0 & 0 & 1 & 0 & 0  \\
0 & 1 & 0 & 0 & 1 & 0  \\
0 & 0 & 1 & 0 & 0 & 	1
\end{tabular}
\right).
$$ 
\end{example}

Now, we formally define the ensemble of graphs constructed by the CRT. First, assume $n = \prod_{i=1}^d f_i$ (i.e. $K = \Theta(n^{1/d})$). Partition the set of $m = \sum_{i=1}^d f_i$ right nodes to $d$ stages in the trivial way. Suppose that the $K$ non-zero components of the signal are chosen uniformly at random with replacement from the $n$ components. Note that the ``with replacement" assumption might lead to having a signal with less than $K$ non-zero components, but this is only a technical assumption that we need to make, and via simulations we will show the good performance of the CRT-based code for exactly $K$-sparse signal. Let $\mathcal{I} = (i_1,i_2,\ldots,i_K)$ denote the non-zero components where $1 \leq i_k \leq n, ~ 1 \leq k \leq K$. We associate the integers from $0$ to $n-1$ to $d$ numbers $(r_1,r_2,\ldots,r_d)$ using the CRT, where $0 \leq r_i \leq f_i -1$; thus, $i_k$ uniquely determines one right node per stage. The way this association is done will be explained shortly. Then, each active left node $i_k$ is connected to the associated set of right nodes that are determined by $(r_1,r_2,\ldots,r_d)$. The ensemble $\mathcal{C}^K_3(\mathcal{F},m)$ is the collection of all the graphs that are constructed as described. Furthermore, the uniformly at random selection of $\mathcal{I}$ makes sure that all these graphs occur with equal probability. See \cite{Sameer} for details.

To show how we associate $\mathcal{I}$ to $(r_1,r_2,\ldots,r_d)$, we need to review the Chinese Remainder Theorem. Let $n = \prod_{i=1}^d f_i$ and $f_i$'s are pairwise co-prime positive integers. The theorem states that every integer $n'$ between $0$ and $n-1$
is uniquely represented by the sequence $(r_1, r_2, . . . , r_d)$ of its remainders modulo $f_1,f_2,\ldots,f_d$ respectively and vice-versa. We use this unique CRT mapping to associate the active left nodes with $d$ right nodes. 

\begin{lemma}\label{lem:CRT}
\cite{Sameer} The ensembles $\mathcal{C}^K_2(\mathcal{F},m)$ and $\mathcal{C}^K_3(\mathcal{F},m)$ are identical.
\end{lemma}
\begin{proof}
Clearly, $\mathcal{C}^K_3(\mathcal{F},m) \subset \mathcal{C}^K_2(\mathcal{F},m)$. The reverse is also true by CRT since there is a unique integer between $0$ to $n - 1$ with remainders $r_i$ modulo $f_i$ for all $i$.
\end{proof}

\begin{figure}[t]        
    \centering
    \includegraphics[width=0.45\textwidth]{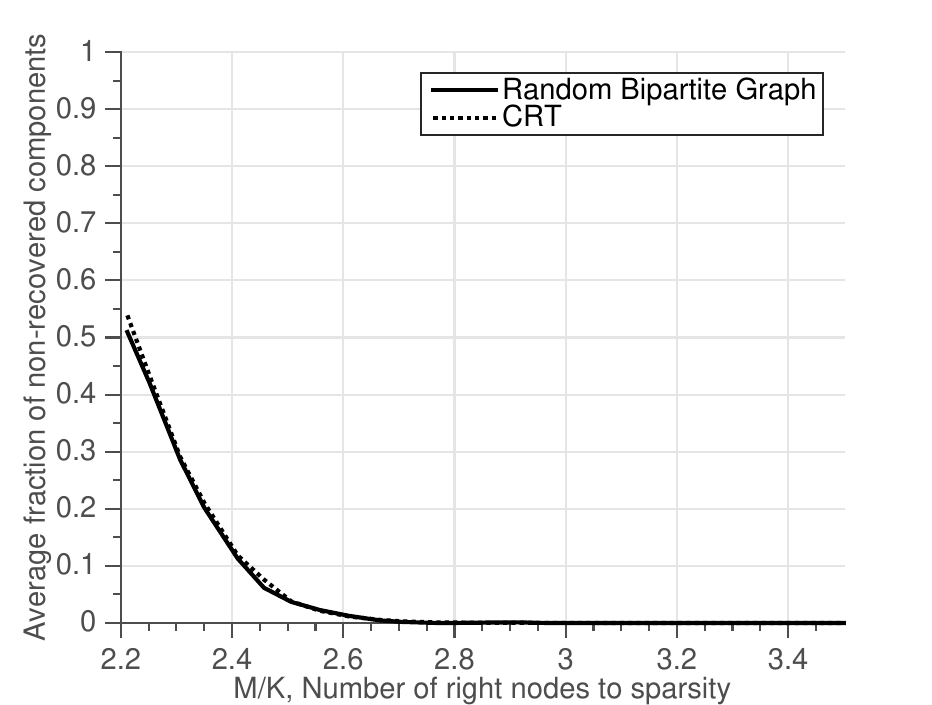}
    \caption{\textbf{Comparison of random left-regular bipartite graph ensemble and CRT ensemble.} We choose the left degree $d=7$, and construct an appropriate CRT ensemble based on $\mathcal{F} = \{47,49,50,53,57,59,61\}$. The number of right nodes is determined by $\mathcal{F}$, i.e., $M = \sum_{i=1}^{d}{f_i} = 376$. 
Each operating point is averaged over $10000$ runs. We observe negligible difference in performance between the two ensembles.}
    \label{fig:fig2_crt}
\end{figure}

Figure \ref{fig:fig2_crt} demonstrates the performance of PhaseCode with two ensembles: $\mathcal{C}^K_1(d,M)$ and $\mathcal{C}^K_3(\mathcal{F},m)$. We choose $d=7$ and $\mathcal{F} = \{47, 49, 50, 53, 57, 59, 61\}$. Thus, $M = \sum_{i=1}^{d}{f_i} = 376$. We varied the value of $K$ ($107 \leq K \leq 170$) such that $M/K$ varies between $2.2$ and $3.5$. Each point is averaged over $10000$ runs to determine the error probability. One can observe negligible difference between the performance of the algorithm for the two ensembles.

In the following we provide remarks of how one can extend the above construction of CRT.
\begin{remark}
In the above example of CRT construction, we implicitly assumed $K = \Theta(n^{1/d})$. The technique can be extended to cases where $K = \Theta(n^{\alpha/d})$ for $0 \leq \alpha < d$. Instead of using $\mathcal{F}$ as heights of the $d$ stages of the bipartite graph, we use $\mathcal{F}' = \{f'_1, ..., f'_d\}$, where
\[
f'_i = \prod_{j=0}^{\alpha-1} f_{\left((i+j)\bmod d\right) + 1}.
\]
For example, if $\alpha = 2$ and $d=7$, one can convert a set of coprimes 
$$
\{f_1, f_2, f_3, f_4, f_5, f_6, f_7\}
$$ 
to the set 
$$\mathcal{F} = \{f_1 f_2, f_2 f_3, f_3 f_4, f_4 f_5, f_5 f_6, f_6 f_7, f_7 f_1\}.
$$
Then, $M = \sum_{i=1}^{d} \prod_{j=0}^{\alpha-1} n_{\left((i+j)\bmod d\right) + 1} = \Theta(n^{2/d})$, which is in the order of $K$. Because $\mathcal{F}$ can be chosen from a dense set of coprimes, one can always choose it carefully to induce a right number of measurements. For the most general case where $K = \Theta(n^{p/q})$ and $0 \leq p/q < 1$, one can use a similar extension and construction by finding $q$ coprimes and stacking $p$ of them in each stage. We omit details of the technique and refer interested readers to \cite{Sameer}.
\end{remark}

\subsection{Fourier-Friendly Compressive Phase Retrieval}\label{sec:FFS}

Without loss of generality, we consider only a 1-D case for $\vect{x}$ here, though our arguments extend in a straight-forward way to 2-D images as well. 
Suppose that the signal of interest $\vect{x}$ is sparse in the Fourier domain, which is of interest in many optical imaging settings. 
Let $\vect{X} = \mat{F}\vect{x}$ be the Fourier transform of the signal. In Subsection \ref{sec:CRT}, we showed that the coding matrix $\mat{H}$ can be realized using $d$ stages of circulant matrices without changing the performance of sparse-graph codes. To have a Fourier-friendly implementation of the CRT code matrix, we expand each stage of the $f_i \times n$ matrix to a circulant $n \times n$ matrix. Let $\mat{C}$ denote this circulant coding matrix for one stage. In the following, we show that how using our proposed CRT code matrix, one can have access to all the necessary measurements using \emph{only diagonal masks and lenses}. 
Note that we are interested in measurements of the modulated signal by complex exponentials such as $e^{\bi \om \ell}$ or by magnitude modulators $\cos(\om \ell)$.
First let us see how the \emph{plain} measurements without these modulations can be obtained if the coding matrix is circulant. The plain measurements are $|\sum_j C_{ij}X_j|$. Since $\mat{C}$ is circulant, the eigenvectors of $\mat{C}$ are the columns of a unitary Fourier matrix \cite{DSP}. Thus, the eigenvalue decomposition of $\mat{C}$ is $\mat{C} = \mat{F} \mat{D} \mat{F}^{-1}$ for some diagonal matrix $\mat{D}$. Hence, we construct our measurements by modulating the signal $\vect{x}$ with the diagonal mask $\mat{D}$ and then taking a Fourier transform by using an optical lens: 
\begin{align*}
|\mat{F}\mat{D}\vect{x}| &= |\mat{F}\mat{F}^{-1}\mat{C}\mat{F}\vect{x}| \\
&= |\mat{C}\mat{F}\vect{x}| \\
&= |\mat{C}\vect{X}|.
\end{align*}  
For each stage of the CRT code matrix (there are $d$ stages overall), we need one physical experiment. The physical experiment corresponding to the $i$-th stage, where $1 \leq i \leq d$, gives us $n/f_i$ replicas of $f_i$ unique measurements in one shot. As illustrated in Figure \ref{fig:CRT}, for each experiment, the camera measures only one copy of the $f_i$ measurements. Let $\vect{y}_i \in \mathcal{C}^{f_i}$ be the measurements corresponding to stage $i$. Then, the measurements of the different stages are gathered to form the measurement vector $\vect{y} \in \mathcal{C}^m$ as follows: 
$$
\vect{y} = [\vect{y}_1^T,\vect{y}_2^T,\ldots,\vect{y}_d^T]^T.
$$
Thus, the actual sample complexity is still $m = \sum_{i=1}^d f_i = \Theta(K)$.

\begin{figure} 
\centering
    \includegraphics[width= 0.45\textwidth]{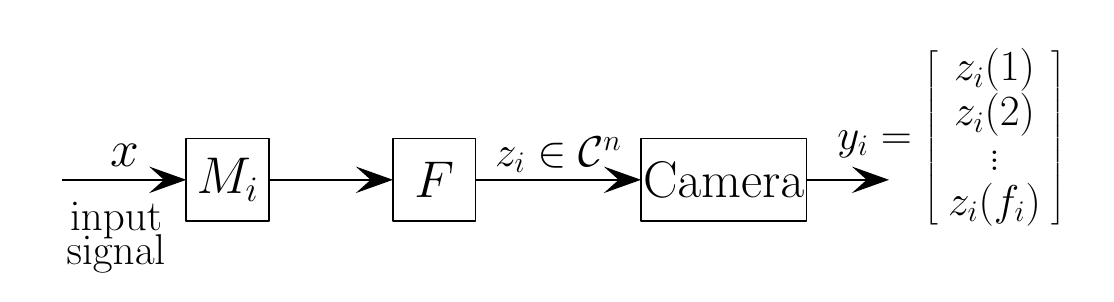}
  \caption{The block diagram of Fourier-friendly compressive phase retrieval using the CRT matrix. The figure shows stage $i$ of the CRT matrix ($1 \leq i \leq d$). The signal of interest, $x$, is passed through a binary mask corresponding to stage $i$, and then the Fourier lens. The output of this experiment is signal $z_i$ of length $n$. However, these $n$ measurements are not unique; they are $n/f_i$ replicas of $f_i$ unique measurements. Thus, the camera only reads the first $f_i$ components of $z_i$. \label{fig:CRT}} 
\end{figure}

Now we explain how one can get access to all the necessary measurement $y_{1,i}$ to $y_{4,i}$. We explain the construction for $y_{1,i}$. Other measurements can be similarly realized. We use $3$ blocks of Fourier transforms (lenses) and $2$ masks as follows. Let $\tilde{\mat{D}}$ be a diagonal matrix such that $\tilde{d}_{\ell \ell} = e^{\bi \om \ell}$. We are interested in constructing the measurements of the form $|\mat{C} \tilde{\mat{D}} \vect{X}|$. This can be done by using two masks, $\tilde{\mat{D}}$ and $\mat{D}$, with three Fourier lenses as follows. 
\begin{align} \label{eq:mask}
|\mat{F} \mat{D} \mat{F} \tilde{\mat{D}} \mat{F} \vect{x}| &= |\mat{F} \mat{F} \mat{C} \mat{F}^{-1} \mat{F} \tilde{\mat{D}} \vect{X}| \\
&= |\mat{F}^2 \mat{C} \tilde{\mat{D}} \vect{X}|.
\end{align}   
Note that $\mat{F}^2$ is just a permutation matrix so we can construct all the measurements $y_{1,i}$ using only two masks and Fourier lenses.  

\begin{remark}
Since each optical lens is equivalent to a Fourier transform, we can also implement a compressive Fourier-friendly phase retrieval algorithm, when $\vect{x}$ is sparse (and $\vect{X}$ is not sparse) by just adding an optical lens to the measurement system as follows. Suppose that $\mat{A}$ is a Fourier-friendly measurement system that is able to recover $\vect{x}$, when $\vect{X}$ is sparse. That is, one is measuring the sparse signal $\vect{X}$ with measurement matrix $\mat{A}\mat{F}^{-1}$. Then, $\mat{A}\mat{F}$ is a Fourier-friendly measurement matrix that is able to recover $\vect{x}$, when $\vect{x}$ is sparse since $\mat{A}\mat{F}\vect{x} = \mat{A}\mat{F}^{-1} \mat{F}^2 \vect{x}$. Note that $\mat{F}^2$ is just a permutation matrix; thus, $\mat{F}^2 \vect{x}$ is a sparse signal that is again measured by $\mat{A} \mat{F}^{-1}$.
\end{remark}

\section{Robust PhaseCode}\label{sec:noisy}
\input{noisy.tex}

\section{Conclusion}\label{sec:con}
\input{conclusion.tex}
\section*{Acknowledgment}
The authors would like to thank the anonymous reviewers for many helpful comments.

\bibliographystyle{ieeetr}
\bibliography{journal_abbr,ref_new}

\begin{IEEEbiography}{Ramtin Pedarsani}
Ramtin Pedarsani is an Assistant Professor in ECE Department at the University of California, Santa Barbara.
He received the B.Sc. degree in electrical engineering from the University of Tehran, Tehran, Iran, in 2009, 
the M.Sc. degree in communication systems from the Swiss Federal Institute of Technology (EPFL), Lausanne, Switzerland,
in 2011, and his Ph.D. from the University of California, Berkeley, in 2015. His research interests include networks, 
machine learning, information and coding theory, and transportation systems. 
Ramtin is a recipient of the IEEE international conference on communications (ICC) best paper award in 2014.
\end{IEEEbiography}

\begin{IEEEbiography}{Dong Yin	}
Dong Yin is a PhD student in Department of Electrical Engineering and Computer Sciences at UC Berkeley, working with Prof. Kannan Ramchandran. He is interested in information and coding theory, machine learning, and signal processing. Before coming to Berkeley, he obtained his B.S. from Tsinghua University in China in 2014.
\end{IEEEbiography}

\begin{IEEEbiography}{Kangwook Lee}
Kangwook Lee is a postdoctoral scholar at Information and Electronics Research Institute at KAIST. He obtained his PhD degree in May 2016 from the EECS department at UC Berkeley. He also obtained his MS degree in EECS from UC Berkeley in 2012, and before that he obtained his BS degree in EE from KAIST in 2010.  He is a recipient of the KFAS Fellowship 2010-15. His research interests lie in information theory and machine learning.
\end{IEEEbiography}

\begin{IEEEbiography}{Kannan Ramchandran}
(Ph.D.: Columbia University, 1993) is a Professor of Electrical Engineering and Computer Sciences at UC Berkeley, where he has been since 1999. He was on the faculty at the University of Illinois at Urbana-Champaign from 1993 to 1999, and with AT\&T Bell Labs from 1984 to 1990. He is an IEEE Fellow, and a recipient of the 2017 IEEE Kobayashi Computers and Com- munications Award, which recognizes outstanding contributions to the integration of computers and communications. His research awards include an
IEEE Information Theory Society and Communication Society Joint Best Paper award for 2012, an IEEE Communication Society Data Storage Best Paper award in 2010, two Best Paper awards from the IEEE Signal Processing Society in 1993 and 1999, an Okawa Foundation Prize for outstanding research at Berkeley in 2001, an Outstanding Teaching Award at Berkeley in 2009, and a Hank Magnuski Scholar award at Illinois in 1998. His research interests are at the intersection of signal processing, coding theory, communications and networking with a focus on theory and algorithms for large-scale distributed systems.
\end{IEEEbiography}

\input{appendix.tex}

\end{document}

%% file: header.tex
\usepackage{amsmath}
\usepackage{graphicx}
\usepackage{framed, color} 
\usepackage{pdflscape}
\definecolor{shadecolor}{rgb}{1, 0.8, 0.3}
\usepackage{epstopdf}

\usepackage{amssymb}
\usepackage{amsthm}

\usepackage{amsfonts}
\usepackage{cite}
\usepackage{ifthen}
\usepackage{epsfig}
\usepackage{subcaption}
\usepackage{array}
\usepackage{enumerate}
\usepackage{algorithm}
\usepackage[noend]{algpseudocode} 
\usepackage[all]{xy}
\usepackage{algpseudocode}

\usepackage{times}
\usepackage{tikz}
\usepackage{float}
\usetikzlibrary{shapes,positioning,arrows,automata}
\usepackage{caption}
\newtheorem{theorem}{Theorem}
\newtheorem{lemma}[theorem]{Lemma}

\newtheorem{corollary}[theorem]{Corollary}

\theoremstyle{definition}
\newtheorem{example}{Example}

\newenvironment{remark}[1][Remark]{\begin{trivlist}
\item[\hskip \labelsep {\bfseries #1}]}{\end{trivlist}}
\newenvironment{uni}[1][PhaseCode Algorithm]{\begin{trivlist}
\item[\hskip \labelsep {\bfseries #1}]}{\end{trivlist}}

\newcommand{\rr}[1]{\operatorname{Re}(#1)}
\newcommand{\ii}[1]{\operatorname{Im}(#1)}

\newcommand{\bi}{\ensuremath{\mathbf{i}}}
\newcommand{\om}{\ensuremath{\omega}}
\newcommand{\EE}{\ensuremath{\mathbb{E}}}
\newcommand{\PP}{\ensuremath{\mathbb{P}}}
\newcommand{\bd}{\ensuremath{\bar{d}}}

\newcommand{\beq}{\begin{equation}}
\newcommand{\eeq}{\end{equation}}
\newcommand{\bea}{\begin{eqnarray}}
\newcommand{\eea}{\end{eqnarray}}
\newcommand{\bean}{\begin{eqnarray*}}
\newcommand{\eean}{\end{eqnarray*}}
\newcommand{\bit}{\begin{itemize}}
\newcommand{\eit}{\end{itemize}}
\newcommand{\ben}{\begin{enumerate}}
\newcommand{\een}{\end{enumerate}}
\newcommand{\blem}{\begin{lem}}
\newcommand{\elem}{\end{lem}}
\newcommand{\bthm}{\begin{thm}}
\newcommand{\ethm}{\end{thm}}
\newcommand{\bpf}{\begin{IEEEproof}}
\newcommand{\epf}{\end{IEEEproof}}

\newcommand{\comment}[1]{}

\newcommand{\cm}[1]{\Comment{#1}}

\newcommand{\algorithmfootnote}[2][\footnotesize]{%
  \let\old@algocf@finish\@algocf@finish
  \def\@algocf@finish{\old@algocf@finish
    \leavevmode\rlap{\begin{minipage}{\linewidth}
    #1#2
    \end{minipage}}%
  }%
}

%% file: noisy.tex
In this section, we consider the noisy compressive phase retrieval problem.
The noisy compressive phase retrieval problem is to recover a $K$-sparse complex signal $\vect{x}$,
from a set of quadratic measurements 
$$
y_i=\ABS{\vect{a}_i\HET\vect{x}}^2+w_i,~~i\in[m],
$$
where $\vect{a}_i\HET\in\CMP^n$ are rows of the measurement matrix $\mat{A}\in\CMP^{m\times n}$, $w_i$'s are noise, and $[m]$ denotes the set $\{1,2,\ldots,m\}$. We consider the regime where there exist two constants $\beta$ and $\delta$ such that $K=\beta n^\delta$, $\delta\in(0,1)$. We assume that $w_i$'s are independent, zero-mean, sub-exponential \cite{nonasym} random variables. This model is considered in many phase retrieval literatures \cite{Candes1, Polar, Sastry}. 

We also assume that signal $\vect{x}$ is quantized, which means that the components of $\vect{x}$ lie in a finite set of complex numbers. 
More specifically, let $L_m$ and $L_p$ be the number of possible magnitudes and phases of the non-zero components, respectively. Then, each component of $\vect{x}$ is in the set 
$$\SET=\{u\varepsilon e^{\IMG \frac{2\pi (v-1)}{L_p}} | u\in[L_m], v \in[L_p] \}\cup\{0\}\subset\CMP,$$
where $\varepsilon>0$. Quantized signals can be good approximations of the real world signals and are natural for signal processing with computers \cite{love2004value, candy1974use}. 

We propose two schemes to robustify PhaseCode in the presence of noise: almost-linear scheme and sublinear scheme. The main results of this section are the following theorems.

\begin{theorem}\label{thm_exh_overall}
The almost-linear scheme can recover a random fraction $1-p$, for arbitrarily small $p$, of the non-zero elements of $\vect{x}$ with probability $1-\BIGO(1/K)$, 
with $\Theta(K\log(n))$ measurements. The computational complexity of the algorithm is $\Theta(L_m L_p n\log(n))$.
\end{theorem}

\begin{theorem}\label{thm_fast_overall}
The sublinear scheme can recover a random fraction $1-p$, for arbitrarily small $p$, of the non-zero elements of $\vect{x}$ with probability $1-\BIGO(1/K)$,
with $\Theta(K\log^3(n))$ measurements. The computational complexity of the algorithm is $\Theta(L_m L_p K\log^3(n))$.
\end{theorem}
See the proofs of Theorems \ref{thm_exh_overall} and \ref{thm_fast_overall} in Appendix \ref{prf_exh_overall} and \ref{prf_thm_fast}. Details of the measurement design and the decoding algorithm are shown in the following subsections.

\subsection{Almost-linear Scheme}\label{sec:exh_search}

The idea of the almost-linear scheme is to encode the columns as different patterns. With the number of measurements in each right node being $\Theta(\log(n))$, the patterns are guaranteed to be different enough, so that we can successfully resolve singletons, mergeable multitons, and resolvable multitons.
\subsubsection{Design of Measurements}

Instead of using the 4-by-$n$ trigonometric modulation matrix, we use a new random matrix $\mat{A}_0=\{a_{ij}\}_{P\times n}$ whose entries are i.i.d. with the following distribution:
\begin{equation}\label{distribution}
 a_{ij}=
\begin{cases}
0, & \text{with probability } 1/2 \\
e^{\IMG \theta_{ij}}, & \text{with probability } 1/2,
\end{cases}
\end{equation}
where $\theta_{ij}$'s are i.i.d. and uniformly distributed in $[0,2\pi)$. We call $\mat{A}_0$ the \emph{test matrix}, and we can show that we need $P=\Theta(\log(n))$ for each right node to achieve successful recovery. 

For the almost-linear algorithm, the measurement matrix of the $l$th right node is
$
\mat{A}_l=\mat{A}_0\DIAG{\vect{h}_l}.
$
Without loss of generality, we omit index $l$, and simply use $\vect{h}$ to denote the coding pattern (the left nodes connected to the right node) of a right node.
Then the measurements of this right node are
\begin{equation}\label{measurement}
y_i=\ABS{\vect{a}_i\HET \DIAG{\vect{h}}\vect{x}}^2+w_i,\ i\in[P],
\end{equation}
where $\vect{a}_i\HET$ is the $i$th row of $\mat{A}_0$, and the noise $w_i\in\REAL$, $i\in[n]$ 
satisfies the properties mentioned earlier.
To simplify notation, we define a linear map $\LNR$ from $\CMP^{n\times n}$ to $\REAL^P$:
\begin{equation}\label{lineardef}
\LNR:\ \mat{Z}\mapsto \{\vect{a}_i\HET \mat{Z} \vect{a}_i\}_{i\in[P]}.
\end{equation}
Now according to (\ref{measurement}), by defining $\vect{z}=\DIAG{\vect{h}}\vect{x}$, we have $\vect{y}=\LNR(\vect{z}\vect{z}\HET)+\vect{w}$, 
where $\vect{y}=\{y_i\}_{i\in[P]}$ and $\vect{w}=\{w_i\}_{i\in[P]}$ are the measurement vector and noise vector, respectively. We call $\vect{z}$ the \emph{true signal} corresponding to this right node. 

\subsubsection{Decoding Algorithm}
As mentioned earlier, the PhaseCode algorithm requires the measurements in each right node to enable three operations: 
detecting singletons, resolving strong doubletons, and detecting resolvable multitons. Using our new measurement system, these operations can be done reliably by a simple guess-and-check method: we guess all possible indices, magnitudes, and relative phases, and use an energy test to decide whether our guess is correct. For any of the three operations, we make a hypothesis on the unknown index, magnitude, and phase of the true signal $\vect{z}$ and construct the corresponding hypothesis signal $\hat{\vect{z}}$. For example, when we do singleton detecting, if our hypothesis is that the right node is a singleton, and that the location index of the active component is 5 with the magnitude being $3\varepsilon$, 
we construct $\hat{\vect{z}}=3\varepsilon \vect{e}_5$, where $\vect{e}_i$ denotes the $i$th vector of the canonical basis.
Similarly, we can resolve strong doubletons. For instance, suppose that we know that a right node is connected to two active components which are located at positions $2$ and $5$, respectively, and we also know the magnitudes of the two components are $2\varepsilon$ and $3\varepsilon$, respectively. Then, if we can make a hypothesis that the relative phase is $\frac{\pi}{4}$, we can construct $\hat{\vect{z}}=2\varepsilon\vect{e}_2+3\varepsilon e^{\IMG\frac{\pi}{4}}\vect{e}_5$. Then, we need to check whether our hypothesis is correct. To do this, we perform an $\ell_1$ norm energy test shown in (\ref{energy_test}):
\begin{equation}\label{energy_test}
\begin{aligned}
\hat{\vect{z}}\sim\vect{z}&,\text{ if } \frac{1}{P}\ONEN{\vect{y}-\LNR(\hat{\vect{z}}\hat{\vect{z}}\HET)}<t_0, \\
\hat{\vect{z}}\nsim\vect{z}&,\text{ otherwise},  
\end{aligned}
\end{equation}
where $\hat{\vect{z}}\sim\vect{z}$ means $\hat{\vect{z}}$ and $\vect{z}$ are equal up to a global phase, and $t_0$ is the threshold. The intuitive reason why we do this test is that when $\hat{\vect{z}}\sim\vect{z}$, $\LNR(\hat{\vect{z}}\hat{\vect{z}}\HET)=\LNR(\vect{z}\vect{z}\HET)$, 
then $\vect{y}-\LNR(\hat{\vect{z}}\hat{\vect{z}}\HET)=\vect{w}$, whose energy should be small. 
Conversely, when $\hat{\vect{z}}\nsim\vect{z}$, the energy of $\vect{y}-\LNR(\hat{\vect{z}}\hat{\vect{z}}\HET)$ should be large. Here, we give a result on the error probability of the energy test.
\begin{lemma}\label{lem:energy_test}
When $P=\Theta(\log(n))$ and $\varepsilon$ is appropriately large, with proper threshold $t_0$, the error probability of the energy test shown in (\ref{energy_test}) is $\BIGO(1/n^2)$.
\end{lemma}
The proof of this lemma follows the similar idea which appears in Lemma 14 in \cite{chen2014convex}. We can also show that we need to perform $\Theta(n)$ energy tests before the algorithm stops. Then, using Lemma \ref{lem:energy_test} and some basic principles in probability theory, we can show that the failure probability of the almost-linear scheme is $\BIGO(1/K)$. As for the sample and computational complexity, since we have $\Theta(\log(n))$ measurements for each right node and $\Theta(K)$ right nodes, the sample complexity of the almost-linear scheme would be $\Theta(K\log(n))$;
and since the computational cost of each test is $\Theta(L_m L_p \log(n))$ and there are $\Theta(n)$ tests, the computational complexity of the almost-linear scheme is $\Theta(L_m L_p n\log(n))$.

\subsection{Sublinear Scheme}\label{sec:fast_search}
Although the $\BIGO(n\log(n))$ computational complexity of almost-linear scheme is compelling, we can further improve the computational complexity. Recall that in the noiseless scenario, we get the location index of the active component in a singleton and the non-recovered active component in resolvable multitons by only decoding the measurements of a recoverable right node. Based on this idea, we propose the sublinear scheme for the noisy scenario, which can achieve much lower computational cost compared to the almost-linear scheme, at the cost of slightly larger sample complexity.

\subsubsection{Design of Measurements}

In the sublinear scheme, the measurement matrix for each right node is designed to be a concatenation of the test matrix $\mat{A}_0$ defined in the almost-linear scheme and $R$ \emph{index matrices} $\mat{F}_1,\ldots,\mat{F}_R$. The test matrix $\mat{A}_0$ is still used to perform the energy tests and the index matrices are used to find the location indices. 

Now we show how to design the index matrices. The main idea is to encode each column as a binary code such that we can directly decode the column index of the component to get recovered from the measurements. A similar idea is also used in the Chaining Pursuit method\cite{chainpursuit}. First, we define a deterministic matrix $\mat{B}=\{b_{ij}\}\in\{0,1\}^{R\times n}$, where $ R=\lceil \log n \rceil$, and the $i$th column of $\mat{B}$ is the binary representation of the integer $i-1$. For example, when $n=4$, we have, 
$$
\mat{B} = \left[
\begin{array}{cccc}
0 & 0 & 1 & 1 \\
0 & 1 & 0 & 1 
\end{array}
\right].
$$
We use $\vect{b}_i$ and $\vect{B}_j$ to denote the $i$th row and $j$th column of $\mat{B}$, respectively. Let $\mat{F}_0\in\CMP^{Q\times n}$ be a random matrix whose elements are i.i.d. and uniformly distributed on the unit circle, and $\mat{F}=\mat{F}_0\otimes\mat{B}\in\CMP^{RQ\times n}$.
This means we have $\mat{F}=[ \mat{F}_1\HET \  \mat{F}_2\HET \ \cdots \mat{F}_R\HET  ]\HET$, where $\mat{F}_i=\mat{F}_0\DIAG{\vect{b}_i}\in\CMP^{Q\times n}$.
By concatenating with the test matrix, the measurement matrix of the $l$th right node is $\mat{A}_l=[\mat{A}_0\HET\ \mat{F}\HET]\HET\DIAG{\vect{h}_l}\in\CMP^{(P+QR)\times n}$. Here, we give a simple example of $\mat{A}_l$. Let $n=4$ and thus $R=2$. We have
\begin{equation}\label{ex:matrix_fast}
\mat{A}_l = 
\left[
\begin{array}{cccc}
\vect{A}_{0,1} & \vect{A}_{0,2} & \vect{A}_{0,3} & \vect{A}_{0,4} \\ \hline
\vect{0}  & \vect{0} & \vect{F}_{0,3} & \vect{F}_{0,4} \\
\vect{0} & \vect{F}_{0,2} & \vect{0} & \vect{F}_{0,4}
\end{array}
\right]\DIAG{\vect{h}_l},
\end{equation}
where $\vect{A}_{0,i}$'s and $\vect{F}_{0,i}$'s are the columns of $\mat{A}_0$ and $\mat{F}_0$. We can show that we need $Q=\Theta(\log^2(n))$ to reliably find the correct location index and we also need $P=\Theta(\log(n))$ to perform energy tests.

Consequently, there are $R+1$ sets of measurements. The first set $\vect{y}_0=\{y_{0,i}\}_{i\in[P]}$ is the same as the measurements in almost-linear scheme and is called the \emph{test measurements}:
$$
y_{0,i}=\ABS{\vect{a}_i\HET\vect{z}}^2+w_{0,i},\ i\in[P],
$$
where $\vect{z}=\DIAG{\vect{h}}\vect{x}$ and is still called the true signal. The other $R$ sets $\vect{y}_j=\{y_{j,i}\}_{i\in[Q]}$, $j\in[R]$ correspond to the index matrices and are called the \emph{index measurements}. Each set is composed of $Q$ measurements:
$$
y_{j,i}=\ABS{\vect{f}_{j,i}\HET \vect{z}}^2+w_{j,i},\ i\in[Q],\ j\in[R],
$$
where $\vect{f}_{j,i}\HET$ is the $i$th row of $\mat{F}_j$. We also let $\vect{w}_j$'s be the noise vectors, $j\in\{0\}\cup[R]$.



\subsubsection{Decoding Algorithm}

The sublinear scheme can find the location index by only looking at the measurements. For example, assume that a right node with measurement matrix in (\ref{ex:matrix_fast}) is a singleton whose non-zero component is at position 2. Then, the decoder can see that the elements of the first set of index measurements $\vect{y}_1$ have small absolute value since these measurements only contain noise. Now the decoder knows that the non-zero element should be in the first half of the signal. Then, the decoder observes that the elements in $\vect{y}_2$ have large energy. The decoder knows that if the right node is indeed a singleton, the only possible index of the non-zero component would be 2. Actually this procedure is a binary search on all the $n$ indices of the signal. After this indexing process, the decoder can use the same procedure as the almost-linear scheme to construct a signal $\hat{\vect{z}}$ as the hypothesis of the true signal of this right node, and then use the testing measurements to perform the same energy test.

Now we formally show the details of the fast index search. Assume that $\ABS{\SUPP{\vect{z}}}=T$, and there are $T_s$ non-recovered active components connected to the right node. More specifically,
$\vect{z}=\vect{z}_c+\vect{z}_s$, $\ABS{\SUPP{\vect{z}_s}}=T_s$, $\SUPP{\vect{z}_c}\cap\SUPP{\vect{z}_s}=\emptyset$,
and we know a vector $\hat{\vect{z}}_c\sim\vect{z}_c$. Note that when $T=T_s=1$, we have $\hat{\vect{z}}_c=\vect{z}_c=0$.
Our goal is to find the index $l_s$ of the non-zero element in $\vect{z}_s$ when $T_s=1$ and $\SUPP{\vect{z}_s}=\{l_s\}$.
When $T=1$ and $T>1$, we are looking for non-zero component in a singleton and non-recovered non-zero component in a resolvable multiton, respectively. We subtract the measurements contributed by the signal components which are already known as follows. Let $\hat{y}_{j,i}=\ABSL{\vect{f}_{j,i}\HET \hat{\vect{z}}_c}^2$; then, $\tilde{y}_{j,i}=y_{j,i}-\hat{y}_{j,i}$. We perform the following index tests for $j\in[R]$ with threshold $t_1>0$ to get $l_s$:
\begin{equation}\label{fast_index_test}
\begin{aligned}
\tilde{b}_j=0&,\text{ if } \ABS{\frac{1}{Q}\sum_{i=1}^Q\tilde{y}_{j,i}} <t_1, \\
\tilde{b}_j=1&,\text{ otherwise}. 
\end{aligned}
\end{equation}
The index tests output a binary string $\tilde{\vect{b}}=\{\tilde{b}_j\}_{j\in[R]}$. Note that if $T_s>1$, we still get an output after the index tests, 
but the energy test with the test measurements prevents us from making mistakes. 
Lemma \ref{lem:fast_one_section} states that with high probability $\tilde{b}_j=b_{jl_s}$.
\begin{lemma}\label{lem:fast_one_section}
When $Q=\Theta(\log^2(n))$, with proper threshold $t_1$, if $\SUPP{\vect{x}_s}=\{l_s\}$, then $\mathbb{P}\{\tilde{b}_j\neq b_{jl_s}\}=\BIGO(1/K^3)$.
\end{lemma}
Similar to the almost-linear scheme, using Lemma \ref{lem:fast_one_section}, we can prove that the failure probability of the sublinear scheme is $\BIGO(1/K)$. Since the total number of measurements per each right node is $P+RQ=\Theta(\log^3(n))$, the sample complexity of the sublinear scheme is $\Theta(K\log^3(n))$. In terms of the computational complexity, since there are $\Theta(K)$ right nodes and a constant number of iterations, the computational complexity of the sublinear algorithm is $\Theta(L_m L_p K\log^3(n))$.
\subsection{Simulation Results}\label{sec:experiment}
In this subsection, we show simulation results for the noisy case that validate our theoretical results. The simulations are conducted in Python. Since the sublinear scheme has much lower computational complexity than the almost-linear scheme, we only conduct simulations on the sublinear scheme here. We define the signal-to-noise ratio (SNR):
$$
\text{SNR}=10\log_{10}{\frac{\sum_{j=0}^R\TWON{\vect{y}_j-\vect{w}_j}^2}{\sum_{j=0}^R\TWON{\vect{w}_j}^2}},
$$
and use Gaussian noise. Since the fraction of non-recovered non-zero components $p$ can be made arbitrarily small, in the simulations, we simply define a successful recovery as the cases when \emph{all} the non-zero components are correctly recovered up to a global phase. In all the simulations, we set $P=5\log(n)$, $d=15$, $M=8K$, and $\varepsilon=1$.
\begin{figure}[h]
\centering
\vspace{-0.15in}
\includegraphics[width=0.5\textwidth]{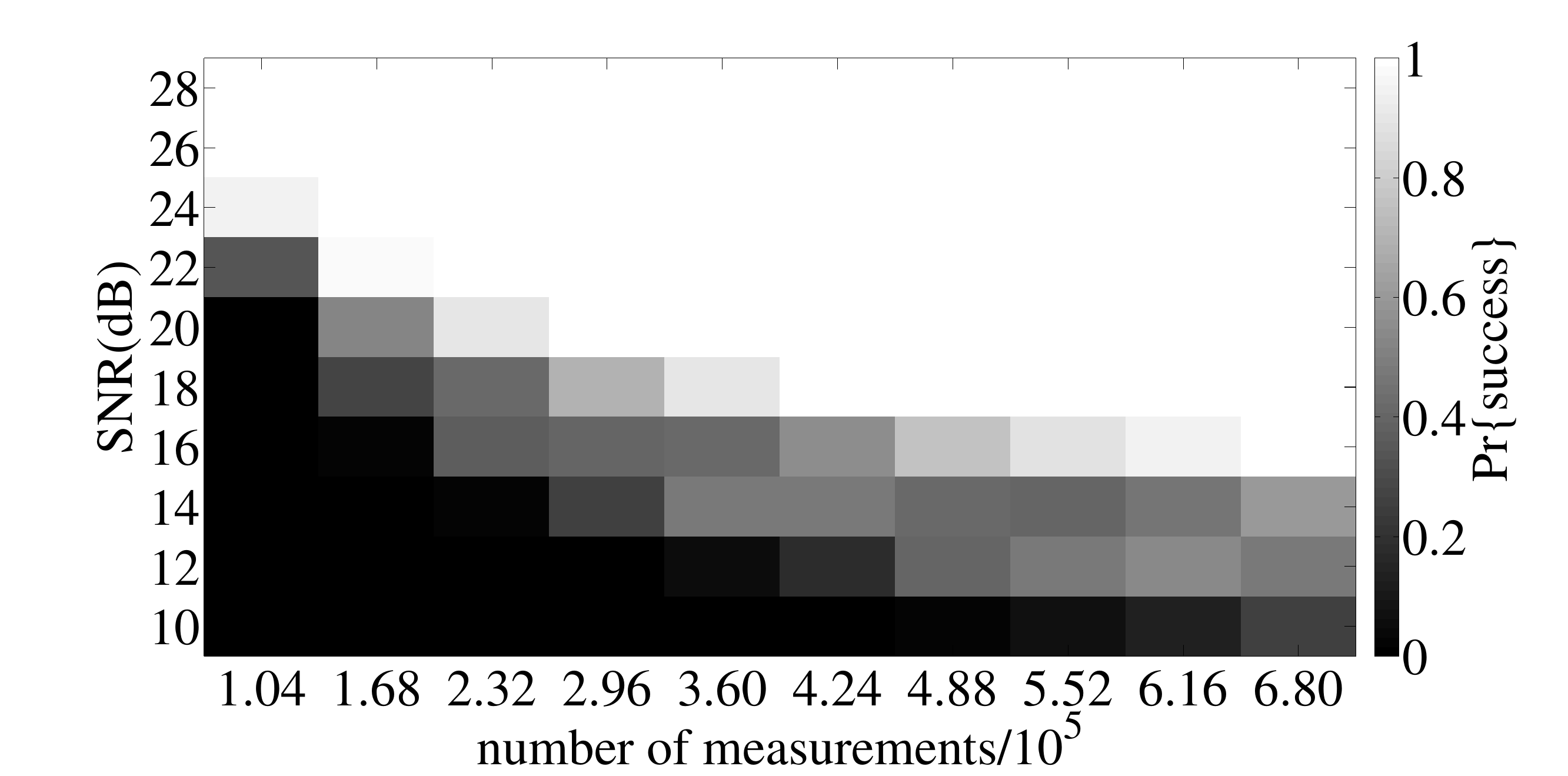}
\caption{{\bf Probability of successful recovery.} We choose $n=2^{20}$, $K=50$, $L_m=3$, and $L_p=6$. Different values of SNR are tested, and for each set of parameters, 1000 experiments are conducted.}
\vspace{-0.2in}
\label{fig_sim_snr}
\end{figure}
\begin{figure}[h]
\centering
\includegraphics[width=0.5\textwidth]{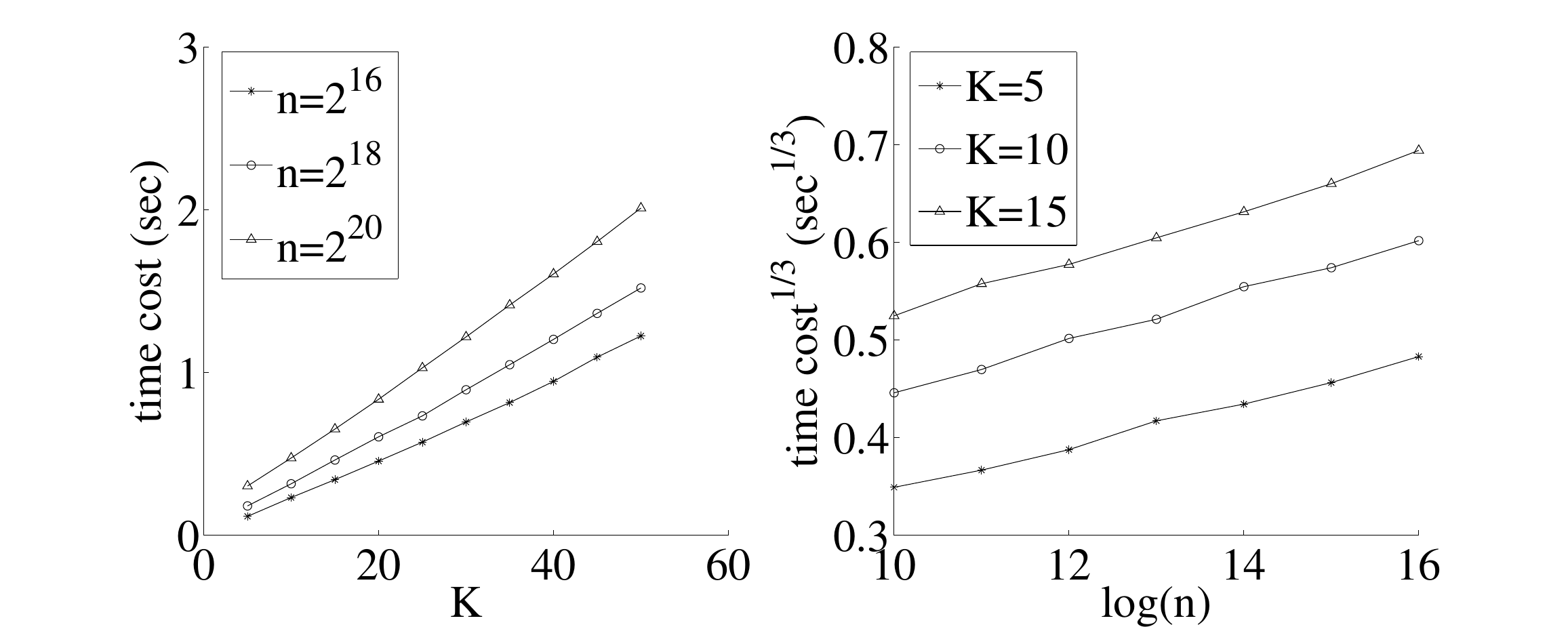}
\caption{{\bf Decoding complexity.} We choose $Q=2\log^2(n)$, $\text{SNR}=20\text{dB}$, $L_m=3$, and $L_p=6$. Different values of $n$ and $K$ are tested, and for each set of parameters, 100 experiments are conducted and the average time cost is shown.}
\label{fig_sim_time}
\end{figure}
\begin{figure}[!h]
\centering
\includegraphics[width=0.5\textwidth]{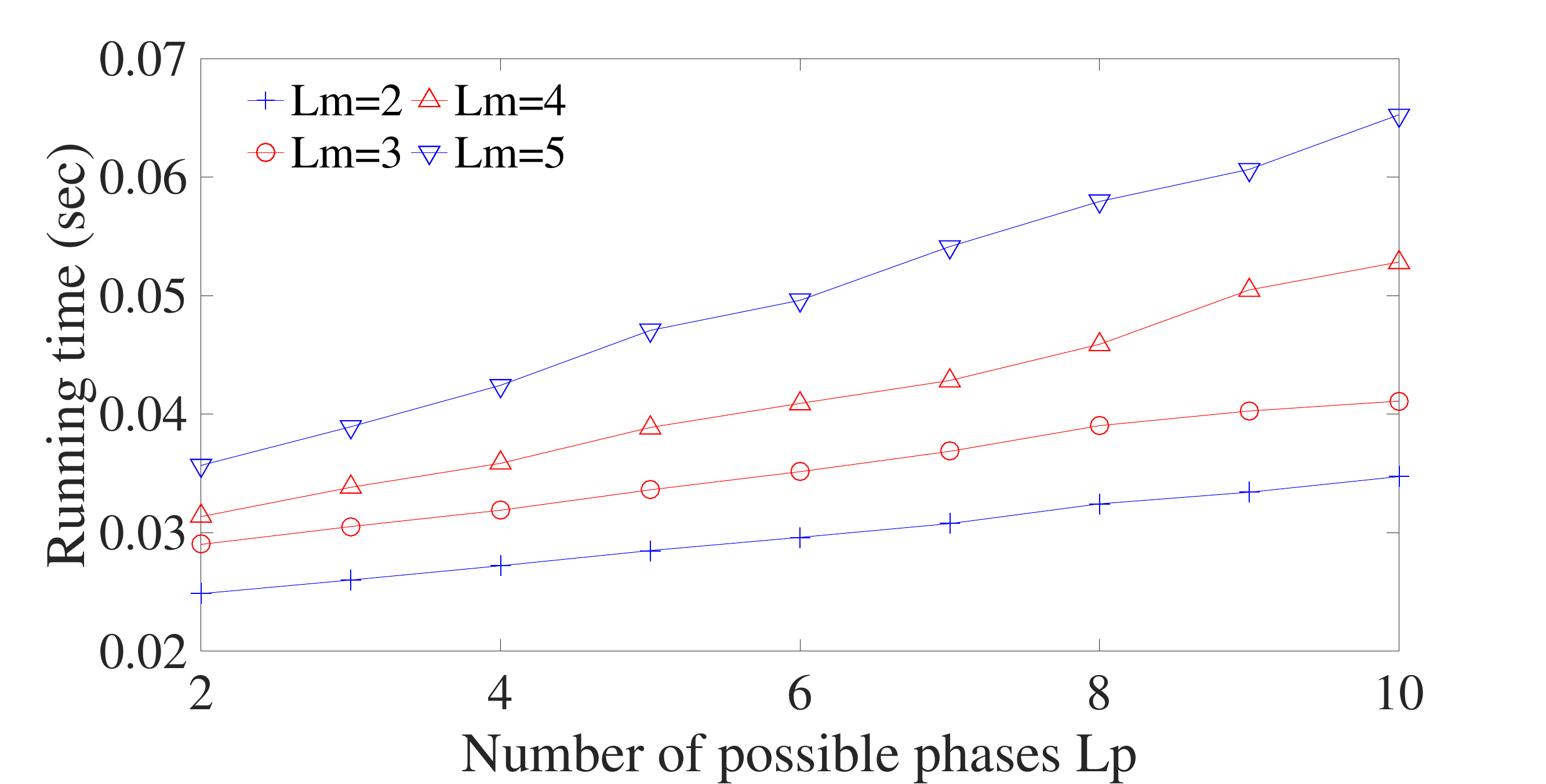}
\caption{{\bf Decoding complexity vs. number of possible magnitudes and phases.} We choose $n=4096$, $K=10$, $Q=5\log^2(n)$ and $\text{SNR}=24\text{dB}$. Different values of $L_m$ and $L_p$ are tested, and for each set of parameters, 100 experiments are conducted and the average time cost is shown.}
\vspace{-0.2in}
\label{fig_lmlp}
\end{figure}

In Figure \ref{fig_sim_snr}, we show the simulation results on the probability of successful recovery as a function of the number of measurements and the SNR.
In Figure \ref{fig_sim_time}, we show the simulation results on the decoding complexity of the sublinear scheme.\footnote{The simulations are conducted on a laptop with 2.8 GHz Intel Core i7 CPU and 16 GB memory.} It can be seen that the time cost of sublinear scheme is indeed low and only linear in $K$ and $\Theta(\log^3(n))$.
In Figure \ref{fig_lmlp}, we show empirical results on the decoding complexity of the sublinear scheme as a function of the number of possible magnitudes and phases ($L_m$ and $L_p$). One can observe that the time cost grows linearly in $L_m$ and $L_p$.

%% file: conclusion.tex
We have considered the problem of recovering a $K$-sparse complex signal $\vect{x} \in \mathbb{C}^n$ from $m$ intensity measurements of the form $|\mat{A} \vect{x}|$, where $\mat{A} \in \mathbb{C}^{m \times n}$ is the measurement matrix. 
Our main focus was on the case where the measurement vectors are unconstrained and noiseless. We proposed the PhaseCode algorithm that is based on a sparse-graph codes framework. We showed that for any signal $\vect{x} \in \mathbb{C}^n$, using order-optimal sample and decoding complexity of $\Theta(K)$, PhaseCode can provably recover all but an arbitrarily small random fraction of the non-zero signal components with high probability. We also showed that PhaseCode can recover almost all the $K$ non-zero signal components using only slightly more than $4K$ measurements if the support of the non-zero components of signal is uniformly random. To the best of our knowledge, our work is the first capacity-approaching low-complexity compressive phase retrieval algorithm.
We furthermore showed that PhaseCode can be used for practical systems such as optical systems with proper modifications. Finally, we demonstrated how PhaseCode can be robustified in the presence of noise. Via extensive simulation results, we validated the performance of PhaseCode for various settings. 

%% file: appendix.tex
\appendix
\subsection{Guess and Check Strategy for Resolvable Multitons}\label{app:quad}
Recall the equations:

\begin{align}\label{eq1-1}
y_{i,1} &= |a + e^{\bi \om \ell} x_\ell| = |u|,\\ \label{eq2-1}
y_{i,2} &= |b + e^{-\bi \om \ell} x_\ell| = |v|, \\ \label{eq3-1}
y_{i,3} &= |c + 2\cos(\om \ell) x_\ell| = |w|, \\ \label{eq4-1}
y_{i,4} &= |d + e^{\bi \om' \ell} x_\ell|,
\end{align}
where complex numbers $a$, $b$, $c$ and $d$ are known values that depend on the values and locations of the known colored active left nodes.
We want to solve the first $3$ equations \eqref{eq1-1}-\eqref{eq3-1} to find $\ell$ and $x_\ell$, and use \eqref{eq4-1} to check if our guess is correct. Since $e^{\bi \om \ell} + e^{-\bi \om \ell} = 2\cos(\om \ell)$, we know that $u+v = w$. Let $\alpha$ be the angle between complex numbers $u$ and $v$. Then,
\begin{align*}
|u+v|^2 = |u|^2 + |v|^2 + 2 |u| |v| \cos(\alpha).
\end{align*}
Thus, one can find $\alpha$ up to a plus-minus sign as,
\begin{align*}
\alpha &= \cos^{-1}(\frac{|u+v|^2 - |u|^2 - |v|^2}{2 |u| |v|}) \\
&= \cos^{-1}(\frac{y_{i,3}^2 - y_{i,1}^2 - y_{i,2}^2}{2 y_{i,1} y_{i,2}}).
\end{align*}

We find possible $x_\ell$'s for two different signs of $\alpha$. If our guess is true, the check measurement $y_{i,4}$ will determine which solution is the right one. Define a known variable $z$ as
$$z = u/v = \frac{|u|}{|v|}e^{\bi \om \alpha}.$$
Thus,
$$
a + e^{\bi \om \ell}x = z (b + e^{-\bi \om \ell}x),
$$
or 
\begin{align}\label{eq:trig}
x = \frac{z b - a}{e^{\bi \om \ell} - z e^{-\bi \om \ell}}.
\end{align}
Replacing $x$ from \eqref{eq3-1} in \eqref{eq:trig}, we have
\begin{align} \nonumber
y_{i,3} & = |c + 2 \cos(\om \ell) \frac{zb - a}{e^{\bi \om \ell} - z e^{-\bi \om \ell}}| \\ \label{eq5}
& = |c\frac{\cos(\om \ell) (1 - z + \frac{2zb - 2a}{c}) + \bi \sin(\om \ell) (1+z)}{\cos(\om \ell) (1 - z ) + \bi \sin(\om \ell) (1+z)}|.
\end{align}
Define the following known complex variables:
\begin{align*}
k_1 &= 1 - z + \frac{2zb - 2a}{c}; \\
k_2 &= 1 + z; \\
k_3 &= 1 - z; \\
k_4 &= y_{i,3}/|c|.
\end{align*}
Also let $k_1 = k_{1r} + \bi k_{1i}$ and use similar notation for the real and imaginary parts of other variables. Then, one can square \eqref{eq5} to get
\begin{align*}
&(k_{1r}\cos(\om \ell) - k_{2i} \sin(\om \ell))^2 + (k_{1i} \cos (\om \ell) + k_{2r} \sin(\om  \ell))^2 \\
& \qquad = k_4^2[(k_{3r} \cos(\om \ell) - k_{2i} \sin(\om \ell))^2 
\\
& \qquad \qquad \qquad + (k_{3i} \cos(\om \ell) + k_{2r} \sin(\om \ell))^2].
\end{align*}
Now defining appropriate new known real variables $k_5$, $k_6$ and $k_7$, we get an equation of the form
$$
k_5 \cos^2(\om \ell) + k_6 \sin^2(\om \ell) = k_7 \sin(\om \ell) \cos(\om \ell).
$$
Squaring the above equation and using $\sin^2(\om \ell) = 1- \cos^2(\om \ell)$, we get a quadratic equation in $\cos^2(\om \ell)$ that one can easily solve to find at most $2$ possible values for $\ell$. Note that $\cos(\om \ell)$ is positive by construction. Now since there are two possible values of $\alpha$, one can get at most $4$ solutions for $\ell$ and $x_\ell$. Those solutions can be checked by \eqref{eq4-1}. If the guess is true, the probability that the check fails is $0$; thus, one can recover the resolvable multiton with probability $1$.

\subsection{Proof of Corollary \ref{cor:l1}}\label{app:l1}

Let $(|x_{(1)}|, |x_{(2)}|, \ldots, |x_{(K)}|)$ be the magnitudes of the non-zero components that are ordered increasingly. We partition the $K$ components to $g = \lfloor K^{(1+\gamma)/2}\rfloor$ subgroups as follows: 
\begin{align*}
&(|x_{(1)}|, \ldots, |x_{(K/g)}|), (|x_{(K/g + 1)}|, \ldots, |x_{(2K/g)}|), \\
&~~\ldots, (|x_{(K - K/g + 1)}|, \ldots, |x_{(K)}|).
\end{align*}
Let $b_i$ be the largest number in subgroup $i$. By Azuma-Hoeffding's inequality, the probability that more than $(p+\epsilon)K / g$ components are missed in a subgroup is upper bounded by $2e^{-2\epsilon^2 K /g}$. Taking $\epsilon = 1/\log(K)$ and using union bound, we have
\begin{align}\label{eq:cor1}
\| \hat{\vect{x}} - \vect{x} \|_1 \leq (p + 1/\log(K)) (\sum_{i=1}^{g} b_i) K/g,
\end{align}
with probability $\mathcal{O}(ge^{-\frac{2K}{g\log^2(K)}}) $. Further, 
\begin{align}
(\sum_{i=1}^{g} b_i) K/g &\leq (|x_{(1)}| + \sum_{i=1}^{g} b_i) K/g\\
& \leq \| \vect{x} \|_1 + b_{g} K/ g \\ 
&\leq \| \vect{x} \|_1 (1 + \Theta(\frac{K^\gamma}{g})) \\ \label{eq:cor2}
&= \| \vect{x} \|_1 (1 + \Theta(K^{-\frac{1-\gamma}{2}})) .
\end{align}

Gathering \eqref{eq:cor1} and \eqref{eq:cor2}, we conclude that with probability $1 - \mathcal{O}(K^{\frac{1+\gamma}{2}}e^{-\frac{2K^{(1-\gamma)/2}}{\log^2(K)}})$, 
\begin{align}
\| \hat{\vect{x}} - \vect{x} \|_1 &\leq p\| x\|_1 (1 + \Theta(\frac{1}{\log(K)}) + \Theta(K^{-\frac{1-\gamma}{2}}))\\
& = p\| x\|_1 (1 + \Theta(\frac{1}{\log(K)})).
\end{align}

\subsection{Proof of Lemma \ref{lem:giant}}\label{app:gc}
\begin{proof}
We form a graph with nodes that are active left nodes which are in singleton right nodes. We construct edges between these nodes if the corresponding active left nodes are connected to a strong doubleton, and we use an Erdos-Renyi random graph model \cite{ER} to find parameters $d$ and $M$ for which there is a giant component of size linear in $K$ after the second step of the algorithm. The Erdos-Renyi random graph model is characterized by $2$ parameters: $n$, the number of nodes in the graph and $p$ which is the probability that each of the ${n \choose 2}$ possible edges are connected. Note that each edge is connected in the graph with probability $p$ independently from every other edge. There is another variant of Erdos-Renyi random graph model which is parametrized by $(n,M)$, where $M$ is the total number of edges. Then, the graph is chosen uniformly at random from the collection of all graphs with $n$ nodes and $M$ edges. By the law of large numbers, the two models are equivalent for $M = {n \choose 2}p$ as long as $n^2p \to \infty$. It is well known that in an Erdos-Renyi model if $np \to c > 1$, as $n \to \infty$, where $c$ is some constant, then the graph will have a unique giant component of size linear in $n$ \cite{ER}. 

Define $K_s$ to be the random variable representing the number of active left nodes that are connected to singletons. We form an Erdos-Renyi random graph model with parameters $(K_s,p_s)$ or equivalently parameters $(K_s,M_s)$ where $p_s$ is the probability that an edge is connected, and $M_s$ is the total number of edges. Thus, as $K_s$ gets large, $M_s$ approaches ${K_s \choose 2}p_s$. Now we compute the parameters $K_s$ and $p_s$ as follows. The probability of an active left node being connected to a singleton right node is the probability that at least one of its $d$ neighbors is a singleton, that is:
\begin{equation}\label{eq:qs}
q_s = 1 - (1- \rho_1)^d.
\end{equation}
Thus, by the law of large numbers as $K$ gets large, there are $K q_s + o(K)$ distinct active left nodes in singleton right nodes. Let $M = cK$ for some constant $c$. As $K$ gets large, the number of doubleton right nodes approaches $M \frac{\eta^2 e^{-\eta}}{2!} + o(K)$ since the degree of right nodes (on the pruned graph with active left nodes) is Poisson distributed with parameter $\eta = Kd/M = d/c$. However, we want to count only distinct doubleton right nodes. It is easy to see that as $K$ gets large, essentially all but a vanishing fraction of the doubleton right nodes are distinct. To this end, fix a doubleton right node with neighbors $(v_1,v_2)$. The probability that a randomly chosen doubleton right node is connected to $(v_1,v_2)$ is $1/{K \choose 2}$. Since the number of doubleton right nodes is linear in $K$, only a vanishing $\Theta(1/K)$ fraction of them are non-distinct. 

Let $M_s$ be the number of strong doubletons (for which both left nodes are also in other singletons). Thus, $M_s$ is the number of edges in our constructed Erdos-Renyi graph. Consider a random left node $i$. Let $D$ be the event that $i$ is connected to a doubleton right node and $S$ be the event that $i$ is connected to a singleton right node. We compute the following $2$ relevant conditional probabilities:
\begin{align*}
p_1 \triangleq \PP(D|S) &= \frac{\PP(D \cap S)}{\PP(S)} \\
&= \frac{1 - \PP(\bar{S}) - \PP(\bar{D}) + \PP(\bar{S}\cap \bar{D})}{1-\PP(\bar{S})} \\
&= \frac{1 - (1 - \rho_1)^d - (1 - \rho_2)^d + (1 - \rho_1 - \rho_2)^d}{1-(1 - \rho_1)^d}.
\end{align*}
\begin{align*}
p_2 \triangleq \PP(D|\bar{S}) &= 1 - \PP(\bar{D}|\bar{S})\\
&= 1 - \frac{\PP(\bar{S}\cap \bar{D})}{\PP(\bar{S})} \\
&= 1 - \frac{(1 - \rho_1 - \rho_2)^d}{(1 - \rho_1)^d}.
\end{align*}
Now we use Bayes' rule to find that
\begin{align*}
q \triangleq \PP(S|D) &= \frac{\PP(D|S)\PP(S)}{\PP(D|S)\PP(S) + \PP(D|\bar{S})\PP(\bar{S})} \\
&= \frac{p_1 q_s}{p_1 q_s + p_2(1-q_s)}.
\end{align*}
Thus, 
\begin{equation}\label{eq:M'}
M_s = M\frac{\eta^2 e^{-\eta}}{2!}q^2.
\end{equation}
The random graph is constructed with $K_s = K (1 - (1- \rho_1)^d)$ nodes and $M_s$ edges chosen uniformly at random among ${K_s \choose 2}$ possible edges. The probability of a randomly chosen edge being connected is thus:
$$ 
p_s = \frac{M \frac{\eta^2 e^{-\eta}}{2!}q^2}{{K_s \choose 2}}.
$$
 
From the well-known Erdos-Renyi random graph result \cite{ER} (also see \cite{Bollobas}), a linear size giant component exists if $K_s p_s > 1$ with probability $1 - \mathcal{O}(1/K_s)$. More precisely, let $Z$ be the size of the giant component. Then, one has
$$
\PP\left(|\frac{Z}{K_s} - \zeta|< \varepsilon \right) = 1 - \mathcal{O}\left(\frac{1}{\varepsilon^2 K_s}\right),
$$
where $\zeta \in (0,1)$ is the unique solution of $\zeta + e^{- 2\zeta M_s/K_s} = 1$, if $2M_s/K_s > 1$ or equivalently $K_s p_s > 1$ \cite{Bollobas,Jaggi}.
Thus, a linear-size giant component exists if
\begin{align*}
\frac{K q_s M_s}{{Kq_s \choose 2}} > 1.
\end{align*}
We present two concrete examples to complete the proof of the lemma. Let $d = 5$. Replacing $M_s$ and $q_s$ by \eqref{eq:M'} and \eqref{eq:qs}, one can check that the inequality holds if $3.11 \leq c \leq 19.24$ (See Figure \ref{fig:gianteq}). 
Similarly, one can set $d = 8$ and see that the inequality holds if $3.48 \leq c \leq 55.36$.
\end{proof}

\begin{figure}
\centering
    \includegraphics[width= 0.35\textwidth]{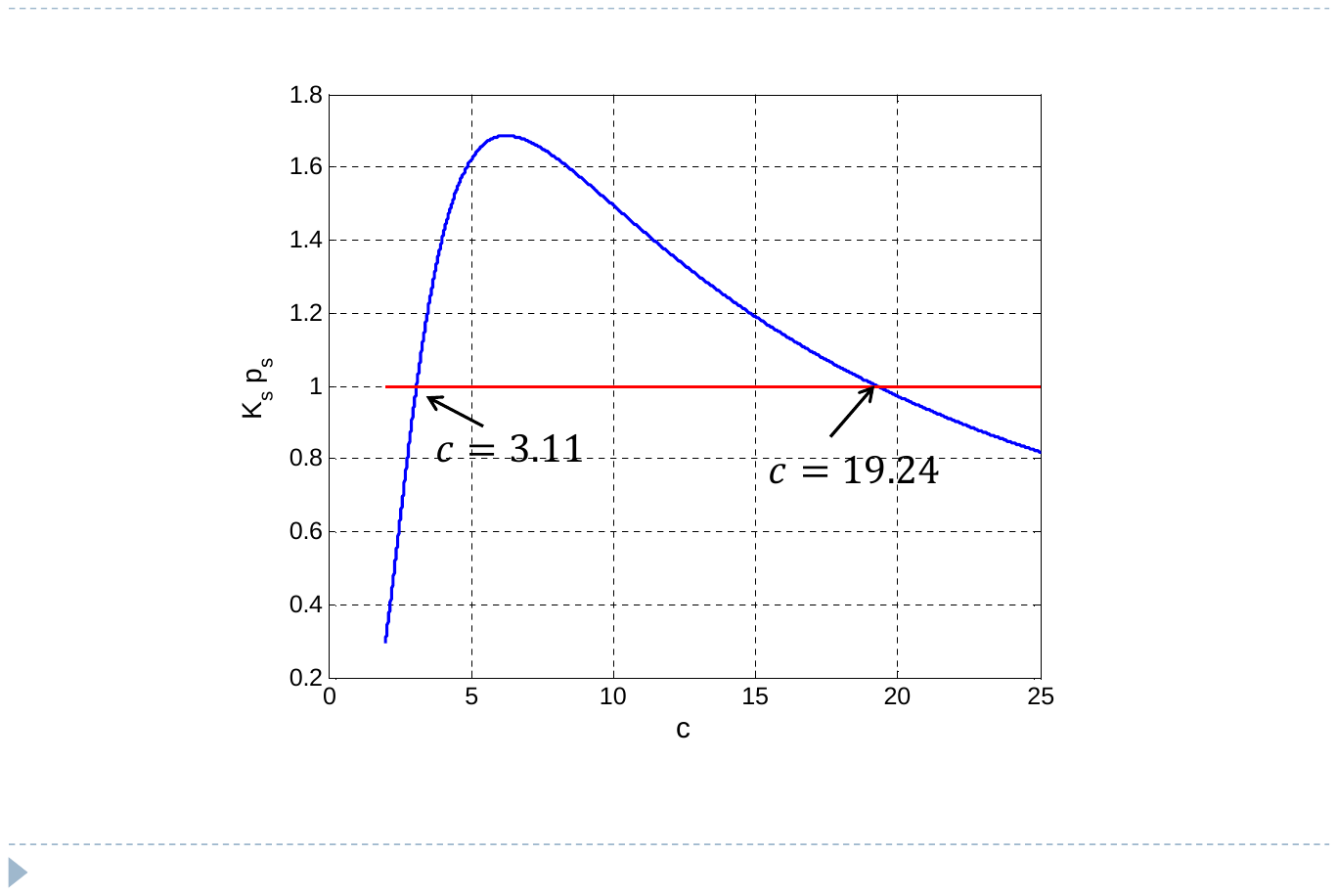}
  \caption{The diagram shows the values of $c$ for which the giant component is formed after step $2$ of the algorithm. Note that $c=M/K$. In the random graph model the giant component is formed if $K_s p_s >1$, where $K_s$ is the number of nodes in the random graph, and $p_s$ is the probability that an edge is connected. From the diagram, one can see that if $3.11 <c < 19.24$, the condition for having a giant component is satisfied. \label{fig:gianteq}} 
\end{figure}

\subsection{Proof of Lemma \ref{lem:unstable}}\label{app:fp}

\begin{figure}
        \centering
        \begin{subfigure}[b]{0.25\textwidth}
                \includegraphics[width=\textwidth]{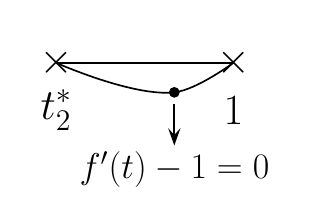}
                \caption{The good case.}
                \label{fig:convex1}
        \end{subfigure}
        \begin{subfigure}[b]{0.25\textwidth}
                \includegraphics[width=\textwidth]{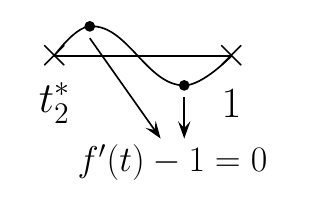}
                \caption{The bad case.}
                \label{fig:convex2}
        \end{subfigure}
        \caption{Figure $(a)$ illustrates the good case that there are no fixed points other than $1$ and $t^*_2$. Figure $(b)$ illustrates the bad case that there is another fixed point in the interval $(t^*_2,1)$. In this case, $f'(t) = 1$ has two solutions for $t \in (t^*_2,1)$, as shown in Figure $(b)$.}\label{fig:convex}
\end{figure}

First, let us consider a small neighborhood around $t^*_1 = 1$. We want 
$$
f(t^*_1-h) < t^*_1 - h = f(t^*_1) -h,
$$
for some small $h > 0$. Equivalently, we want
$$
\frac{f(t^*_1) - f(t^*_1-h)}{h} > 1.
$$
Letting $h \to 0$, the condition becomes $f'(t)|_{t=1} > 1$. This is a necessary and sufficient condition for instability of point $t = 1$. In other words, this condition makes sure that \eqref{eq:decreasing} holds for $p_j$ close to 1. Thus, in picking parameters $d$ and $\eta$, one makes sure that
$$
f'(t)|_{t=1} = (d-1)\eta e^{-\eta} > 1.
$$
For $d = 5$, this leads to $0.3574 < \eta < 2.1533$ or $2.32K < M < 13.99K$. For $d=8$, this leads to $0.17 < \eta < 3.06$ or $2.62K < M < 47.06K$. To complete the proof, we need to show that $f(t) - t < 0$ for $t \in (t^*_2,1)$. Note that $f(t)$ is continuous and continuously differentiable. Thus to show that $f(t) - t < 0$ for $t \in (t^*_2,1)$, it is enough to show that $f'(t) - 1 = 0$ has only one solution in that interval (the ``good'' case: See Figure \ref{fig:convex1}). To see this, suppose that $f(t) - t = 0$ for some $t$ in the interval $(t^*_2,1)$.  Since $1$ and $t^*_2$ are also solutions of $f(t) - t =0$, then $f'(t) - 1$ must change sign at least twice in the interval $(t^*_2,1)$ (the ``bad'' case: See Figure \ref{fig:convex2}).  Therefore, to ensure that $f(t) < t, ~ \forall t \in (t^*_2,1)$  it is sufficient to show that
$$
f'(t) = \eta e^{-\eta t}  (d-1)(1+ e^{-\eta} - e^{-\eta t})^{d-2} = 1,
$$
has only one solution in the interval $t \in (t^*_2,1)$. After some algebra, one can re-write the above equation as
$$
(\eta(d-1))^{-\frac{1}{d-2}}e^{\eta t(1/(d-2)+1)} = e^{\eta t}(1+e^{-\eta}) - 1. 
$$
Replacing $x = e^{\eta t}$, we get an equation of the form $x^a = bx - c$ for $a > 1$ and $b,c > 0$. This equation has clearly at most two solutions for $x \geq 0$. On the other hand, $f'(1) > 1$ and $f'(\infty) = 0$. Thus, $f'(1) = 1$ has a solution for $t > 1$, which shows that $f'(t) = 1$ has at most one solution in $t \in [0,1]$.

\subsection{Probability of Tree-like Neighborhood}\label{app:tree}
In this section, we give a short proof of Lemma \ref{lem:tree}. Let $C_\ell$ be the number of right-nodes and $V_\ell$ be the number of left-nodes in $\mathcal{N}_{\vec{e}}^{2\ell}$. Since the ensemble of the graphs that we consider is only left-regular (and not right-regular), we cannot immediately use the result of \cite{RU01}. Note that the degree distribution of right nodes is Poisson distribution with constant rate. The key idea is to show that the size of the tree is bounded by $\mathcal{O}(\log(K)^\ell)$ with high probability. This is intuitively clear since Poisson distribution has a tail decaying faster than exponential decay. To formally show this, we keep unfolding the tree up to level $\ell^*$, and at each level $\ell$  we upper bound the probability that the size of the tree grows larger than $\mathcal{O}(\log(K)^\ell)$. Fix some constant $c_1$. We upper bound the probability of not having a tree as follows.
\begin{align*}
&\PP (\mathcal{N}_{\vec{e}}^{2\ell^*} \text{~ is not a tree}) \leq \PP(V_{\ell^*} > c_1 \log(K)^{\ell^*}) +\\ 
&\PP(C_{\ell^*} > c_1 \log(K)^{\ell^*}) +\\
& \PP(\mathcal{N}_{\vec{e}}^{2\ell^*} \text{~ is not a tree}|V_{\ell^*} <  c_1 \log(K)^{\ell^*}, ~ C_{\ell^*} <  c_1 \log(K)^{\ell^*}).
\end{align*}

Note that since the left degree is a constant, $d$, if $V_{\ell^*}$ is $\mathcal{O}(\log(K)^{\ell^*})$, $C_{\ell^*}$ is also $\mathcal{O}(\log(K)^{\ell^*})$. Let $\alpha_{\ell} = \PP(V_{\ell} > c_1 \log(K)^{\ell})$. Then, 
\begin{align}\label{eq:alpha1}
\alpha_{\ell} & \leq \alpha_{\ell - 1} + \PP(V_{\ell} > c_1 \log(K)^{\ell}|V_{\ell - 1} < c_1 \log(K)^{\ell - 1}) \\ \label{eq:alpha2}
& \leq \alpha_{\ell - 1} + \PP(V_{\ell} > c_1 \log(K)^{\ell}|C_{\ell} < c_2 \log(K)^{\ell - 1}),
\end{align}
where \eqref{eq:alpha2} is due to the fact that every left node has exactly $d$ edges connected to right nodes so if $V_{\ell - 1} < c_1 \log(K)^{\ell - 1}$, there exists some constant $c_2$ such that $C_{\ell} < c_2 \log(K)^{\ell - 1}$. To count the number of left nodes in depth $\ell$, let $n_\ell < C_{\ell}$ be the number of right nodes exactly at depth $\ell$ after unfolding the tree. Let $X_i, ~1 \leq i \leq n_\ell$ be the degree of these right nodes. Given that $V_{\ell -1} < c_1 \log(K)^{\ell - 1}$, one has $V_{\ell} > c_1 \log(K)^{\ell}$, only if $X = \sum_{i = 1}^{n_\ell} X_i > c_3 \log(K)^{\ell}$ for some constant $c_3$. The distribution of $X$ is Poisson distribution with parameter  $n_{\ell} \lambda$. We know that the tail probability of a Poisson random variable $Y$ with parameter $\lambda$ can be upper bounded as follows: $\PP(Y \geq y) \leq \left( \frac{e \lambda}{y}\right)^y$. Thus, 
$$
\PP(X > c_3 \log(K)^{\ell}) \leq \left(\frac{c_4}{\log(K)}\right)^{c_3 \log(K)^{\ell}} \leq \mathcal{O}(\frac{1}{K}).
$$
Thus, 
\begin{equation}\label{eq:alpha3}
\alpha_{\ell} \leq \alpha_{\ell-1} + \frac{c_5}{K},
\end{equation}
for some constant $c_5$. Now since $\ell^*$ is a constant, summing up the inequalities in \eqref{eq:alpha3}, we show that 
$$\alpha_{\ell^*}  = \PP(V_{\ell^*} > c_1 \log(K)^{\ell^*})\leq \mathcal{O}(\frac{1}{K}).$$ 
Similarly, one can show that 
$$
\PP(C_{\ell^*} > c_1 \log(K)^{\ell^*}) \leq \mathcal{O}(\frac{1}{K}).
$$
To complete the proof, we need to show that with high probability, we have a tree-like neighborhood, given that the number of nodes is bounded by $\mathcal{O}(\log(K)^{\ell^*})$. First, we find a lower bound on the probability that $\mathcal{N}_{\vec{e}}^{2\ell + 1}$ is a tree-like neighborhood if $\mathcal{N}_{\vec{e}}^{2\ell}$ is a tree-like neighborhood, when $\ell < \ell^*$.  Assume that $t$ additional edges have been revealed at this stage without forming a cycle. The probability that the next edge from a left node does not create a cycle is the probability that it is connected to one of the right nodes that is not already in the subgraph which is lower bounded by $1 - \frac{C_{\ell^*}}{m}$. Thus, the probability that $\mathcal{N}_{\vec{e}}^{2\ell + 1}$ is a tree-like neighborhood if $\mathcal{N}_{\vec{e}}^{2\ell}$ is a tree-like neighborhood, is lower-bounded by $(1 - \frac{C_{\ell^*}}{M})^{C_{\ell + 1} - C_\ell}$. Similarly, the probability that $\mathcal{N}_{\vec{e}}^{2\ell + 2}$ is a tree-like neighborhood if $\mathcal{N}_{\vec{e}}^{2\ell + 1}$ is a tree-like neighborhood, is lower-bounded by $(1 - \frac{V_{\ell^*}}{K})^{V_{\ell + 1} - V_\ell}$. Therefore, the probability that $\mathcal{N}_{\vec{e}}^{2\ell^*}$ is a tree-like neighborhood is lower-bounded by
$$
(1 - \frac{V_{\ell^*}}{K})^{V_{\ell^*}}(1 - \frac{C_{\ell^*}}{M})^{C_{\ell^*}}.
$$
For large $M$ and $K$, the above expression is approximately 
$$
e^{-(V^2_{\ell^*}/K + C^2_{\ell^*}/M)} \geq 1 - (V^2_{\ell^*}/K + C^2_{\ell^*}/M).
$$
Now since $V_{\ell^*}$ and $C_{\ell^*}$ are upper-bounded by $\mathcal{O}(\log(K)^{\ell^*})$, the probability of having a tree-like neighborhood is at least $1 - \mathcal{O}(\log(K)^{\ell^*}/K)$. 

\subsection{Convergence to Cycle-free Case}\label{app:mg}
In this section, we give a short proof of Lemma \ref{lem:concentration}. The proof follows similar steps as in \cite{RU01}, with the difference that the right degree is irregular and Poisson-distributed. 

First, we prove \eqref{eq:ctcf}. Let $Z_i = 1_{\{\vec{e}_i ~\text{is colored}\}}, ~ 1 \leq i \leq Kd ~$ be the indicator that $\vec{e}_i$ is colored after $\ell$ iterations of the algorithm. Let $B$ be the event that $\mathcal{N}^{2\ell}_{\vec{e}_1}$ is tree-like. Then,
\begin{align*}
\EE[Z_1] &= \EE[Z_1 | B]\PP(B) + \EE[Z_1|\bar{B}]\PP(\bar{B}) \\
&\leq \EE[Z_1 | B] + \PP(\bar{B}) \\
& \leq p_\ell + \frac{\gamma \log(K)^{\ell}}{K},
\end{align*}
for some constant $\gamma$, where the last inequality is by Lemma \ref{lem:tree}. Trivially, $|\EE[Z_1 | B]| \leq 1$. Furthermore, $\EE[Z] = Kd \EE[Z_1]$. Hence,
$$ Kd(1 - \frac{\gamma \log(K)^{\ell}}{K})<\EE[Z] < Kd(p_\ell + \frac{\gamma \log(K)^{\ell}}{K}).$$ 
Then, \eqref{eq:ctcf} follows from choosing $K$ large enough such that $\frac{K}{\log(K)^{\ell}} > \frac{2 \gamma}{\epsilon}$. 

Second, we prove that
\begin{equation}\label{eq:mg2}
\PP(|Z - Kdp_\ell| > Kd\epsilon/2) < 2e^{-\beta \epsilon^2 K^{1/(2\ell + 1)}}.
\end{equation}
Then, \eqref{eq:mg} follows from \eqref{eq:ctcf} and \eqref{eq:mg2}.
To prove \eqref{eq:mg2}, we use the standard Martingale argument and Azuma's inequality provided in \cite{RU01} with some modifications to account for the right irregular degree. Suppose that we expose the $Kd$ edges of the graph one at a time. Let $Y_i = \EE[Z|e^{i}_1]$. By definition, $Y_0,Y_1,\ldots,Y_{Kd}$ is a Doob's martingale process, where $Y_0 = \EE[Z]$ and $Y_{Kd} = Z$. To use Azuma's inequality, we find the appropriate upper bound: $|Y_{i+1} - Y_i| \leq \alpha_i$. If the right degree is regular and equal to $d_c$, it is shown in \cite{RU01} that $\alpha_i$ can be chosen as $8(d_v d_c)^{\ell}$. We show that when the right degree has Poisson distribution with constant rate, the degree of all of the right nodes can be upper bounded by $\mathcal{O}(K^{\frac{1}{2\ell + 0.5}})$ with probability at least $c_6K(e^{-\beta_1 K^{\frac{1}{2\ell + 0.5}}})$ for some constants $c_6$ and $\beta_1$. To show this, let $X$ be a Poisson random variable with parameter $\lambda$ and $c_7$ be some constant. Then,
\begin{align*}
\PP(X > c_7K^{\frac{1}{2\ell + 0.5}}) &\leq \left( \frac{e\lambda}{c_7K^{\frac{1}{2\ell + 0.5}}} \right)^{c_7K^{\frac{1}{2\ell + 0.5}}}\\
& \leq c_6(e^{-\beta_1 K^{\frac{1}{2\ell + 0.5}}}).
\end{align*}
Now considering $M = \Theta(K)$ right nodes and using union bound, one can see that the probability that all the right nodes have degree less than $\mathcal{O}(K^{\frac{1}{2\ell + 0.5}})$ is at least $1 - \mathcal{O}(K(e^{-\beta_1 K^{\frac{1}{2\ell + 0.5}}}))$. Let $E$ be the event that at least one right node has degree larger than $c_6K(e^{-\beta_1 K^{\frac{1}{2\ell + 0.5}}})$. Given that $E$ has not happened, one can upper bound $\alpha_i^2$ by $\mathcal{O}(K^{\frac{2\ell}{2\ell + 0.5}})$. Then, 
\begin{align*}
\PP(|Z - Kdp_\ell| &> Kd\epsilon/2) \\
&\leq \PP(|Z - Kdp_\ell| > Kd\epsilon/2 | \bar{E}) + \PP(E) \\
&\leq 2e^{-\frac{K^2d^2 \epsilon^2/4}{2\sum_i \alpha_i^2}} + c_6K(e^{-\beta_1 K^{\frac{1}{2\ell + 0.5}}}) \\
& \leq 2e^{-\beta \epsilon^2 K^{1/(4\ell + 1)}}.
\end{align*}

\subsection{Proof of Lemma \ref{lem:de}}\label{app:de2}

First note that it is easy to prove the lemma for specific parameters by plotting the function. See for example Figure \ref{fig:de1}. To formally show it, note that $f(1) = 1$ is one solution of the fixed point equation, since $\lambda(1) = 1$. Also $f(0) = \lambda(e^{-\eta}) > 0$. Thus, by continuity of $f(x)$ and using the assumption that $f'(1) > 1$, there is another fixed point $x^*_2$. Now since $f'(1) > 1$, $f(x) < x$ for $x$ close to 1. In order to show that $f(x) < x$ for all $x \in (x^*_2,1)$, it is enough to show that $f'(x) - 1 = 0 $ has only one solution in $x \in (0,1)$. To this end, see that
$$
f'(x) = \eta e^{-\eta x} \lambda'(1 + e^{-\eta} - e^{-\eta x}).
$$
For ease of notation, let $y = 1 + e^{-\eta} - e^{-\eta x}$ and $y \in (e^{-\eta},1)$. Equivalently, we want to show that 
$$
C(1 + e^{-\eta} - y)(1 + y + y^2 + \ldots + y^{D-2}) =1 
$$
has only one solution where $C = \eta/h(D-1)$. This is easy to see since $D$ is large so $y \simeq  \frac{1-C-Ce^{-\eta}}{1-C}$. 

\subsection{Proof of Lemma \ref{lem:ir}}\label{app:ir}

We show that if 
\begin{align} \label{eq:D}
D=\max\{(\frac{e}{1-\epsilon})^{2/\epsilon},(1+\frac{1}{p^*})^{1/(1-\epsilon)}\},
\end{align}
then, 
\begin{equation}\label{eq:stability}
f'(1)  = \eta e^{-\eta} \sum_{i \geq 1} \lambda_i (i-1) > 1,
\end{equation} 
and the error floor which is approximately $\lambda(e^{-\eta})$ is at most $p^*$; that is,
\begin{equation}\label{eq:error}
\sum_{i \geq 1} \lambda_i e^{-\eta(i-1)} \leq p^*.
\end{equation} 
This shows that in the density evolution equation, $p_j$ converges to $p^*$ as $j$ goes to infinity. This is illustrated in Figure \ref{fig:density1}. 

Recall that 
$$
\bar{d} = (\sum_{i=2}^D \frac{\lambda_i}{i})^{-1} = h(D-1) \frac{D}{D-1}.
$$
Thus, since $M = K/(1-\epsilon)$,
$$
\eta = \frac{K \bar{d}}{M} = h(D-1) \frac{D}{D-1} (1-\epsilon). 
$$ 
First, we show \eqref{eq:error} in the following.

\begin{align*}
\sum_{i=2}^D \lambda_i e^{-\eta(i-1)} &= \frac{1}{h(D-1)} \sum_{i=2}^D \frac{1}{i-1}e^{-\eta(i-1)} \\
& \leq \frac{1}{h(D-1)} \sum_{i=1}^\infty e^{-\eta i} \\
& = \frac{e^{-\eta}}{h(D-1)(1-e^{-\eta})}. 
\end{align*}
It is enough to show that $h(D-1)(e^{\eta} - 1) \geq \frac{1}{p^*}$. We have
\begin{align*}
h(D-1)(e^{\eta} - 1) &\geq e^{\eta} - 1 \\
& \geq e^{\log(D).\frac{D}{D-1}(1-\epsilon)} -1 \\
& \geq D^{1- \epsilon} - 1 \\
& \geq \frac{1}{p^*},
\end{align*}
where the last inequality is due to \eqref{eq:D}.

Second, we show that \eqref{eq:stability} is satisfied in the following.

\begin{align}
\eta e^{-\eta} \sum_{i=2}^D \mu_i (i-1) &= \eta e^{-\eta} \frac{D-1}{h(D-1)} \\
&= D(1-\epsilon)e^{-h(D-1)\frac{D}{D-1}(1-\epsilon)} \\
& \geq D(1-\epsilon)e^{-(1+\log(D))\frac{D}{D-1}(1-\epsilon)} \\
& = \frac{1-\epsilon}{e} D^{\frac{\epsilon D - 1}{D-1}} \\ \label{eq:D1}
& \geq \frac{1-\epsilon}{e} D^{\epsilon/2} \\ \label{eq:D2}
& \geq 1,
\end{align}
where \eqref{eq:D1} is due to \eqref{eq:D} since $D \geq (\frac{e}{1-\epsilon})^{2/\epsilon} \geq \frac{2}{\epsilon}$ implies that $\frac{\epsilon D - 1}{D-1} \geq \frac{\epsilon}{2}$, and \eqref{eq:D2} is due to \eqref{eq:D}. This shows that $p_j, ~j\geq 1$ is a strictly decreasing sequence which is lower bounded by $p^*$. Thus, $p_j \to p^*$ as $j \to \infty$. This completes the proof.


\subsection{Proof of Theorem \ref{thm_exh_overall}}\label{prf_exh_overall}

We first introduce some notation. Here, $\FBNL{\cdot}$ denotes the Frobenius norm of a matrix, 
$\OPNL{\cdot}$ denotes the operator norm of a matrix.
For a sub-exponential random variable, $\SEPL{\cdot}$ denotes the sub-exponential norm of it; 
for a sub-gaussian random variable, $\SGSL{\cdot}$ denotes the sub-gaussian norm of it \cite{nonasym}.
The notations $c$, $c_i$, $C$, and $C_i$ represent absolute constants with positive value. 

In our model, we also assume that the noise $w_i$ satisfies $\EXPL{\ABS{w_i}}=\mu$, $\EXPL{w_i^2}=\sigma^2$, and $\SEPL{w_i} = \nu$.
Since the entries in $\mat{A}_0$ and $\mat{F}_0$ are bounded and thus sub-gaussian, we let $\eta=\SGSL{\ABS{a_{ij}}}$ 
and $\eta_0=\SGSL{\ABS{f_{0,ij}}}$, where $a_{ij}$ and $f_{0,ij}$ are entries of $\mat{A}_0$ and $\mat{F}_0$.

In order to prove Theorem \ref{thm_exh_overall}, we need to prove Lemma \ref{lem:energy_test} first. Here, we restate Lemma \ref{lem:energy_test} with more details.
\renewcommand\thelemma{1}
\begin{lemma}\label{the_lemma_1}
There exists $\zeta>0$, determined by $\eta$, $\nu$, and $\sigma$, such that when $\phi>\mu/\zeta$, for any $t_0\in(\mu, \zeta\phi)$,
\begin{equation}\label{exh_test1}
\PRO{ \frac{1}{P}\ONEN{\vect{w}}\ge t_0}
= \BIGO(1/n^2),
\end{equation}
and
\begin{equation}\label{exh_test2}
\PRO{ \frac{1}{P}\ONEN{\vect{y}-\LNR(\hat{\vect{z}}\hat{\vect{z}}\HET)}< t_0}
= \BIGO(1/n^2),
\end{equation}
when $\hat{\vect{z}}\nsim\vect{z}$.
\end{lemma}

See the proof of Lemma \ref{lem:energy_test} in Appendix \ref{prf_lem_energy}. Now we can analyze the failure probability of the almost-linear scheme. Recall that the bipartite graph is $d$-left-regular; thus, there are $dn$ edges in the graph. In the first iteration, we need to check every edge and detect the singletons. Therefore, we need to do $\Theta(n)$ tests in the first iteration. Similarly, in the following iterations, we need to do at most $\Theta(n)$ tests. Since the number of iterations is a constant, we need to do $N_t=\Theta(n)$ tests.
Lemma \ref{lem:energy_test} tells us that, for any energy test, if no error has been made in the previous tests, the error probability of the energy test is
$\BIGO(1/n^2)$. More specifically, let $E_i$ be the event that there is an error in the $i$th test, while the tests $1,\ldots,i-1$ are all correct. The event $E_{\text{test}}$ that there exists an error in at least one energy test can be decomposed as
$$
E_{\text{test}}=\bigcup_{i=1}^{N_t}{E_i}.
$$
By union bound, we have
$$
\PRO{E_{\text{test}}}\le N_t\sum_{i=1}^{N_t}\PRO{E_i}=\Theta(n)\BIGO(1/n^2)=\BIGO(1/n).
$$
Another possibility of making an error lies in the  coloring algorithm itself. When there is no error in energy tests, this probability is $\BIGO(1/K)$ as analyzed in the noiseless case. Therefore the failure probability of the almost-linear scheme is
\begin{align}
\PRO{E_a}&=\PRO{E_a|E_{\text{test}}} \PRO{E_{\text{test}}} + \PROL{E_a|E_{\text{test}}^\complement} \PROL{E_{\text{test}}^\complement} \nonumber\\
&\le \PRO{E_{\text{test}}}+\PROL{E_a|E_{\text{test}}^\complement} \nonumber\\
&= \PRO{E_{\text{test}}}+\PRO{E_{\text{coloring}}} \nonumber\\
&=\BIGO(1/n)+\BIGO(1/K) \nonumber\\
&=\BIGO(1/K) \nonumber
\end{align}
The sample and computational complexity are already analyzed in Section \ref{sec:exh_search}. This completes the proof of Theorem \ref{thm_exh_overall}.

\subsection{Proof of Lemma \ref{the_lemma_1}}\label{prf_lem_energy}

To prove Equation (\ref{exh_test1}), we use the Bernstein's inequality in \cite{nonasym} as follows. For any $t>0$,
\begin{align*}
\PRO{\frac{1}{P}\sum_{i=1}^P(\ABS{w_i}-\EXP{\ABS{w_i}})>t}	 
\le \exp\left[-C_1P\min\left\{ \frac{t^2}{\nu^2},\frac{t}{\nu} \right\} \right]. 
\end{align*}
Therefore, by choosing $t_0>\EXP{\ABS{w_i}}=\mu$ and $t=t_0-\mu$, we have
$$
\PRO{\frac{1}{P}\ONEN{\vect{w}}\ge t_0}\le \exp\left[ -\delta_1 P \right].
$$
Since $\delta_1$ is a constant and $P=\Theta(\log(n))$, (\ref{exh_test1}) is proved.

Now we prove Equation (\ref{exh_test2}). Before getting into the details of the proof, we give the definition of a new notation $\phi$. For two vectors $\vect{p}, \vect{q}\in\SET^n$, it is easy to see that $\vect{p}\nsim\vect{q}\Leftrightarrow\vect{p}\vect{p}\HET-\vect{q}\vect{q}\HET\neq0$.
Since the entries of $\vect{p}$ and $\vect{q}$ lie in the quantized set $\SET$,
we know that there exists $\phi>0$, such that $\FBNL{\vect{p}\vect{p}\HET-\vect{q}\vect{q}\HET}>\phi$,
when $\vect{p}\nsim\vect{q}$, where $\phi$ depends on $\varepsilon$, $L_m$, and $L_p$.

\renewcommand\thelemma{3}
\begin{lemma}\label{RIP}
Given two vectors $\vect{x}_1,\vect{x}_2\in\CMP^N$, let $\mat{X}=\vect{x}_1\vect{x}_1\HET-\vect{x}_2\vect{x}_2\HET\neq0$. $\LNR$ is the linear function defined in (\ref{lineardef}), and $\vect{w}$ is the noise. Then, for any $s>0$, we have,
\begin{align*}
&\PRO{ \frac{1}{P}\ONEN{\LNR(\mat{X})+\vect{w}}< (\zeta-s\eta_d)\FBN{\mat{X}}-2s\nu }\\
&\le \exp\left[ -C_0P\min{\{s^2,s\}} \right],
\end{align*}
where $\zeta>0$ depends on $\eta$, $\sigma$, and $\nu$, $\eta_d>0$ only depends on $\eta$.
\end{lemma}
See the proof of Lemma $\ref{RIP}$ in Appendix $\ref{prf_RIP}$. Note that $\vect{y}-\LNR(\hat{\vect{z}}\hat{\vect{z}}\HET)=\LNR(\vect{z}\vect{z}\HET-\hat{\vect{z}}\hat{\vect{z}}\HET)+\vect{w}$, and that $\FBNL{\vect{z}\vect{z}\HET-\hat{\vect{z}}\hat{\vect{z}}\HET}>\phi$. Now using Lemma $\ref{RIP}$, conditioning on $\vect{h}$, we have for any $s>0$,
\begin{align}
&\PRO{ \frac{1}{P}\ONEN{\vect{y}-\LNR(\hat{\vect{z}}\hat{\vect{z}}\HET)}< \zeta\phi-(\eta_d\phi+2\nu)s \ |\ \vect{h} } \nonumber \\
&\le \exp\left[ -C_0P\min{\{s^2,s\}} \right]. \label{rip_s}
\end{align}
Since (\ref{rip_s}) holds for any $\vect{h}$, we know that it also holds without conditioning on $\vect{h}$. If $\zeta\phi>t_0$, we can choose $s=\frac{\zeta\phi-t_0}{\eta_d\phi+2\nu}$. Then
$$
\PRO{ \frac{1}{P}\ONEN{\vect{y}-\LNR(\hat{\vect{z}}\hat{\vect{z}}\HET)}< t_0}\le \exp\left[ -\delta_2P \right].
$$
Since $\delta_2$ is a constant and $P=\Theta(\log(n))$, Equation (\ref{exh_test2}) is proved.

We conclude that there exists $\zeta$, determined by the statistics of noise, such that when $\phi>\mu/\zeta$, for any threshold $t_0\in(\mu, \zeta\phi)$,
the energy test fails with probability $\BIGO(1/n^2)$. This completes the proof of Lemma \ref{the_lemma_1}. 

\subsection{Proof of Lemma \ref{RIP}}\label{prf_RIP}
The proof of Lemma \ref{RIP} is based on similar ideas in \cite{chen2014convex}. 
Let $\vect{\xi}=\LNR(\mat{X})+\vect{w}$, then $\xi_i=\vect{a}_i\HET \mat{X} \vect{a}_i+w_i$. 
By the definition of matrix $\mat{A}$, we know that the Hanson-Wright inequality for complex random variables (shown in Appendix \ref{prf_comp_hw}) holds for $\vect{a}_i\HET \mat{X} \vect{a}_i$. That is, for every $t>0$,
\begin{align}
&\PRO{\ABS{\vect{a}_i\HET \mat{X} \vect{a}_i-\EXP{\vect{a}_i\HET \mat{X} \vect{a}_i}}>t}\\
&\le  6\exp{\left[-c\min\left\{ \frac{t^2}{\eta^4\FBN{\mat{X}}^2}, \frac{t}{\eta^2\OPN{\mat{X}}} \right\}\right]} \nonumber \\
&\le  6\exp{\left[-c\min\left\{ \frac{t}{\eta^2\FBN{\mat{X}}}-\frac{1}{4}, \frac{t}{\eta^2\FBN{\mat{X}} }\right\}\right]} \nonumber\\
&\le  6\exp{\left[c\left(\frac{1}{4}-\frac{t}{\eta^2\FBN{\mat{X}}}\right)\right]}, \nonumber
\end{align}
where the second inequality is due to the fact that $(a-1/2)^2\ge 0$ and $\|\mat{X}\|\le \FBNL{\mat{X}}$. From \cite{nonasym}, we know that $\vect{a}_i\HET \mat{X} \vect{a}_i-\EXP{\vect{a}_i\HET \mat{X} \vect{a}_i}$ is a sub-exponential random variable with sub-exponential norm 
\begin{equation}\label{subexp1}
\SEP{\vect{a}_i\HET \mat{X} \vect{a}_i-\EXP{\vect{a}_i\HET \mat{X} \vect{a}_i}}\le C\eta^2\FBN{\mat{X}}.
\end{equation}
On the other hand, 
$$
\ABS{\EXP{\vect{a}_i\HET \mat{X} \vect{a}_i}}=\frac{1}{2}\ABS{\TWON{\vect{x}_1}^2-\TWON{\vect{x}_2}^2}\le\frac{1}{2}\FBN{\mat{X}}.
$$
Thus, 
\begin{align}
\SEP{\xi_i} =& \SEP{\vect{a}_i\HET \mat{X} \vect{a}_i-\EXP{\vect{a}_i\HET \mat{X} \vect{a}_i} + \EXP{\vect{a}_i\HET \mat{X} \vect{a}_i} + w_i}  \nonumber\\
\le&\SEP{\vect{a}_i\HET \mat{X} \vect{a}_i-\EXP{\vect{a}_i\HET \mat{X} \vect{a}_i}} + \ABS{\EXP{\vect{a}_i\HET \mat{X} \vect{a}_i}} +\nu \nonumber\\
 \le& (C\eta^2+1/2)\FBN{\mat{X}}+\nu, \label{subexp2}
\end{align}
where the first inequality is due to the fact that $\EXP{\vect{a}_i\HET \mat{X} \vect{a}_i}$ is a constant, $\SEPL{\EXP{\vect{a}_i\HET \mat{X} \vect{a}_i}} = \ABSL{\EXP{\vect{a}_i\HET \mat{X} \vect{a}_i}}$, and that $\SEPL{w_i} = \nu$. Then,
\begin{equation}\label{subexp3}
\SEP{\ABS{\xi_i}-\EXP{\ABS{\xi_i}}} \le 2\SEP{\xi_i}\le \eta_d\FBN{\mat{X}}+2\nu,
\end{equation}
where $\eta_d=2C\eta^2+1$. Now by Bernstein's inequality in \cite{nonasym}, for every $t>0$,
\begin{align*}
&\PRO{ \frac{1}{P}\sum_{i=1}^P(\ABS{\xi_i}-\EXP{\ABS{\xi_i}})<-t } \\
&\le \exp\left[-C_0P\min\left\{\frac{t^2}{(\eta_d\FBN{\mat{X}}+2\nu)^2},\frac{t}{\eta_d\FBN{\mat{X}}+2\nu}\right\}\right].
\end{align*}
Let $t=s(\eta_d\FBN{\mat{X}}+2\nu)$. For any $s>0$,
\begin{align}
&\PRO{ \frac{1}{P}\sum_{i=1}^P(\ABS{\xi_i}-\EXP{\ABS{\xi_i}})<-s(\eta_d\FBN{\mat{X}}+2\nu) } \nonumber\\
&\le \exp\left[ -C_0P\min{\{s^2,s\}} \right].	\label{Bernstein}
\end{align}
By Cauchy-Schwartz inequality,  for any $i\in[P]$, we have 
$$
\left( \EXP{\xi_i^2} \right)^2\le \EXP{\ABS{\xi_i}}\EXP{\ABS{\xi_i}^3} \le \EXP{\ABS{\xi_i}}\sqrt{\EXP{\xi_i^2}\EXP{\xi_i^4}},
$$
which implies
\begin{equation}\label{expbound0}
\EXP{\ABS{\xi_i}}\ge\sqrt{\frac{(\EXP{\xi_i^2})^3}{\EXP{\xi_i^4}}}.
\end{equation}
By the definition of sub-exponential norm and the fact that $\eta_d>1$, we have
\begin{align}
\EXP{\xi_i^4}\le(4\SEP{\xi_i})^4 &\le(2\eta_d\FBN{\mat{X}}+4\nu)^4  \nonumber \\
&\le (8\eta_d^2\FBN{\mat{X}}^2+32\nu^2)^2. \label{expbound1}
\end{align}
On the other hand, we have
\begin{align}
\EXP{\xi_i^2}&=\EXP{(\vect{a}_i\HET \mat{X} \vect{a}_i)^2}+\EXP{w_i^2}  \nonumber\\
&=\EXP{(\vect{a}_i\HET \mat{X} \vect{a}_i)\TR{\vect{a}_i\vect{a}_i\HET\mat{X}}}+\sigma^2 \nonumber\\
&=\EXP{\TR{(\vect{a}_i\HET \mat{X} \vect{a}_i)\vect{a}_i\vect{a}_i\HET\mat{X}}}+\sigma^2 \nonumber\\ 
&=\TR{\EXP{(\vect{a}_i\HET \mat{X} \vect{a}_i)\vect{a}_i\vect{a}_i\HET}\mat{X}}+\sigma^2 \nonumber\\ 
&=\frac{1}{4}\TR{(\mat{X}+\TR{\mat{X}}\mat{I})\mat{X}}+\sigma^2 \label{inter1} \\ 
&\ge \frac{1}{4}\FBN{\mat{X}}^2+\sigma^2. \label{expbound2}
\end{align}
Here we give an explanation of (\ref{inter1}). Let $\mat{Y}=(\vect{a}_i\HET \mat{X} \vect{a}_i)\vect{a}_i\vect{a}_i\HET$. Then,
\begin{align}
\EXP{Y_{jk}} 
&= \EXP{\sum_{1\le g,h\le n}a_{ig}^*X_{gh}a_{ih}a_{ij}a_{ik}^*}  \nonumber\\
&= \sum_{g=1}^n\EXP{a_{ig}^*X_{gg}a_{ig}a_{ij}a_{ik}^*} + \sum_{g\neq h}\EXP{a_{ig}^*X_{gh}a_{ih}a_{ij}a_{ik}^*}.   \nonumber
\end{align}
If $j=k$, we have
\begin{align}
\EXPL{Y_{jj}}&=\sum_{g=1}^n X_{gg}\EXPL{\ABSL{a_{ig}}^2\ABSL{a_{ij}}^2} + \sum_{g\neq h} X_{gh}\EXPL{a_{ig}^*a_{ih}\ABSL{a_{ij}}^2}  \nonumber\\
&=\frac{1}{4}(\TR{X}+X_{jj}). \nonumber
\end{align}
If $j \neq k$, we have
\begin{align}
\EXPL{Y_{jk}}&=\sum_{g=1}^n X_{gg}\EXPL{\ABSL{a_{ig}}^2a_{ij}a_{ik}^*} + \sum_{g\neq h} X_{gh}\EXPL{a_{ig}^*a_{ih}a_{ij}a_{ik}^*}  \nonumber\\
&= X_{jk} \EXPL{\ABSL{a_{ij}}^2\ABSL{a_{ik}}^2}  \nonumber\\
&=\frac{1}{4}X_{jk}. \nonumber
\end{align}
Therefore, $\EXPL{\mat{Y}}=\frac{1}{4}(\mat{\mat{X}}+\TR{\mat{X}}\mat{I})$.

By combining (\ref{expbound0}), (\ref{expbound1}), and (\ref{expbound2}), we have
\begin{align}
\EXP{\ABS{\xi_i}}&\ge\sqrt{\left(\frac{\frac{1}{4}\FBN{\mat{X}}^2+\sigma^2}{8\eta_d^2\FBN{\mat{X}}^2+32\nu^2}\right)^2\left(\frac{1}{4}\FBN{\mat{X}}^2+\sigma^2\right)} \nonumber\\
&\ge \zeta\FBN{\mat{X}}, \nonumber
\end{align}
where $\zeta=\frac{1}{2}\min\left\{\frac{1}{32\eta_d^2},\frac{\sigma^2}{32\nu^2}\right\}$ is a constant determined by the distribution of $a_{ij}$ and $w_i$.
Then, by (\ref{Bernstein}), we have
\begin{align*}
&\PRO{ \frac{1}{P}\sum_{i=1}^P\ABS{\xi_i}< \zeta\FBN{\mat{X}}-s(\eta_d\FBN{\mat{X}}+2\nu) }  \\
&\le \exp\left[ -C_0P\min{\{s^2,s\}} \right], 
\end{align*}
which completes the proof.
\subsection{Proof of Theorem \ref{thm_fast_overall}}\label{prf_thm_fast}

To prove Theorem \ref{thm_fast_overall}, we make essential use of Lemma \ref{lem:fast_one_section}. Here, we restate Lemma \ref{lem:fast_one_section}, providing more details.
\renewcommand\thelemma{2}
\begin{lemma}\label{the_lem_2}
If $T_s=1$, $\SUPP{\vect{z}_s}=\{l_s\}$, and threshold $t_1\in(0,\varepsilon^2/2)$, then for any $j\in[R]$,
$$
\PRO{\tilde{b}_j\neq b_{jl_s}}=\BIGO(1/K^3).
$$
\end{lemma}
See the proof of Lemma \ref{the_lem_2} in Appendix \ref{prf_lem_2}. Then, by union bound, $\mathbb{P}\{\tilde{\vect{b}} \neq \vect{B}_{l_s}\}=\BIGO(R/K^3)\le\BIGO(1/K^2)$, since $K=\beta n^\delta$. Thus, we can reliably find $l_s$ from the measurements with probability $1-\BIGO(1/K^2)$. For a right node with $T_s=1$, the probability of error in the index tests and the probability of error in the energy test are $\BIGO(1/K^2)$ and $\BIGO(1/n^2)$, respectively. Therefore, the error probability of the tests for a right node is $\BIGO(1/K^2)$. For a bin with $T_s>1$, only the energy test needs to be considered and its error probability is $\BIGO(1/n^2)$. Then, we know the probability of error in the index and energy tests is $\BIGO(1/K^2)$. Since there are $\Theta(K)$ right nodes and a constant number of iterations, using the same decomposition method as in the proof of Theorem \ref{thm_exh_overall}, the error probability of all the tests is $\BIGO(1/K)$. Similar to the almost-linear scheme, considering the $\BIGO(1/K)$ probability of unsuccessful recovery in the coloring algorithm when there is no error in the index and energy tests, the failure probability of sublinear scheme is $\PROL{E_s}=\BIGO(1/K)$. Since the sample and computational complexity of the algorithm are already analyzed in Section \ref{sec:fast_search}, the proof of Theorem \ref{thm_fast_overall} is now complete.

\subsection{Proof of Lemma \ref{the_lem_2}}\label{prf_lem_2}
First, we define an event $E_h$ such that there are more than $C_3\log K$ active left nodes connected to a right node. As shown in \cite{pawar2013pulse}, we have $\PROL{E_h}=\BIGO(1/K^3)$. Now we condition on the coding pattern $\vect{h}$ such that $E_h^\complement$ happens, and thus $\ABS{\SUPP{\vect{z}}}=T\le C_3\log K$.
Similar to the almost-linear algorithm, we define $R+1$ linear mappings, $\LNR_0, \LNR_1, \ldots, \LNR_R$, where
$$
\LNR_0:\ \mat{Z}\mapsto \{\vect{a}_i\HET \mat{Z} \vect{a}_i\}_{i\in[P]},
$$
$$
\LNR_j:\ \mat{Z}\mapsto \{\vect{f}_{j,i}\HET \mat{Z} \vect{f}_{j,i}\}_{i\in[Q]}, \ {\rm for\ } j\in[R].
$$
Then, $\vect{y}_j=\LNR_j(\vect{z}\vect{z}\HET)+\vect{w}_j$, $j\in\{0\}\cup[R]$.

Define the matrix $\tilde{\mat{Z}}=\{\tilde{Z}_{ij}\}_{N\times N}:=\vect{z}\vect{z}\HET-\tilde{\vect{z}}_c\tilde{\vect{z}}_c\HET=\vect{z}\vect{z}\HET-\vect{z}_c\vect{z}_c\HET$. 
Then, $\tilde{\vect{y}}_j=\LNR_j(\tilde{\mat{Z}})+\vect{w}_j$ and $\tilde{y}_{j,i}=\vect{f}_{j,i}\HET \tilde{\mat{Z}} \vect{f}_{j,i}+w_{j,i}$.
Let $f_{j,i,m}$ be the $m$th element of $\vect{f}_{j,i}$. Since for a fixed $j$, $f_{j,i,m}$'s are independent, using similar argument to the one in Appendix \ref{prf_RIP}, we have 
$$
\SEP{\vect{f}_{j,i}\HET \tilde{\mat{Z}} \vect{f}_{j,i} - \EXP{\vect{f}_{j,i}\HET \tilde{\mat{Z}} \vect{f}_{j,i}}} \le C_2\eta_0^2\FBN{\tilde{\mat{Z}}}.
$$
Thus, $\SEPL{\tilde{y}_{j,i}-\EXPL{\tilde{y}_{j,i}}}\le C_2\eta_0^2\FBNL{\tilde{\mat{Z}}} + \nu$.
Since there are $2T-1$ nonzero entries in $\tilde{\mat{Z}}$, we have $\FBNL{\tilde{\mat{Z}}}\le \sqrt{2T-1}L_m\varepsilon$. 
Moreover, $T\le C_3\log K$, which implies that
$\SEPL{\tilde{y}_{j,i}-\EXPL{\tilde{y}_{j,i}}}\le C_4\eta_0^2L_m\varepsilon \sqrt{\log K}+\nu\le \zeta_0\sqrt{\log K}$, where $\zeta_0$ is determined by $\eta_0$, $L_m$, $\varepsilon$, and $\nu$.

On the other hand, since $T_s=1$ and $\SUPP{\vect{z}_s}=\{l_s\}$, $\tilde{\mat{Z}}$ has only one non-zero element on the diagonal, i.e.,
$\tilde{Z}_{l_sl_s}=\ABSL{z_{l_s}}^2$. Note that $\EXPL{\tilde{y}_{j,i}}=\EXPL{\ABSL{f_{j,i,l_s}}^2}\ABSL{z_{l_s}}^2=b_{jl_s}\ABSL{z_{l_s}}^2$. Thus, by Bernstein's inequality, for every $t\ge0$, 
\begin{align*}
&\PRO{\ABS{\frac{1}{Q}\sum_{i=1}^Q (\tilde{y}_{j,i} - b_{jl_s}\ABS{z_{l_s}}^2 )}>t\ |\ \vect{h}} \\ 
&\le 2\exp \left[ -C_5 Q \min \left\{ \frac{t^2}{\zeta_0^2\log K}, \frac{t}{\zeta_0\sqrt{\log K}} \right\} \right] \\
&\le 2\exp \left[ -\frac{C_5}{\zeta_0^2} \sqrt{Q} \min\left\{ t^2,t \right\} \right],
\end{align*}
where the last inequality is due to the fact that $Q=\Theta(\log^2N)$. We choose $t_1=t<\varepsilon^2/2$. When $b_{jl_s}=0$, we have
\begin{align}
\PRO{\ABS{\frac{1}{Q}\sum_{i=1}^Q \tilde{y}_{j,i}}>t_1\ |\ \vect{h}}  
\le 2\exp \left[ -\frac{C_5}{\zeta_0^2} \sqrt{Q} \min\left\{ t_1^2,t_1 \right\} \right],  \label{fast_test1}
\end{align}
and when $b_{jl_s}=1$, we have
\begin{align}
&\PRO{\ABS{\frac{1}{Q}\sum_{i=1}^Q \tilde{y}_{j,i}}<t_1\ |\ \vect{h}}    \nonumber \\
&\le \PRO{\frac{1}{Q}\sum_{i=1}^Q \tilde{y}_{j,i}<t_1\ |\ \vect{h}}  \nonumber\\
&\le \PRO{\frac{1}{Q}\sum_{i=1}^Q \tilde{y}_{j,i}<\ABS{z_{l_s}}^2-t_1\ |\ \vect{h}}  \label{interm2} \\
&\le \PRO{\ABS{\frac{1}{Q}\sum_{i=1}^Q \tilde{y}_{j,i}-\ABS{z_{l_s}}^2}>t_1\ |\ \vect{h}}    \nonumber\\
&\le 2\exp \left[ -\frac{C_5}{\zeta_0^2} \sqrt{Q} \min\left\{ t_1^2,t_1 \right\} \right],   \label{fast_test2}
\end{align}
where the inequality (\ref{interm2}) is due to the fact that $t_1<\varepsilon^2/2$ and $\ABS{z_{l_s}}^2\ge \varepsilon^2$. Define the error events $E_{\text{index}}=\{\ABSL{\frac{1}{Q}\sum_{i=1}^Q \tilde{y}_{j,i}}>t_1\}$, when $b_{jl_s}=0$, and $E_{\text{index}}=\{\ABSL{\frac{1}{Q}\sum_{i=1}^Q \tilde{y}_{j,i}}<t_1\}$, when $b_{jl_s}=1$.
Then, since $Q=\Theta(\log^2(n))$ and inequalities (\ref{fast_test1}) and (\ref{fast_test2}) hold for any $\vect{h}\in E_h^\complement$, we have,
$$
\PROL{E_{\text{index}}|E_h^\complement}=\BIGO(1/K^3).
$$
Now we know that
\begin{align}
\PROL{E_{\text{index}}}&=\PROL{E_{\text{index}}|E_h^\complement}\PROL{E_h^\complement}+\PROL{E_{\text{index}}|E_h}\PROL{E_h} \nonumber\\
&\le \PROL{E_{\text{index}}|E_h^\complement} + \PROL{E_h} \nonumber\\
& = \BIGO(1/K^3)+\BIGO(1/K^3) \nonumber\\
& = \BIGO(1/K^3), \nonumber
\end{align}
which completes the proof.

\subsection{Hanson-Wright Inequality for Complex Random Variables}\label{prf_comp_hw}
\begin{theorem}\label{thm_comp_hw}
Let $\vect{\gamma}=\{\gamma_i\}_{i\in[n]}\in\CMP^n$ be a random vector with independent entries $\gamma_i$, satisfying 
$\EXP{\gamma_i}=0$, and $\ABS{\gamma_i}$ is sub-gaussian with $\SGS{\ABS{\gamma_i}}\le\eta$ for all $i\in[n]$. 
Let $\mat{U}\in\CMP^{n\times n}$ be a Hermitian matrix. Then, for every $t\ge 0$,
\begin{align*}
&\PRO{\ABS{\vect{\gamma}\HET \mat{U} \vect{\gamma}-\EXP{\vect{\gamma}\HET \mat{U} \vect{\gamma}}}>t} \\
&\le 6\exp{\left[-c_0\min\left\{\frac{t^2}{\eta^4\FBN{\mat{U}}^2}, \frac{t}{\eta^2\OPN{\mat{U}}}\right\}\right]}.
\end{align*}
\end{theorem}
\begin{proof}
Let $\vect{\alpha}=\{\alpha_i\}_{i\in[n]}$ and $\vect{\beta}=\{\beta_i\}_{i\in[n]}$ be the real and imaginary parts of $\vect{\gamma}$.
Then, we know that $\alpha_i$'s and $\beta_i$'s are sub-gaussian random variables with $\SGS{\alpha_i}\le \eta$ and 
$\SGS{\beta_i}\le \eta$ for all $i\in[n]$.
Note that here, although $\gamma_i$'s are independent, the real and imaginary parts of $\gamma_i$ are not necessarily independent for a certain $i$. In other words, for any $i$, $\alpha_i$ and $\beta_i$ may not be independent.

Let $\mat{V}$ and $\mat{W}$ be the real and imaginary parts of $\mat{U}$. 
Since $\mat{U}$ is a Hermitian matrix, we have $\mat{V}=\mat{V}\TSP$ and $\mat{W}=-\mat{W}\TSP$. We also know that 
$\vect{\gamma}\HET \mat{U} \vect{\gamma}$ is a real number. Then, we have
$$
\vect{\gamma}\HET \mat{U} \vect{\gamma} = \vect{\alpha}\TSP \mat{V} \vect{\alpha} - 2\vect{\alpha}\TSP\mat{W}\vect{\beta}+\vect{\beta}\TSP\mat{V}\vect{\beta}.
$$
Therefore, $\PRO{\ABS{\vect{\gamma}\HET \mat{U} \vect{\gamma}-\EXP{\vect{\gamma}\HET \mat{U} \vect{\gamma}}}>t}$ is upper bounded by three terms,
\begin{align}
&\PRO{\ABS{\vect{\gamma}\HET \mat{U} \vect{\gamma}-\EXP{\vect{\gamma}\HET \mat{U} \vect{\gamma}}}>t} \\
&\le \PRO{\ABS{\vect{\alpha}\TSP \mat{V} \vect{\alpha}-\EXP{\vect{\alpha}\TSP \mat{V} \vect{\alpha}}}>t/4}  \nonumber\\ 
&+ \PRO{\ABS{\vect{\alpha}\TSP\mat{W}\vect{\beta}-\EXP{\vect{\alpha}\TSP\mat{W}\vect{\beta}}}>t/4} \nonumber\\
&+ \PRO{\ABS{\vect{\beta}\TSP\mat{V}\vect{\beta}-\EXP{\vect{\beta}\TSP\mat{V}\vect{\beta}}}>t/4}. \label{separating}
\end{align}
Since $\alpha_i$'s are independent and $\EXP{\alpha_i}=0$, according to the Hanson-Wright inequality for real numbers\cite{HansonWright}, we have
\begin{align*}
&\PRO{\ABS{\vect{\alpha}\TSP \mat{V} \vect{\alpha}-\EXP{\vect{\alpha}\TSP \mat{V} \vect{\alpha}}}>t/4} \nonumber \\
&\le 2\exp{\left[-c_1\min\left\{\frac{t^2}{\eta^4\FBN{\mat{V}}^2}, \frac{t}{\eta^2\OPN{\mat{V}}}\right\}\right]}. 
\end{align*}
Further, we have $\FBN{\mat{V}}\le\FBN{\mat{U}}$, $\OPN{\mat{V}}\le\OPN{\mat{U}}$. Therefore,
\begin{align}
&\PRO{\ABS{\vect{\alpha}\TSP \mat{V} \vect{\alpha}-\EXP{\vect{\alpha}\TSP \mat{V} \vect{\alpha}}}>t/4} \nonumber\\
&\le 2\exp{\left[-c_1\min\left\{\frac{t^2}{\eta^4\FBN{\mat{U}}^2}, \frac{t}{\eta^2\OPN{\mat{U}}}\right\}\right]}.  \label{alpha_part}
\end{align}
And similarly,
\begin{align}
&\PRO{\ABS{\vect{\beta}\TSP \mat{V} \vect{\beta}-\EXP{\vect{\beta}\TSP \mat{V} \vect{\beta}}}>t/4}	\nonumber \\
&\le 2\exp{\left[-c_2\min\left\{\frac{t^2}{\eta^4\FBN{\mat{U}}^2}, \frac{t}{\eta^2\OPN{\mat{U}}}\right\}\right]}. \label{beta_part}
\end{align}

Now consider the cross term. Let $W_{ij}$ be the entries of $\mat{W}$. Since $\mat{W}=-\mat{W}\TSP$, $W_{ii}=0$ for all $i\in[n]$. 
Then, $\vect{\alpha}\TSP\mat{W}\vect{\beta}=\sum_{i\neq j}{W_{ij}\alpha_i\beta_j}$, 
and $\EXP{\vect{\alpha}\TSP\mat{W}\vect{\beta}}=0$. Then, we can bound $\PRO{\ABS{\vect{\alpha}\TSP\mat{W}\vect{\beta}}>t/4}$ in the same way as in \cite{HansonWright} so that
\begin{align}
&\PRO{\ABS{\vect{\alpha}\TSP\mat{W}\vect{\beta}}>t/4} \nonumber \\
&\le 2\exp{\left[-c_3\min\left\{ \frac{t^2}{\eta^4\FBN{\mat{U}}^2}, \frac{t}{\eta^2\OPN{\mat{U}}} \right\}\right]}.  \label{cross_answer}
\end{align}
By combining (\ref{alpha_part}), (\ref{beta_part}), and (\ref{cross_answer}), Theorem \ref{thm_comp_hw} is proved.

\end{proof}

\subsection{Pseudocode}\label{app:pseudocode}
In this subsection, we provide the pseudocode of the PhaseCode algorithm. Moreover, we provide the pseudocodes of the right node processors: singleton processor, mergeable multiton processor, and resolvable multiton processor.


\begin{algorithm*}[h!]
\floatname{algorithm}{Pseudocode}
\footnotetext{$^*$One can use Breadth-first search to find the largest component of a graph and its time complexity is $\mathcal{O}(K)$.}
\begin{algorithmic}
\State $\mathcal{I} \gets \emptyset$ \cm{No active component is found in the beginning}
\For{\textbf{each} i in $\{1, 2, ..., M\}$} \cm{Find all singletons}
    \State Singleton Processor
\EndFor

\For{\textbf{each} i in $\{1, 2, ..., M\}$} \cm{Find all doubletons and merge}
    \State Mergeable Multiton Processor
\EndFor

\State $\text{Color}_{0} \gets$ Color of the largest colored component \cm{Find the largest colored component$^*$}

\For{\textbf{each} $\ell$ in $\mathcal{I}$} \cm{Uncolor all other left nodes and delete all values of them}
    \If {$\text{Color}_{\ell} \neq \text{Color}_0$}
        \State $x_{\ell} \gets $ None
        \State $\text{Color}_{\ell} \gets $ None
        \State $\mathcal{I} \gets \mathcal{I} - \{\ell\}$
    \EndIf
\EndFor

\While{$|\mathcal{I}| < K$ and any changes are made in the previous loop} \cm{Keep resolving multitons}
    \State Resolvable Multiton Processor
\EndWhile
\end{algorithmic}
\caption{PhaseCode Algorithm}
\label{alg:unicolor_alg}
\end{algorithm*}


\begin{algorithm*}[h!]
\floatname{algorithm}{Pseudocode}
\begin{algorithmic}
\If {$y_{i,1} = y_{i,2} = y_{i,4}$} \cm{Check whether this right node is a singleton or not}
    \State $\ell \gets \frac{1}{\om}\cos^{-1}(\frac{y_{i,3}}{2y_{i,1}})$ \cm{Find the index of the active left node connected to this right node}
    \State $x_\ell \gets y_{i,1}$ \cm{Assign a value to the active left node}
    \State $\mathcal{I}_0 \gets \mathcal{I}_0 \cup \{\ell\}$ \cm{Declare a new found active left node}
    \State $\text{Color}_\ell \gets $ new color \cm{Color the new active left node with a new color}
\EndIf
\end{algorithmic}
\caption{Singleton Processor}
\label{alg:singleton_processor}
\end{algorithm*}

\begin{algorithm*}[h!]
\floatname{algorithm}{Pseudocode}
\footnotetext{$^{*}$One has to color not only blue active left nodes connected to this right node but all blue active left nodes. This can be done with a special data structure based on linked-lists.}
\begin{algorithmic}
\If {Right node $i$ is connected to no colored active left node or the number of colors connected to the right node is not exactly $2$}
    \State Return \cm{If this right node is not mergeable}
\EndIf
\State Red, Blue $\gets $ Two colors of the active left nodes connected to the right node
\State $\mathcal{R} \gets $ indices of the active left nodes that are colored with Red
\State $\mathcal{B} \gets $ indices of the active left nodes that are colored with Blue
\State $r \gets \sum_{\ell \in \mathcal{R}} x_j e^{\bi \om \ell}$
\State $b \gets \sum_{\ell \in \mathcal{B}} x_\ell e^{\bi \om \ell}$
\For{\textbf{each} $z_1$ in $\{+1,-1\}$} \cm{Consider two candidate}
    \State $\phi \gets z_1 \cos^{-1}\left(\frac{|r|^2 + |b|^2 - y_{i,1}^2}{2 |r| |b|}\right) + \angle r - \angle b$ \cm{Find a candidate for phase offset} 
    \If {$\left| \sum_{\ell \in \mathcal{R}} x_\ell e^{\bi \om' \ell} + \exp(\bi \phi) \times \sum_{\ell \in \mathcal{B}} x_\ell e^{\bi \om' \ell} \right| = y_{i,4}$} \cm{Check the candidate with $y_{i,4}$}
        \State Color Red and Color Blue are combined to a new color
        \For{\textbf{each} $\ell$ in $\mathcal{B}$} \cm{Adjust phase of the components that are colored with Color Blue $^*$}
        \State $x_{\ell} \gets x_{\ell} \times \exp(\bi \phi)$
        \EndFor
        \State Return
    \EndIf 
\EndFor
\end{algorithmic}
\caption{Mergeable Multiton Processor}
\label{alg:mergeable_multiton_processor}
\end{algorithm*}

\begin{algorithm*}[h!]
\floatname{algorithm}{Pseudocode}
\begin{algorithmic}
\If {Right node $i$ is connected to no colored active left node or they are colored with more than $1$ color}     
\State Return\cm{If this right node is not resolvable}
\EndIf
\State $\text{Color} \gets $ Common color of the connected active left nodes
\State $\mathcal{I}' \gets \mathcal{I} \cap \{j| H_{i,j} = 1, 1 \leq j \leq n\}$ \cm{Colored active left nodes connected to this right node}
\State $a \gets \sum_{i \in \mathcal{I}'} { x_i e^{\bi \om \ell} }$
\State $b \gets \sum_{i \in \mathcal{I}'} { x_i e^{-\bi \om \ell} }$
\State $c \gets \sum_{i \in \mathcal{I}'} { 2 \cos(\om \ell) x_i  }$
\State $d \gets \sum_{i \in \mathcal{I}'} { x_i e^{\bi \om' \ell} }$

\For{\textbf{each} $z_1$ in $\{+1,-1\}$} \cm{Consider two signs of $\alpha$}
    \State $\alpha \gets z_1 \cos^{-1}(\frac{y_{i,3}^2 - y_{i,1}^2 - y_{i,2}^2}{2 y_{i,1} y_{i,2}})$
    \State $z \gets \frac{y_{i,1}}{y_{i,2}} \exp(\alpha\bi)$
    \State $k_1 \gets 1 - z + \frac{2(zb - a)}{c}$
    \State $k_2 \gets 1 + z$
    \State $k_3 \gets 1 - z$
    \State $k_4 \gets \frac{y_{i,3}}{|c|}$
    \State $k_5 \gets |k_1|^2 - k_4^2 |k_3|^2$
    \State $k_6 \gets |k_2|^2 - k_4^2 |k_2|^2$
    \State $k_7 \gets 2\rr{k_1}\ii{k_2} - 2\ii{k_1}\rr{k_2} + k_4^2 (2\rr{k_2}\ii{k_3} - 2\rr{k_3}\ii{k_2})$
    \State $k_8 \gets k_6^2 + k_7^2 - 2k_6k_7 + k_8^2$
    \State $k_9 \gets 2k_6k_7 - k_8^2 - 2k_7^2$
    \State $k_{10} \gets k_7^2$
    \For{\textbf{each} $z_2$ in $\{+1,-1\}$} \cm{Consider two solutions of a quadratic equation}
        \If {$k_9^2 - 4k_8k_{10} < 0$}
            \State Continue
        \EndIf    
        \If {$\frac{-k_9 + z_2 \sqrt{k_9^2 - 4k_8k_{10}}}{2k_8} < 0$}
            \State Continue
        \EndIf
        \State $\ell' \gets \cos^{-1} \left[ \sqrt{\frac{-k_9 + z_2 \sqrt{k_9^2 - 4k_8k_{10}}}{2k_8}} \right] / \om$ \cm{Find a candidate of $\ell$}
        \State $x' \gets \frac{z b - a}{e^{\bi \om \ell} - z e^{-\bi \om \ell}}$ \cm{Find a candidate of $x_\ell$}
        
        \If {$y_{i,4} = |d + e^{\bi \om' \ell'} x'|$} \cm{Check the validity of the candidates with $y_{i,4}$}       
            \State $x_{\ell'} \gets x'$ \cm{Assign a value to the component}
            \State $\mathcal{I}_0 \gets \mathcal{I}_0 \cup \{\ell' \}$ \cm{Declare a new found component}
            \State $\text{Color}_\ell' \gets $ Color \cm{Color the new component with the color of the other components connected to the right node}
            \State Return
        \EndIf
    \EndFor
\EndFor
\end{algorithmic}
\caption{Resolvable Multiton Processor}
\label{alg:resolvable_multiton_processor}
\end{algorithm*}